\documentclass[prx,aps,twocolumn,notitlepage,superscriptaddress,showpacs,nofootinbib,10pt]{revtex4-2}

\usepackage[T2A]{fontenc}
\usepackage[russian,english]{babel}
\usepackage{enumerate,appendix}
\usepackage{amsmath,amsthm, amssymb,commath}
\usepackage{graphicx}
\usepackage[colorlinks]{hyperref}
\hypersetup{
	colorlinks = true,
	urlcolor = {blue},
	citecolor = {blue},
	linkcolor= {blue}
} 
\usepackage{subcaption}
\usepackage{verbatim}
\usepackage[justification=RaggedRight]{caption}

\makeatletter
\long\def\@makecaption#1#2{%
  \par
  \vskip\abovecaptionskip
  \begingroup
    \small\rmfamily
    \samepage
    \flushing
    \let\footnote\@footnotemark@gobble
    \@make@capt@title{#1}{#2}\par
  \endgroup
  \vskip\belowcaptionskip
}
\makeatother

\usepackage{optidef} 
\usepackage{xspace}

\usepackage{algorithm}
\usepackage{algpseudocode} 

\usepackage{etoolbox}
\newtoggle{isjournal}
\togglefalse{isjournal}

\newtheorem{theorem}{Theorem} 
\newtheorem{lemma}[theorem]{Lemma} 
\newtheorem{corollary}[theorem]{Corollary}

\DeclareMathOperator{\wt}{wt}
\DeclareMathOperator{\Var}{Var}

\newcommand{\tr}{\ensuremath{\mathrm{tr}}}

\newcommand{\yo}[1]{{#1}}

\def\>{\rangle} 
\def\<{\langle}
\def\bx{{\bf x}}

 \def\case#1{{\left\{  
	\begin{array}{ll}
  #1
	\end{array}
 \right.   }} 

\def \app{Appendix\xspace}

\begin{document}

\title{Finite-round quantum error correction on symmetric quantum sensors} 
\author{Yingkai Ouyang}
\email{y.ouyang@sheffield.ac.uk}
\affiliation{School of Mathematical and Physical Sciences, University of Sheffield, Sheffield, S3 7RH, United Kingdom} 

\author{Gavin K. Brennen}
\email{gavin.brennen@mq.edu.au}
\affiliation{Center for Engineered Quantum Systems, Dept. of Physics \& Astronomy, Macquarie University, 2109 NSW, Australia}

\affiliation{BTQ Technologies, 16-104 555 Burrard Street, Vancouver, British Columbia, Canada V7X 1M8}

\iftoggle{isjournal}{\maketitle}{}

\begin{abstract} 
In quantum sensing using $N$ entangled probes, the variance of the estimated signal strength $\hat{\theta}$ scales like $ \Theta(N^{-2})$ at the Heisenberg limit, which is a quadratic improvement over the standard quantum limit, and is the maximum quantum advantage  over classical methods. This limit remains elusive, however, because of the inevitable presence of noise decohering quantum sensors. 
% In some quantum sensor architectures such as photonic or atomic systems, leakage/loss errors, and in the worst case, deletion errors, which are untracked particle losses, can occur.
Here, we introduce a quantum sensing protocol \texttt{ECSense} based on permutation-invariant quantum error correction (QEC) codes that support tunable code parameters to suit the physical noise model. We show that when the signal duration is much shorter than decoherence times, such that errors only accumulate during the $N$ qubit probe state preparation and idle stage, then the estimate's variance of $\Theta( N^{-3/2})$ is achievable in the presence of $\Theta(\sqrt{N})$ errors while the Heisenberg limit is achieved when the number of errors is a constant. In the more challenging setting where errors also occur during signal accumulation, we prove using a non-Markovian QEC strategy, that even for a linear number of deletion errors, a variance approaching the Heisenberg limit is still achievable.
We illustrate a concrete way to  implement our protocol on near-term quantum hardware using cavity-assisted geometric phase gates. 
\end{abstract}

\iftoggle{isjournal}{}{\maketitle}

Estimation of physical parameters is a staple in fundamental science. 
Abstractly, one may use $N$ repeated experiments to estimate a physical parameter, and find this estimate's variance to scale as $1/N$.
In contrast, quantum sensors promise the estimation of physical parameters, such as gravitational, magnetic and electrical fields, with unprecedented precision \cite{RevModPhys.89.035002}. 
The idea behind quantum sensing is that instead of using $N$ experiments, we employ $N$ entangled qubits as a probe state for the physical process to act on. Then, by measuring the evolved probe state, the variance of the estimate of a physical parameter can scale as $1/N^2$, which is a quadratic improvement over the purely classical case. 

While ideal quantum sensors would unlock this quadratic advantage, practical quantum sensors lose their quantum advantage because of noise \cite{JZL2024}. Incorporating quantum error correction (QEC) codes into quantum sensors is an attractive theoretical approach to combat this problem \cite{dur2014improved,arrad2014increasing,kessler2014quantum,unden2016quantum,matsuzaki2017magnetic,Zhou2018,layden2019ancilla,gorecki2020optimal,Zhou2020optimal,shettell2021practical}. However this strategy is beset with challenges, such as the potential of a given QEC code to destroy the signal while correcting errors.
Another challenge is preparing large enough QEC codes in near-term quantum hardware to deliver the promised quadratic advantage of quantum sensors.

 \begin{figure*}[htp]
 \centering  
\includegraphics[width=0.9\textwidth]{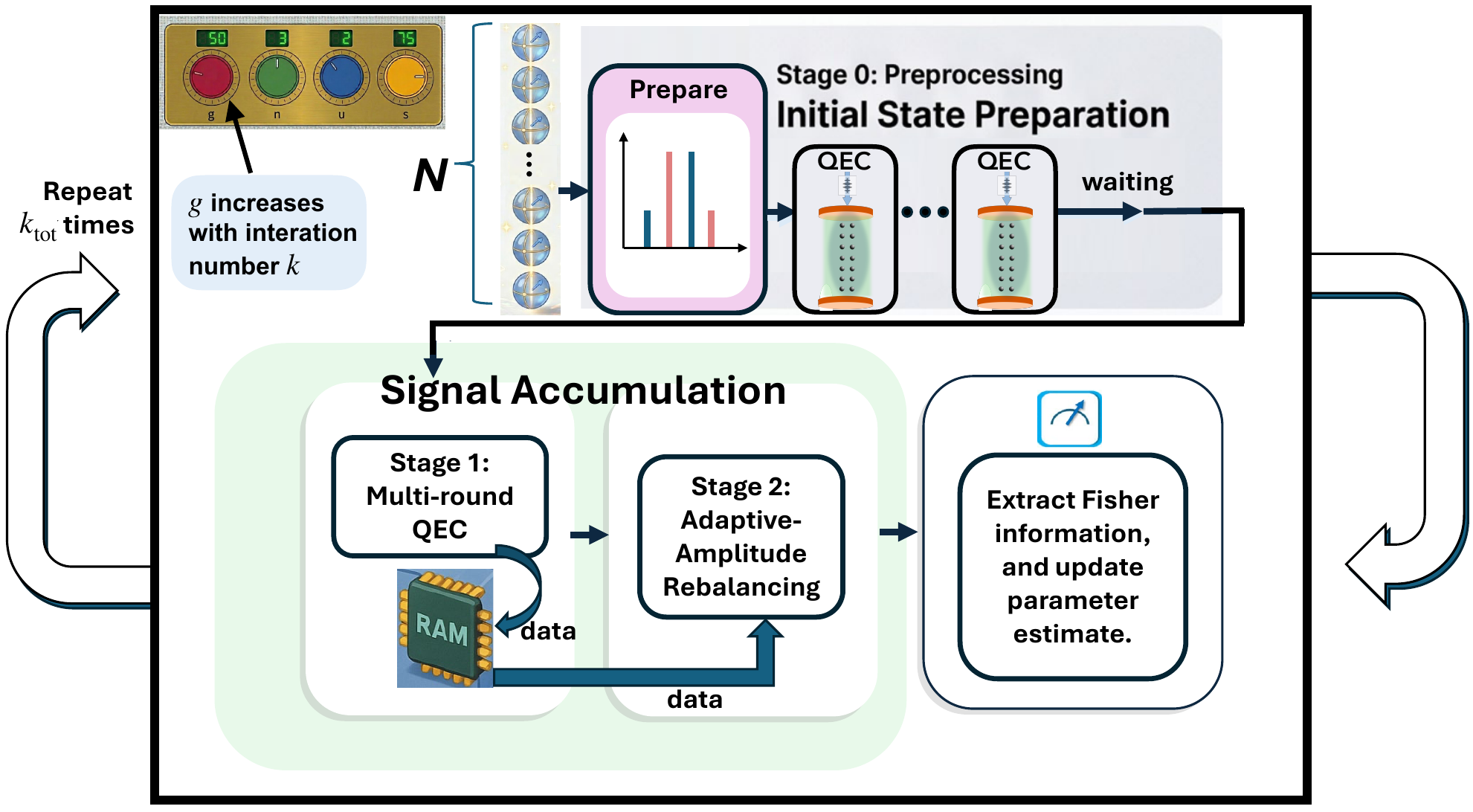} 
 \caption{{\bf{\texttt{ECSense}.}} We illustrate the general structure of our quantum sensing protocol that iterates a constant number $k_{\rm tot}$ times. 
 The parameters of the shifted gnu probe state is determined by four tunable integers, $g,n,u$ and $s$. We illustrate the structure of the shifted gnu probe state by plotting amplitudes versus Dicke state weight. The blue lines indicate the location of the logical 0 state, and the red lines indicate the location of the logical 1 state.
 Within each iterate, there is a preprocessing stage, a signal accumulation stage, followed by the signal extraction stage. Within each stage, an appropriate amount of QEC is used, depending on the amount of noise in the signal accumulation stage. With PI codes the QEC can be physically implemented using cavity assisted gates as illustrated.}
 \label{fig:ECSense}
 \end{figure*}
 
For a QEC-based quantum sensing protocol to be viable for near term implementations, it is advantageous to 
use a QEC code that allows for (1) controllability with a large number of qubits, (2) correction of a constant-rate of the most likely errors, and (3) precision that approaches the Heisenberg limit. However, there is no existing protocol that simultaneously satisfies (1), (2) and (3).
This situation gives rise to the following question:
\newline

\noindent
{\bf
How can one achieve quantum advantage with quantum sensors by using noise-tailored error correction?}
\newline

Here, the effect of the noise enters into consideration in two primary ways: (1) how it is biased, e.g. bit flip vs. phase flip noise, and (2) how it is apportioned during the various stages of a sensing protocol, e.g. errors during probe state preparation and waiting for the signal vs. errors that accumulate during the signal. Our answer to this question is our 
design of a QEC-enhanced field Sensing protocol (\texttt{ECSense}) based on permutation-invariant (PI) quantum codes. 
We show that PI codes are a good choice in that they provide tunable code parameters to achieve biased bit flip or phase flip distances addressing (1), and also operate well in the demanding scenario of consideration (2) where noise occurs during the signal. Indeed regarding the later, even with a constant rate of deletion errors, \texttt{ECSense} has a precision that approaches the Heisenberg scaling in the limit of a large number of qubits. 
Our construction relies on a non-Markovian QEC strategy where information is fed forward from one round to the next. We furthermore demonstrate that our scheme admits near-term implementations using simple quantum control techniques, such as global collective couplings, mode displacements, and heterodyne-style detections, which lies within reach of near-term quantum platforms. 
 
  While deletion errors are not the most general single qubit error model during signal accumulation, it allows us to prove quantum advantage with PI codes and is relevant in practical scenarios. A prominent example is trapped neutral atoms, which is a leading  modality for quantum computing and sensing \cite{Srivastava_2026}. Leakage outside the computational subspace and loss are significant sources of error in atomic systems. For example a recent experiment~\cite{lancaster2025quantumsensingpresencepulse} on magnetic dipole field sensing with Rb atomic ensembles trapped in a Ne matrix observed that qubit leakage was the dominant source of error. Fortunately, it is possible to achieve near-perfect conversion of leakage to loss errors in trapped atomic systems \cite{Wu2022} which can then be corrected for  with QEC.
 However, using standard methods to correct for loss it is necessary to track where a loss occurred via so-called Leakage Detection Units (LDUs) \cite{LDU1,LDU2,LDU3,LDU4}. These LDUs involve small circuits attached to each data qubit and impose substantial overhead in control complexity, addressability, and detector efficiency.
 Therefore, simplifying the process by removing tracking requirements would improve implementation efficiency at the expense of having untracked particle losses, called deletion errors \cite{leahy2019quantum,HagiwaraISIT2020,ouyang2021permutation}.
Deletion errors are more damaging than tracked particle losses, because while QEC codes of distance $d$ can correct $d-1$ tracked particle losses, they are generally susceptible to even a single deletion error.

\texttt{ECSense} builds upon the theory of quantum sensing using symmetric states that support QEC. Such states reside in permutation-invariant codes, which not only comprise of states that are invariant under any permutation of the underlying qubits.
Permutation-invariant codes \cite{Rus00,PoR04,ouyang2014permutation,ouyang2015permutation,OUYANG201743,ouyang2019permutation,ouyang2021permutation}, have several features that make them attractive candidates for use in quantum sensing.
First, their controllability by global fields could allow for their scalable physical implementations \cite{johnsson2020geometric} in near-term devices such as trapped ions, or ultracold atoms where addressability without cross-talk is challenging. 
Second, unlike most conventional QEC codes, permutation-invariant codes can correct deletions \cite{ouyang2021permutation}. Third, the simple structure of permutation-invariant codes makes it easier to design specialized QEC protocols for quantum sensing.

\texttt{ECSense} comprises three stages, Stage 0, Stage 1 and Stage 2.
In Stage 0, the preprocessing stage of \texttt{ECSense}, we prepare a quantum state in a permutation-invariant code's logical plus state, and correct errors that occurred during state preparation and whilst waiting for the signal to arrive.
In Stage 1, we allow deletions to occur in concert with the signal accumulation. 
The effect of the signal accumulation is to rotate the codespace.
With each round of QEC, we project the state onto one of two spaces, then perform a conditional unitary to bring the state back to the codespace.
This projection allows the phase to effectively accumulate as a rotation about the logical Z-axis in the original codespace.
The effective rotation angle depends on which of the two spaces the state was projected onto, and how many deletions occurred.

Deletion errors also degrade the amount of extractable information from the accumulated phase by skewing the ratio of the modulus of the amplitudes on logical zero and logical one state away from parity. Stage 2 performs an adaptive sequence of specially designed projective measurements that uses the entire history of all preceding measurements, to, in a deliberate manner, bring this ratio back towards parity. 
This highly non-Markovian process thereby recovers the information of the accumulated signal, and allows an improved precision of the field's estimate compared to its prior precision.
Repeating our subroutine, 
we iteratively improve the estimated field's precision. 
We show that this precision not only surpasses the standard quantum limit, but can also approach the Heisenberg limit.

 In summary, our paper shows, how under realistic conditions, we can have QEC-enhanced quantum sensors that outperform traditional quantum sensors using existing quantum control techniques. 
 This will be especially relevant to practical quantum sensing experiments, where the potential of QEC in the design of practical quantum sensors has not yet been explored.

%%%%%%%%%%%%%%%%%
\section{Symmetric probe states and local estimation} 
%%%%%%%%%%%%%%%%%

\begin{table}[t]
\centering
\caption{Tunable \texttt{ECSense} protocol parameters}
\label{tab:notation-ecsense-table1}
\renewcommand{\arraystretch}{1.15}
\begin{tabular}{|l|p{0.75\linewidth}|}
\hline
\multicolumn{2}{|c|}{\textbf{Tunable \texttt{ECSense} protocol parameters}}\\
\hline
$r$ & Rounds of signal accumulation in Stage~1 before rebalancing.\\ 
 $k_{\rm tot}$ & Total number of iterations in \texttt{ECSense}.\\
$\delta$ & Small positive parameter. The smaller the $\delta$, the closer we get to the HL, and the slower the convergence will be in the total number of qubits $N$. \\
 
${B}_0$ & The standard deviation of a prior estimate of the parameter $\theta \in[-\pi,\pi]$ is ${B}_0 = \Theta ( N^{-{b}_0} )$, with $b_0=1/2$
corresponding to sensing according to the SQL.\\
\hline
\end{tabular}
\end{table} 

Quantum sensing of a classical field's magnitude proceeds in several stages, which are (1) preparation of an $N$-qubit probe state, (2) signal accumulation of the classical field on the probe state, and (3) measurement. 
We repeat (1), (2) and (3) on multiple copies of $N$-qubit probe states, 
and based on the measurement statistics, we estimate the classical field's magnitude.
During signal accumulation, a classical field of initially unknown magnitude $E$ interacts with the qubits via the interaction Hamiltonian $\mu E {\hat J^z}$. \footnote{Hereafter we set $\hbar\equiv 1$.} Here, $\mu$ is the known coupling coefficient, ${\hat J^z} = \frac{1}{2} \sum_j Z_j$ is the collective spin operator, and $Z_j$ applies a phase-flip on only the $j$-th qubit. 
The signal modeled by the unitary
$U_\theta = \exp(-i \theta {\hat J^z})$ maps $\rho$ to $\rho_\theta = U_\theta \rho U_\theta^\dagger$. 
The parameter $\theta$ is dimensionless, 
and can be expressed as $\theta = b T$, where both $T$ and $b = \mu E t_0$ are dimensionless, $T t_0$ is the time, and $t_0$ is a unit of time. For convenience, we set $T=1$, so $b= \theta$, though $T$ really could be any constant.
The only unknown parameter is the field strength $E$.
We maximize the information we extract on $\theta$ by
measuring $\rho_\theta$ in a well-chosen basis.
Using the measurement statistics, we construct a locally unbiased estimator $\hat \theta$ for the true value of $\theta$ with minimum variance ${\rm Var}(\hat \theta)$. With $\hat \theta$, one estimates the field strength.
It is standard practice to use a classical sensor to have estimated the parameter $\theta$ to a sufficient level of precision so the uncertainty in the value of $\theta$ lies in a range from $-\pi$ to $\pi$.

Using the language of quantum metrology, given a probe state $\rho_\theta$ that depends on a parameter $\theta$, we want to find the minimum variance estimator $\hat \theta$ of $\theta$ that is furthermore locally unbiased. 
The celebrated quantum Cram\'er-Rao bound gives an attainable lower bound 
$\Var(\hat \theta) \ge Q(\rho_\theta, \frac{d\rho_\theta}{d\theta})^{-1}$
\cite{sidhu2019geometric},
where $Q(\rho_\theta, \frac{d\rho_\theta}{d\theta}) = \tr(\rho_\theta L^2)$ denotes the quantum Fisher information (QFI) and $L$, the symmetric logarithmic derivative (SLD), is any Hermitian solution of the Lyapunov equation
$    \frac{d\rho_\theta }{d \theta}
    =\frac{1}{2}(L \rho_\theta +  \rho_\theta L). $ 
In a suboptimal measurement strategy, where we do not measure in the SLD operator's eigenbasis, we can still obtain a probability distribution from which we can estimate $\theta$, by calculating the Fisher information (FI) of that probability distribution. The Cram\'er-Rao bound then tells us that $\Var(\hat \theta) $ is at least the inverse of the FI of such a probability distribution.

Here, because we use the FI, we work in the local estimation theory paradigm, where we require prior knowledge of $\theta$. We can obtain such prior knowledge using classical techniques, such as with $N$ measurements.
Then the variance of the classical estimator $\hat \theta_{\rm classical}$ is proportional to $1/N$. This precision, which represents the ultimate precision limit of classical and semi-classical estimation strategies, is the standard quantum limit (SQL). 
Noiseless quantum sensing of a classical field promises a Heisenberg limited (HL) precision, where there exists an estimator
$\hat \theta_{\rm HL}$ that has variance proportional to $1/N^2$. 

Symmetric states, invariant under any permutation of their underlying particles, are superpositions of the Dicke states  
$| D^N_w\rangle   ={\binom N w}^{-1/2}
\sum_{\substack{ x_1, \dots, x_N \in \{0,1\} \\ x_1+ \dots + x_N = w}}
|x_1, \dots, x_N\rangle ,$ 
of weights $w = 0,\dots, N$.  
Permutation-invariant codes \cite{Rus00,PoR04,ouyang2014permutation,ouyang2015permutation,OUYANG201743,movassagh2020constructing} are QEC codes comprising of symmetric states, and we consider $s$-shifted gnu codes \cite{ouyang2021permutation,OUYANG201743} on $N=gnu+s$ qubits with logical codewords
\begin{align}
|j_{g,n,u,s}\rangle  &= 2^{-(n-1)/2} \sum_{\substack{ {\rm mod}(k,2)=j \\ 0 \le k \le n}} {\binom n k}^{1/2}  |D^{gnu+s}_{gk+s}\rangle ,
\end{align}
for $j=0,1$ as candidates for QEC-enhanced quantum sensing. 
Intuitively, $g$ corresponds to the bit-flip distance, $n$ corresponds to the phase-flip distance. Here, $u\ge 1$ and $s\ge 0$ both specify the number of qubits used, and the distribution of the Dicke weights.
These $s$-shifted gnu codes have distance ${\rm min}\{g,n\}$, which allows the correction of ${\rm min}\{g,n\}-1$ deletion errors,
and the correction of
$\lfloor ({\rm min}\{g,n\}-1)/2 \rfloor $ general errors.  
Shifted gnu codes transform under the unitary $U_\theta$ into
{\em $\theta$-rotated and $s$-shifted gnu codes}, which have logical codewords
$|j_{g,n,u,s,\theta}\rangle  = U_\theta |j_{g,n,u,s}\rangle $ 
for $j=0,1$ and  
also have a distance of ${\rm min}\{g,n\}$.

% Table of notation/parameters used in this section 
\begin{table}[t]
\centering
\caption{Notation for code and signal parameters}
\label{tab:notation-amplitude-rebalancing}
\renewcommand{\arraystretch}{1.15}
\begin{tabular}{|l|p{0.75\linewidth}|}
\hline
\multicolumn{2}{|c|}{\textbf{Shifted gnu code parameters}}\\
\hline
$N$ & number of physical qubits.\\
$g$ & bit-flip distance and the  spacing between Dicke weights. The choice of $g$ varies with the \texttt{ECSense} iterate number $k$ as \newline $g=\Theta(
    N^{1 + O(\delta)
    -(2/5 + O(\delta))^{k}(1/2+O(\delta)) } 
    )$.
\\
$n$ & phase flip distance: the number of Dicke weights  is $n+1$.  \\
$s$ & shift: offset of Dicke weights from zero.\\
$u$ & scaling parameter: $u = (N-s)/gn$. \\
\hline
\multicolumn{2}{|c|}{\textbf{Signal parameters}}\\
\hline
$T$ & dimensionless time set to $T=1$\\
$t_0$ & unit of time.\\
$Tt_0$ & total time of each iteration of \texttt{ECSense}.\\
$\theta \hat J^z$ & interaction Hamiltonian.\\
$\theta$ & dimensionless parameter to estimate. \\
\hline
\end{tabular}
\end{table}

Throughout our paper, we focus on shifted gnu codes, and therefore as a shorthand, we write $|j_L\rangle = |j_{g,n,u,s}\rangle$.
The probe state of our quantum sensing protocol is the logical plus state of the shifted gnu code, given by
\begin{align}
|+_{L}\rangle  = (|0_{L}\rangle  + |1_{L}\rangle )/ \sqrt 2.
\end{align}
Furthermore, we choose this logical plus state to be concentrated about the half-Dicke state;
namely, we choose the shift as $s = N/2 - gn/2 +O(1)$.
In the \app, we show that the optimal FI for such a shifted gnu code, 
where we measure in the eigenbasis of the SLD operator,
is given by $g^2n$.
We summarize notation for code and signal parameters in Table \ref{tab:notation-amplitude-rebalancing}.

%%%%%%%%%%%%%%%%%%%%%%
\section{Stage 0: Pre-processing before signal accumulation}
\label{sec:II}

In this section we focus on errors that accumulate during state preparation and after state preparation but before the signal has arrived. The errors can become substantial if the timescale at which state preparation takes place is comparable to the timescale of the errors or if the wait-time for the signal is too long. In such a scenario we will perform quantum error correction to clean up to probe state, before the signal accumulation stage. This is the pre-processing stage, which we call Stage 0.

We are able to correct up to $t$ arbitrary errors on the probe state that is prepared within a 
permutation-invariant codespace, when no qubits are lost. For this, we perform the following steps.
First, we measure the total angular momentum on nested subsets of qubits. These measurements of total angular momentum occur in the sequentially coupled basis \cite{jordan2009permutational,havlivcek2018quantum}, where 
subsets of qubits that we measure are 
$\smash{[\ell] = \{ 1 ,\dots, \ell \}}$, where $\ell=1,\dots,N$. The corresponding total angular momentum operators to be measured are
$\hat{J}^2_{[1]} , \dots, \hat{J}^2_{[N]}$\footnote{Note that throughout we have used the simplified notation for the total angular momentum of all spins as $\hat{J}$ instead of $\hat{J}_{[N]}$.} where 
\begin{align}
\hat{J}_{[\ell]}^2 &= {({\hat J}^x_{[\ell]})}^2 + {(\hat{J}^y_{[\ell]})}^2 + {(\hat{J}^z_{[\ell]})}^2,
\end{align}
and
\begin{align}
\hat{J}^x_{[\ell]} &= \frac{1}{2} \sum_{i =1}^\ell X_i,\quad
\hat{J}^y_{[\ell]} = \frac{1}{2}\sum_{i =1}^\ell Y_i,\quad
\hat{J}^z_{[\ell]} = \frac{1}{2}\sum_{i =1}^\ell Z_i.
\end{align}
The eigenvalues of the operators $\hat{J}_{[\ell]}^2$ are of the form $j_{\ell}(j_{\ell}+1)$ where $2j_{\ell}$ are positive integers. 
After measurement, $\hat{J}_{\ell}^2$ gives an eigenvalue of $j_{\ell}(j_{\ell}+1)$, and we can infer  the total angular momentum number $j_{\ell}$. These total angular momentum numbers belong to the set
\begin{align}
T_{\ell} = \{  \ell/2 -j : j = 0,\dots, \lfloor \ell/2 \rfloor, \ell/2 - j \ge 0 \}.
\end{align}
Since the total angular momentum operators $\hat{J}_{[\ell]}^2$ all commute, the order of measuring these operators does not affect the measurement outcomes. Hence we may measure $\hat{J}_{[\ell]}^2$ sequentially; that is we measure $\hat{J}_{[2]}^2$, followed by $\hat{J}_{[3]}^2$, and so on (noting that $\langle \hat{J}_{[1]}^2\rangle =\frac{3}{4}$ always). Using the observed total angular momentum $j_1, \dots, j_N$,
we construct standard Young tableaus (SYTs) encapsulating the measurement information \cite{newpaper}.

Second we perform a quantum process that brings the state back to the original PI code. 
We can do this using a partial inverse quantum Schur transform together with an adaptive measurement-based strategy on the symmetric subspace, or a teleportation-based scheme \cite{newpaper}. For the teleportation-based scheme we use two registers: a freshly prepared ancillary PI code basis state in register A, 
and the output of the probe state after the total angular momentum measurements in register B. The state in register B is essentially in a codespace labelled by an SYT. 
Then we perform logical CNOT gate with control on register A and target on register B, before performing a logical Z measurement on register B. 
Depending on the measurement outcome, we apply a logical X correction on register A. 
The teleportation procedure is agnostic to which SYT the state in register B is in and how many qubits are in registers A and B. We can perform such a logical CNOT with geometric phase gates \cite{newpaper}.

The procedure of QEC when the errors are deletion errors, is considerably simpler than the case of general errors for permutation-invariant codes \cite{newpaper}. 
Namely, for shifted gnu codes with $g$ larger than the number of deletions, the effect of deletions is to randomly shift the state into a linear combination of states supported on Dicke states of weights with different values modulo $g$. 
Then the error correction procedure is to first measure the Dicke weights of the states modulo $g$, then apply geometric phase gates to bring the state back into the codespace.

Now, recall that our protocol uses shifted gnu codes initialized in the logical plus state with appropriate values of $g,n$ as probe states. 
In the limiting case where the signal is much faster than the error rate, we can treat the signal as essentially error-free. 
Then, after the error correction step and the fast signal accumulation, by measuring the probe state in the optimal basis, the QFI is $g^2n$.
This corresponds to a mean-square error in estimating $\theta$ that scales as $g^{-2}n^{-1}$ according to the quantum Cram\'er-Rao bound. For probe states of length $N=gn$ (that is when $u=1$ and $s=0$),
the QFI is $N^2/n$.
Therefore, if $n$ is constant, 
the quantum FI is $\Theta(N^2)$ and we achieve the HL while correcting a constant number of errors. If $n = \Theta(\sqrt N)$, we can correct more errors, up to $\Theta(\sqrt N)$ errors, and with a lower QFI of $\Theta(N^{3/2})$, but still beating the SQL.
Indeed, as long as the number of errors is sublinear in $N$, we can set the minimum of $g$ and $n$ to be of the order of the number of errors to surpass the SQL in estimating $\theta$, and have a quantum advantage.

We can also conveniently measure the shifted gnu probe state in the logical plus-minus basis, to obtain a FI of $9 g^2 \sin^2 (g\theta) /4 $ when $n=3$. 
Hence, measuring in the logical plus-minus basis is close to optimal for field sensing using shifted gnu probe states with $n=3$ when $g \theta$ is close to $\pi/2$.

% When $g\ge n$, such shifted gnu codes have a code distance of $n$, these codes can correct up to {\em any} $t$ errors, provided that $t \le (n-1)/2$.
% Hence if $t \le (n-1)/2$, we can perform QEC to obtain a quantum FI of $g^2n$.
% Hence, when 
% $t=\Theta(N^{\eta})$,
% $g = \Theta(N^\alpha)$,
% the mean square error for estimating $\theta$ is proportional to 
% $N^{-1-\alpha}$ whenever {$\alpha \ge \eta$} and {$\alpha\ge 1/2$}.
% In particular, whenever $\eta \le 1/2$,
% the mean square error for estimating $\theta$ is proportional to 
% $N^{-2+\eta}$, which approaches the HL only in the limit of vanishing errors with $\eta \to 0$.
% However, as long as the number of errors is sublinear in $N$, so that $\eta < 1$, the mean square error for estimating $\theta$ surpasses the SQL, and we have a quantum advantage.

%%%%%%%%%%%%%%%%%%%%%%
\section{Errors during signal accumulation}
%{\bf Errors during signal accumulation}:-
%\label{ssec:QEC-during-sensing}
%%%%%%%%%%%%%%%%%%%%%%
%We assume that the signal accumulates and deletion errors occur at a uniform rate and simultaneously.  

% We estimate $\theta$, and thereby, also the parameter $\theta$, since we assume accurate clocks to measure the time. 
The protocol that estimates $\theta$ assumes prior knowledge of its standard deviation, which we denote as ${B}$. The goal of each of $k_{\rm tot}$ runs within \texttt{ECSense} is to iteratively improve our estimate of $\theta$, which translates to an iterative reduction in the value of ${B}$.

Each run during the signal accumulation stage comprises two main stages.
The first stage is QEC-enhanced signal accumulation (see Figure \ref{fig:protocol1}), where we perform QEC after possible occurrences of deletions and the signal accumulates in each of $r$ timesteps. 
The QEC steps, are specially tailored to the quantum sensing protocol and are non-standard. These QEC steps ensure that the signal accumulates effectively on the codespace by projecting the evolved probe state back into the codespace, whilst limiting the effects of deletions.
Deletions create the problem of ``amplitude imbalance'' that reduces our ability to extract the signal. This occurs because the distortion ratio $R$, defined as the ratio between the magnitude of the probability amplitudes of the logical zero and logical one states, can deviate far from its original value of 1.

The second stage performs $v$ rounds of amplitude rebalancing steps via an adaptive sequence of quantum channels (see Figure \ref{fig:protocol2}). This amplitude rebalancing, with probability almost 1, restores the distortion ratio $R$ back towards 1. Amplitude rebalancing occurs while the signal continues to accumulate, and its effect is analogous to taking a controlled biased random walk to get within a small distance of the target position. We conclude the second stage once $R$ is sufficiently close to 1, which makes $v$ a number that is not determined in advance.

In Stage 1, in each of the $r$ rounds, we allow the signal to accumulate for a small amount of dimensionless time $1/r$, during which deletions may occur. 
Since the signal commutes with deletions, we may without loss of generality take the signal to accumulate either before or after the deletions.
The time of each timestep is $t_0/r$ for some constant time $t_0$.
In each round, $U_{\Delta}$ acts on the probe state, where $\Delta = \theta/r$.
We assume that $\theta$ is a non-zero constant independent of $N$.  
The QEC part of this stage proceeds in four steps.
First, we count the number of qubits deleted from the system.  
Second, we perform a modular measurement in the $\hat J^z$ basis, which determines which set of Dicke states the state is supported on.
Third, we perform a projective measurement, that determines whether the state is in (1) the codespace $\mathcal C$, 
or (2) a subspace $\mathcal Q$ orthogonal to the codespace.
In (2), quantum control returns the state into a shifted gnu code's codespace on the remaining qubits. 
Intuitively, these steps allow the signal to accumulate within the codespace as an effective rotation generated by the code's logical $Z$ operator.

Stage 2, the amplitude rebalancing stage, uses information gathered from Stage 1 on the number of deletions and the shift in the weights of the Dicke states in each round. 
Based on this information from Stage 1, a classical computer estimates the value of the distortion ratio $R$.
In each amplitude rebalancing round, we first count the number of deletions. 
If there is at least 1 deletion, 
we perform the second to fourth steps of the QEC steps in Stage 1. 
If there are no deletions, then the second step considers the value of $R$, and if $R> 1$ we project the state with accumulated phase onto one of two orthogonal spaces $S^0_-$ and $S^1_-$.
If $R < 1$,
we project the state with accumulated phase onto one of two orthogonal spaces $S^0_+$ and $S^1_+$.
If the state projects onto $S^0_\pm$, the distortion ratio shifts closer to 1 by a constant factor, and effectively no phase accumulates. 
Otherwise, the distortion ratio shifts in the worst case away from 1 by a constant factor, 
and we keep track of the accumulated phase. 
There is some uncertainty in the updated values of the distortion ratios because of the intrinsic uncertainty in the parameter $\theta$ that we wish to estimate.
After each round, we bring the state back into a shifted gnu codespace.
The algorithm stops when the distortion ratio is within a constant of 1.

We propose implementing state preparation, QEC, and measurements entirely with geometric phase gates (GPGs) \cite{johnsson2020geometric}, which rely on a dispersive coupling of the qubits with a catalytic bosonic mode. GPGs require only four native operations. The first operation is the initialization of the bosonic mode, which is achievable with the use of a laser.
The second operation is driving field on a mode that is also linearly coupled to the spins. This kind of interaction has been generated with trapped ions whose spin state is coupled to a motional mode \cite{trappedionrev2025}, trapped neutral atoms coupled to optical~\cite{Neutralatomoptical} or microwave cavity mode~\cite{RydbergMicrowave}, and in cavity QED architectures~\cite{eickbusch2022fast,wang2024dispersive}. Third, we need displacement of the mode, and fourth, homodyne detection, both of which have mature implementations across an array of platforms \cite{PhysRevLett.94.113601,PhysRevLett.130.143004}. 
By moving beyond the usual Clifford + $T$ paradigm, GPGs eliminate the need for individual qubit addressability, and exploit 
readily available bosonic manipulations. This offers a hardware-efficient route to realizing our protocol in near-term devices. 

We carefully choose the parameters of the protocol according to the current precision of the parameter $\theta$. Namely, given a current value of the precision of $\theta$, we choose a corresponding value for $g$. Then $r$ is chosen as $g^{1+\delta}$ for some small positive constant $\delta$.
Next, the number of steps $v$ required in Stage 2, is, on average, roughly $\sqrt {g^3/N}$ (see Table \ref{tab:notation-ecsense-tab3}). We denote the number of iterations within \texttt{ECSense} as a constant $k$.

 \begin{figure}[!h]
 \centering  
\includegraphics[width=0.5\textwidth]{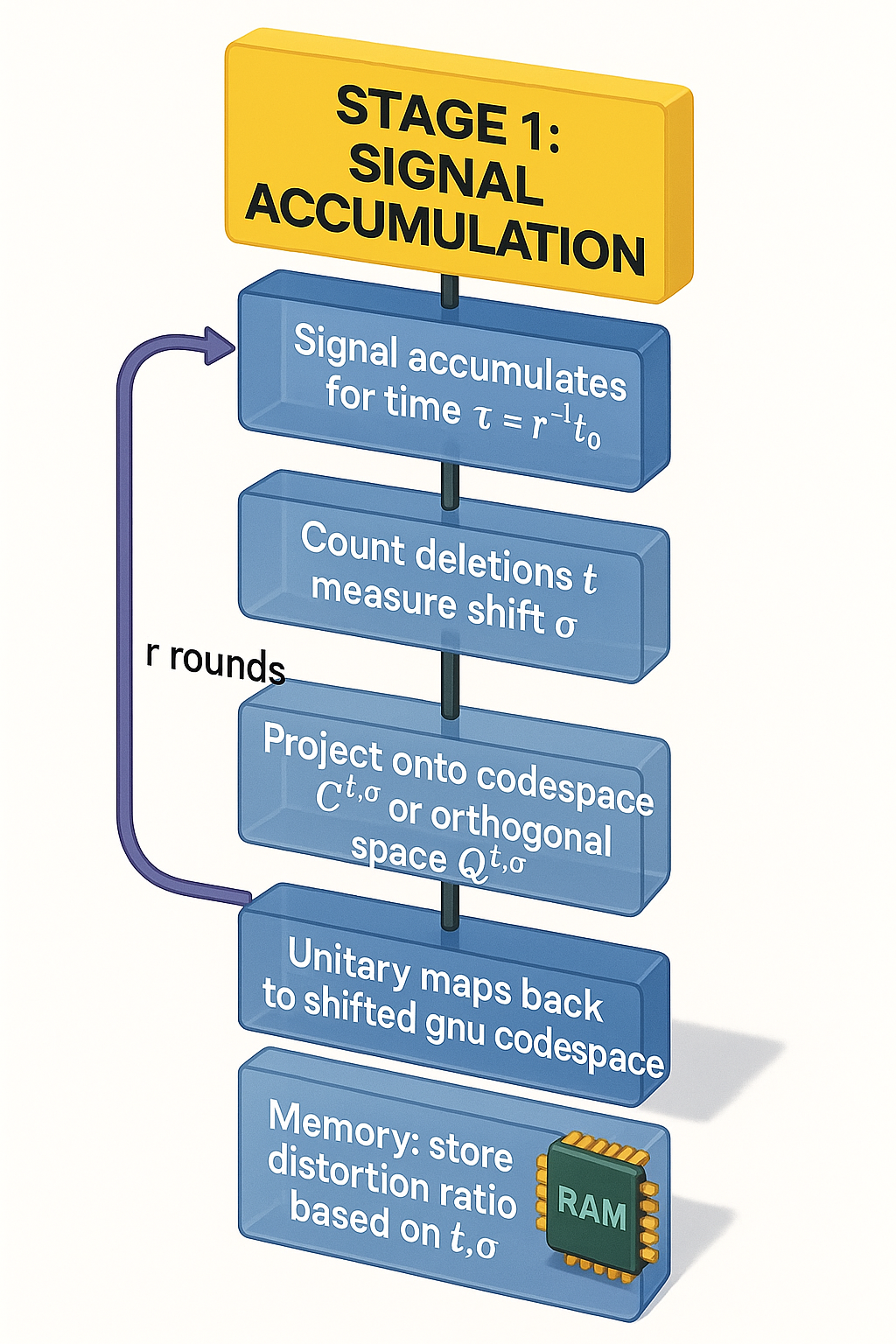} 
 \caption{{\bf Stage 1 in \texttt{ECSense}}: $r$ rounds of signal accumulation, each over a time $\tau t_0$, followed by QEC. 
 If the projection is on the codespace $\mathcal C^{t,\sigma}$, the subsequent unitary does nothing.
 Otherwise, if the 
  projection is on the orthogonal space $\mathcal Q^{t,\sigma}$, the subsequent unitary maps $\mathcal Q^{t,\sigma}$ back to the codespace.
 Here, the dimensionless time is $\tau = r^{-1}$, and $t_0$ is a dimension-full time unit.
 Here, the size of $r$ is determined by the gnu code parameter $g$ and a positive exponent $\delta$ that 
shows how $r/g$ grows with $N$.}
 \label{fig:protocol1}
 \end{figure}

 \begin{figure}[!h]
 \centering  
\includegraphics[width=0.5\textwidth]{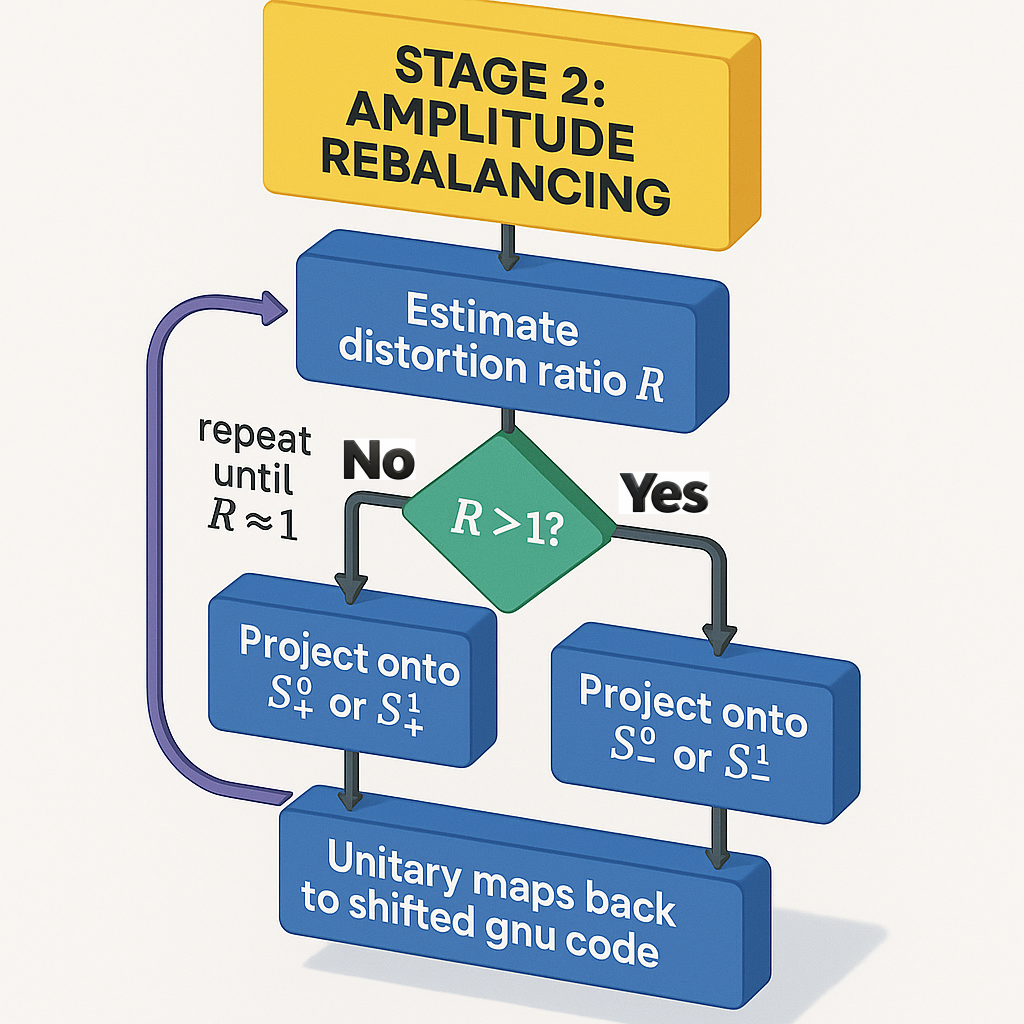} 
 \caption{{\bf Stage 2 in \texttt{ECSense}}: Rounds of amplitude rebalancing until the amplitudes are balanced. Each timestep takes a time of $t_0 / N^{1+\delta}$, and so essentially no deletions occur in this stage. 
 Successful rebalancing occurs with projection onto $S^0_\pm$ with probability at least 5/8.
 If deletions do happen, we measure the number of deletions and the resultant random shift in Dicke weight, and calculate the resultant change in the distortion ratio. The expected number of rounds is $v$.
 }\label{fig:protocol2}
 \end{figure}

\subsection{The case of no deletion when $n=3$}

Prior to the occurrence of deletions in each round, we denote the probe state as $\xi_0|0_{L}\rangle  + \xi_1|1_{L}\rangle .$ We indicate how our calculations generalize to larger values of odd $n$ by keeping the dependence of expressions on $n$ until the very end, and supply full details of what happens with our protocol for larger values of $n$ in the \app.

During each round with no deletions, 
the QEC step 
first performs a projective measurement 
of the state 
$U_{\Delta}
(\xi_0 |0_{L}\rangle 
+ \xi_1 |1_{L}\rangle )$
onto one of two spaces. 
The first space is the codespace $\mathcal{C}$ of the shifted gnu code, and
the second space $\mathcal{Q}$ is spanned by
\begin{align}
|{q_0}\rangle  &= 
\frac{1}{ g\sqrt n }
\left({\hat J^z} | 0_{L}\rangle  - \langle 0_{L} | {\hat J^z} | 0_{L}\rangle   |0_{L}\rangle \right) ,\notag
\\
|{q_1}\rangle  &=  
\frac{1}{ g\sqrt n }
\left({\hat J^z} | 1_{L}\rangle  - \langle 1_{L} | {\hat J^z} | 1_{L}\rangle   |1_{L}\rangle\right)
\label{eq:q0q1}.
\end{align}
Since each of the logical codewords is supported on two Dicke states, each logical codeword evolves under $U_\Delta$ to a superposition of two Dicke states. Hence, we project the evolved codespace onto either 
(1) the original codespace, or 
(2) an orthogonal space $\mathcal Q$.

The projected evolved state on $\mathcal C$ has the form 
$\xi_0|0_{L}\rangle  + \xi_1 e^{i \zeta_0}|1_{L}\rangle $
 where
\begin{align} 
\zeta_0 &= 
-2\arctan\left(  \tan^{3}(g\Delta/2)  \right) \approx -(g \Delta)^3/4.\notag
% 2\arctan\left(   {(-1)^{(n-1)/2} }\tan^{n}(g\Delta/2)  \right)  .
% \label{def:zetas0}
\end{align}
Projection onto the space $\mathcal Q$
gives the projected evolved state 
$\xi_0|q_{0}\rangle  + \xi_1 e^{i \zeta_1}|q_{1}\rangle $
where 
$\zeta_1 =  g\Delta \notag. $

The probabilities of projecting onto the spaces 
$\mathcal C$ and $\mathcal Q$ are 
$\cos^6(g\Delta/2) + \sin^6(g\Delta/2)$ and
$(3/4) \sin^2(g\Delta)$ respectively, and since these probabilities sum to 1, there is no chance of projection onto any other space.
% Hence, there is no chance of an uncorrectible error leading to a Type 1 error when no deletions occur. 
When $g \Delta = g\theta/r$ is small, then the expected accumulated phase per round is dominated by 
$(g \Delta)^3/2,$
% the expression
% \begin{align}
%     -(g \Delta)^3/4
%     + (3/4)(g\Delta)^2 (g \Delta )
% =(g \Delta)^3/2,
% \end{align}
and hence a non-trivial phase accumulates in expectation. 
Hence, the dominant contribution in phase arises from our projection onto the space $\mathcal Q$.
After projection onto the space $\mathcal Q$, we perform a unitary that maps the space
$\mathcal Q$ to the codespace $\mathcal C$, mapping the states $|q_j\rangle $ to the states $|j_L\rangle $ for $j=0,1.$

% \subsubsection{The case of larger odd $n$.}
% To see why the choice of $n=3$ allows the
% probability of projecting onto the uncorrectible space ${(\mathcal C \cup \mathcal Q)^\perp}$ to be zero, note the following.
% The logical operators of the shifted gnu code are supported on two Dicke states. 
% Now, the unitary operator $U_\Delta$ takes a superposition of two Dicke states to a superposition of two Dicke states. 
% Next, the operators $|Q_j\rangle $ and $|j_{L}\rangle $ are linearly independent, and span a two dimensional complex space. 

\subsection{The case of multiple deletions}

In a general setting during signal accumulation for each round, a number of deletions may occur. Let $t$ denote the number of such deletions. 
By the law of large numbers, in each round, the number of such deletions is a random variable that is between $N/(2r)$ and $2N/r$ for some positive constant $c$, with probability at least $1-e^{-c N/r}$.
Starting from
an initial shifted gnu state 
$\xi_0|0_{L}\rangle  + \xi_1 |1_{L}\rangle $, 
the occurrence of $t$ deletions produces
a probabilistic mixture of $(N-t)$-qubit states
\begin{align}
|\psi^t_\sigma\rangle   = 
\xi_0|0^t_\sigma\rangle 
+ \xi_1 |1^t_\sigma\rangle  \notag
\end{align}
where the (unnormalized) basis states are 
\begin{align}
|0^t_\sigma\rangle  &\coloneqq 2^{(-n+1)/2}\sum_{k\ {\rm even}}  \sqrt{\binom n k} \sqrt{\frac{\binom{N-t}{gk+s-\sigma}} {\binom{N}{gk+s} } } | D^{N-t}_{gk+s-\sigma}\rangle ,\notag\\
|1^t_\sigma\rangle  &\coloneqq 2^{(-n+1)/2} \sum_{k\ {\rm odd}}  \sqrt{\binom n k} \sqrt{\frac{\binom{N-t}{gk+s-\sigma}} {\binom{N}{gk+s} } } | D^{N-t}_{gk+s-\sigma}\rangle .\notag
\end{align}% approximately 
% $|0^t_\sigma\rangle  \coloneqq |0^t_{g,n,(N-1-s)/gn,s-\sigma}\rangle $
% and 
% $|1^t_\sigma\rangle  \coloneqq |1^t_{g,n,(N-1-s)/gn,s-\sigma}\rangle $ 
% respectively.
Here $\sigma$ is a random shift in the weight of the Dicke states that ranges from 0 to $t$. Signal accumulation on such a state gives the evolved probe state $U_{\Delta} |\psi^{t}_\sigma\rangle $.

Subsequently, we project $U_{\Delta} |\psi^t_\sigma\rangle $ onto either the codespace
$\mathcal C^{t,\sigma}$ 
of shifted gnu codes on $N-t$ qubits and with shift $s-\sigma$,
or the space $\mathcal Q^{t,\sigma}$ which comprises of the part of vectors in $\hat J^z \mathcal C^{t,\sigma}$ that are orthogonal to 
$\mathcal C^{t,\sigma}$.
Applying the Gram-Schmidt procedure to $|j^{t,\sigma}_L\rangle $ and $\hat J^z |j^{t,\sigma}_L\rangle  $, 
we obtain the orthonormal states 
$|j^{t,\sigma}_L\rangle  $
and 
$|q_j^{t,\sigma} \rangle $ respectively.
The space $\mathcal Q^{t,\sigma}$ is therefore spanned by
$|q_0^{t,\sigma} \rangle   $ 
and 
$|q_1^{t,\sigma} \rangle $.

We observe that the random shift $\sigma$ is for large $N$ almost surely concentrated about the value $t/2$,
that is,
\begin{align}
|\sigma - t/2 | \le c t^{1/2+\delta}
\label{sigma-concentrated}
\end{align}
for some positive constants $c$ and $\delta$,
when $r=\Theta(g^{1+\delta})$. Here, we are free to choose a positive $\delta$, and the constant $c$ then depends on the size of $N$ and $\delta$.
In such a situation, when $g \ge \sqrt N$, the distortion ratio that accumulates in a given round is 
$
|\langle 1_L^{t,\sigma}|
  U_\Delta 
|1^t_\sigma\rangle |/
|\langle 0_L^{t,\sigma}|
 U_\Delta 
|0^t_\sigma\rangle |$
and 
$|\langle q_1^{t,\sigma}|
  U_\Delta 
|1^t_\sigma\rangle |/
|\langle q_{0}^{t,\sigma}|
 U_\Delta 
|0^t_\sigma\rangle |$ for projections onto 
$\mathcal C^{t,\sigma}$
and $\mathcal Q^{t,\sigma}$ respectively.
% , which we calculate as 
% \begin{align}
%   \frac{
% |\langle 1_L^{t,\sigma}|
%   U_\Delta 
% |1^t_\sigma\rangle |}
% {
% |\langle 0_L^{t,\sigma}|
%  U_\Delta 
% |0^t_\sigma\rangle |
% }\notag
% \end{align}
% when the projection is onto $\mathcal C^{t,\sigma}$
% and 
% \begin{align}
%   \frac{
% |\langle q_1^{t,\sigma}|
%   U_\Delta 
% |1^t_\sigma\rangle |}
% {
% |\langle q_{0}^{t,\sigma}|
%  U_\Delta 
% |0^t_\sigma\rangle |
% }\notag
% \end{align}
% when the projection is onto $\mathcal Q^{t,\sigma}$,
These quantities and their inverses are upper-bounded by 
\begin{align}
% \frac{
% |\langle \psi_{1,k}^{t,\sigma}|
%   U_\Delta 
% |1^t_\sigma\rangle |}
% {
% |\langle \psi_{0,k}^{t,\sigma}|
%  U_\Delta 
% |0^t_\sigma\rangle |
% }
% =
1+O(gt^{1/2+\delta}/N).
 % +O( t/N).
 \label{distortion-bound}
\end{align} 
Hence, with probability at least $(1- e^{-c' t^{2\delta}})$, the distortion ratio in both
$\mathcal C^{t,\sigma}$ and $\mathcal Q^{t,\sigma}$ changes by a multiplicative factor of at most $1+ O(gt^{1/2+\delta}/N)$ in each round for some positive constant $c'$.
Most importantly, we show in the \app that this distortion ratio is independent of the signal strength $\theta$, regardless of the number of deletions and shifts encountered. 
We therefore know precisely the distortion ratio at each round.

Reminiscent to the no-deletion scenario, the effective phase that accumulates in the space 
$\mathcal C^{t,\sigma}$ and
$\mathcal Q^{t,\sigma}$ 
is of the order between $\Theta( g\Delta)$ and 
$\Theta( g^3\Delta^3)$ for vanishing $g\Delta$. Pessimistically, we can take the effective phase that accumulates to be of order $\Theta( g^3\Delta^3)$. A similar result also applies to shifted gnu codes with larger values of odd $n$.

In a nutshell, Stage 1 has the following steps.
First, the signal accumulates in the quantum state. 
Second, deletions happen, and we count the number of deletions. Third, we determine the random shift $\sigma$ by measurement of the weights of the Dicke states modulo $g$.
Fourth, we perform projections onto the orthogonal spaces given by 
$\mathcal C^{t,\sigma}$ and 
$\mathcal Q^{t,\sigma}$. If needed, we perform a unitary that brings 
$\mathcal Q^{t,\sigma} $ to the codespace $\mathcal C^{t,\sigma} $. Then, the probe state is mapped to a shifted gnu code with $N-t$ qubits and shift $s-\lfloor t/2 \rfloor$,
with an effective accumulated phase in the codespace.
Then we update the value of  $N$ with $N-t$, and $s$ with $s-\lfloor t/2 \rfloor$.

The distortion ratio 
$R_{j} 
= 
|\langle \psi_j| 1_L\rangle |^2 /
|\langle \psi_j| 0_L\rangle |^2 $
quantifies the amplitude imbalance at the $j$th round,
where $|\psi_j\rangle $ denotes the state at the conclusion of the $j$th round of signal accumulation.
In each run of Stage 1, we can calculate $R_r$ precisely and store this information into a memory. 
For the purpose of theoretical analysis, we estimate the likely value of $R_r$ after $r$ rounds of signal accumulation.  
Since the number of deletions per round is concentrated around a constant factor of $N/r$,
the distortion ratio $R_j$ changes by 
at most $(1+g O(N/r)^{1/2+\delta}/N)$
with probability at least 
$1-e^{-cN/r}-e^{-c''(N/r)^{2\delta}}
\ge 1- 2 e^{-c'''(N/r)^{2\delta}}$ for some positive constant $c'',c'''$.
Therefore the likely maximum value of 
$R_r$ and $R_r^{-1}$ is at most
% $w' = (1+c(N/r)^{1/2+\delta})^r$ for some constant $c.$
% We can improve this bound to
$
% w' 
% = (1+O(g/N)(N/r)^{1/2+\delta})^r
% = O(\exp(
% ( g/N )r(N/r)^{1/2+\delta}
% ))
% =
% O(\exp(
%  c g r^{1/2-\delta}
%  N^{-1/2+\delta}
% ))
% =
O(\exp(
 O( g^{1+(1/2-\delta)(1+\delta)}
 N^{-1/2+\delta})
))$
with probability at least 
$1- 2r e^{-c'''(N/r)^{2\delta}}$.
Since $r$ is bounded by a polynomial in $N$, this probability approaches 1 for large $N$.
The logarithm of this maximum rounded off to the nearest integer is then  
$w = O(g^{1+(1/2-\delta)(1+\delta)}
 N^{-1/2+\delta})$.
 Here, $w$ is the number of amplitude rebalancing steps  
% From the typical maximum change in the distortion per-round as quantified by \eqref{distortion-bound}, we deduce that the likely maximum value of $R_r$ and $R_r^{-1}$ is almost surely 
% $w=(1+c(N/r)^{1/2+\delta})^r
% \approx e^{  c r (N/r)^{1/2+\delta}}$.
% Consequently, $|\log R_r|$ is almost surely at most 
% $w' = ({1+2\delta})\log w
% \approx 
% c ({1+2\delta}) N^{1/2+\delta} r^{1/2-\delta} $.
% % \begin{align}
% % O(r(g/N)^{1/2-\delta}). \label{expected-Sw}
% % \end{align}.
% This number $w'$ gives the typical number of amplitude rebalancing shifts 
of constant magnitude 
needed to bring the amplitude imbalance back within a constant multiplicative factor of 1.
 
  \subsection{Amplitude rebalancing}

\begin{table}[t]
\centering
\caption{Notation within \texttt{ECSense}}
\label{tab:notation-ecsense-tab3}
\renewcommand{\arraystretch}{1.25}
\begin{tabular}{|l|p{0.75\linewidth}|}
\hline
\multicolumn{2}{|c|}{\textbf{Internal \texttt{ECSense} protocol parameters}}\\
\hline
% $\theta$ & dimensionless parameter\\
% $R_j$ & Distortion ratio on the $j$th round of Stage 1.\\
% $w$ & The likely maximum of $R_r$ and $R_r^{-1}$.\\
% $w'$ & Equal to $(1+2\delta) \log w$.\\
$t$ & number of deletions in a given Stage 1 round\\
$\tau$ & dimensionless timestep in Stage 1 equal to $1/r$.\\
$\sigma$ & random shift of Dicke weight in a given Stage 1 round\\
$\Delta$ & equal to $\theta/r$.\\  
$v$ & expected number of rounds in Stage 2 to rebalance the amplitudes:
\newline
$v= 
\Theta(
g^{3/2 + \delta - 3 \delta^2/2 - \delta^3}
 N^{-1/2+\delta/2+\delta^2})$.\\
$\Delta'$ & equal to $\theta/N^{1+\delta}$.\\  
$x$ & equal to $g\Delta/2$ in Stage 1.\\  
$x$ & equal to $\Theta(g/{N^{1+\delta}})$ in Stage 2.\\  
$S^0_+$ & successful amplitude rebalancing: returns more amplitude on logical zero.\\
$S^0_-$ & successful amplitude rebalancing: returns more amplitude on logical one.\\
$S^1_\pm$ & failed amplitude rebalancing projection.\\
% $\sqrt 3 (1+O(x^2))/2$ & Successful amplitude rebalancing probability. \\
\hline
\end{tabular}
\end{table}

Amplitude rebalancing (see Figure \ref{fig:protocol2}) is an adaptive algorithm that takes as input the state at the end of Stage 1 ($|\psi_r\rangle $) and outputs the state $|\bar \psi\rangle $,
where
 $|\langle \bar \psi|1_L\rangle | \approx |\langle \bar \psi|0_L\rangle |$ while preserving the accumulated signal. 
Rebalancing the amplitudes maximizes the FI on $\theta$.
The algorithm proceeds in repeated tiny timesteps of $t_0 /N^{1+\delta}$ (with $\Delta'=\theta/N^{1+\delta}$), during which practically no deletions occur and signal accumulates.
In each timestep, we project the state onto one of two two-dimensional spaces,
motivated by the amplitude balancing construction of Ref \cite{newpaper}.
In Ref \cite{newpaper}, given a base code with logical codewords $|0_L\>$ and $|1_L\>$,
the paper's amplitude rebalancing algorithm takes any input state  
\begin{align}
     |\psi\> = 
     \cos \alpha |0_L\> + e^{i \phi }\sin \alpha |1_L\>\label{w-input},
 \end{align}
 where $\alpha, \phi \in \mathbb R$ and outputs either a state
\begin{align}
     |\psi_h\> = 
     \frac{ \sqrt{3+h}  \cos \alpha |0_L\> 
     + e^{i \phi }\sqrt{3-h} \sin \alpha |1_L\>}
     {\sqrt{ 3 + h \cos 2\alpha } } \label{w-form}
 \end{align}
with probability $\frac{3}{4} + \frac{h}{4} \cos 2\alpha$
 or outputs the state
\begin{align}
|\bar \psi_h\> =  \frac{\sqrt{1-h }}{2} \cos \alpha |0_h\>
+
\frac{\sqrt{1+h}}{2} e^{i\phi }\sin \alpha |1_h\>
\end{align}
 with probability $\frac{1}{4} - \frac{h}{4} \cos 2\alpha $, 
 where $h \in [-1/2, 1/2]$.
 Intuitively, the effect of the amplitude rebalancing algorithm deforms the scalings of the logical codewords by a real amount.
 By setting $|h|=1/4$, the probability of obtaining $|\psi_h\>$ varies between $11/16$ and $13/16$ for all values of $\alpha.$

We use the amplitude rebalancing algorithm of \cite{newpaper} to project onto one of two orthogonal spaces 
spanned respectively by the vectors
\begin{align}
|0_h\>  &\coloneqq 
(\sqrt{3+h} |0_L\> + \sqrt{1-h}|q_0\>)/2,\\
|1_h\>  &\coloneqq 
(\sqrt{3-h} |1_L\> + \sqrt{1+h}|q_1\>)/2
\end{align}
and 
\begin{align}
|\bar 0_h\>  &\coloneqq 
(\sqrt{1-h} |0_L\> - \sqrt{3+h}|q_0\>)/2,\\
|\bar 1_h\>  &\coloneqq 
(\sqrt{1+h} |1_L\> - \sqrt{3-h}|q_1\>)/2,
\end{align}
where $|q_0\rangle$ and $|q_1\rangle$ are as defined in \eqref{eq:q0q1}.  

Ideally, we like to project our state onto the spaces given by 
$S^0_\pm = {\rm span}\{ |0_{\pm 1/4}\rangle , |1_{\pm 1/4} \rangle  \}$.
Projection on the spaces
$S^0_+$ 
and
$S^0_-$ 
returns a state with more amplitude on the logical zero and the logical one state respectively. We interpret the projection of the state onto $S^0_\pm$ as a successful amplitude rebalancing step, because such a projection brings the distortion ratio closer to 1. There is effectively no phase accumulation on the codespace.

The precise effects of projection onto $S^0_\pm$ are as follows.
Let $x= g\Delta'/2$. Then
\begin{align} 
&\frac
{\langle    1_{1/4} | U_{{\Delta'}}   | 1_L \rangle   }
{\langle   0_{1/4} | U_{{\Delta'}}   | 0_L \rangle   }
=\sqrt{\frac{11}{13}}
- \left(
\sqrt{\frac{15}{13}} 
+
\frac{3\sqrt{11}}{13} 
\right)x +  O(x^2)
\end{align}
and
\begin{align} 
&\frac
{\langle    1_{-1/4} | U_{{\Delta'}}   | 1_L \rangle   }
{\langle   0_{-1/4} | U_{{\Delta'}}   | 0_L \rangle   }
=\sqrt{\frac{13}{11}}
- \left(
\frac{3}{\sqrt{11}} 
+
\frac{\sqrt{195}}{11} 
\right)x +  O(x^2).
\end{align}
Moreover the accumulated phase $\arg( {\langle   1_{\pm 1/4} | U_{{\Delta'}}   | 1_L \rangle   }/
{\langle   0_{\pm 1/4} | U_{{\Delta'}}   | 0_L \rangle   }  )$ is $O(x^2)$. The probability of projecting onto the space $S^0_\pm$ is $11/16+O(x^2)$, which is at least 5/8 for small enough $x$.

If we fail to project the state onto $S^0_{\pm}$, the state is in a space $S_\pm^1$ spanned by the vectors 
$|\bar 0_{\pm h}\>$ and $|\bar 1_{\pm 1/4}\>$.  
 We find that 
\begin{align} 
&\frac
{\langle    \bar 1_{1/4} | U_{{\Delta'}}   | 1_L \rangle   }
{\langle   \bar 0_{1/4} | U_{{\Delta'}}   | 0_L \rangle   }
=\sqrt{\frac{5}{3}}
- \left(
\sqrt{11}
+
\sqrt{\frac{65}{3} }
\right)x +  O(x^2)
\end{align}
and
\begin{align} 
&\frac
{\langle   \bar  1_{-1/4} | U_{{\Delta'}}   | 1_L \rangle   }
{\langle  \bar  0_{-1/4} | U_{{\Delta'}}   | 0_L \rangle   }
=\sqrt{\frac{3}{5}}
- \left(
\sqrt{\frac{39}{5}} 
+
\frac{3\sqrt{11}}{5} 
\right)x +  O(x^2).
\end{align}
Moreover the accumulated phase $\arg( {\langle \bar   1_{\pm 1/4} | U_{{\Delta'}}   | 1_L \rangle   }/
{\langle  \bar  0_{\pm 1/4} | U_{{\Delta'}}   | 0_L \rangle   }  )$ is $O(x^2)$.

If we find that $t$ deletions occur, instead of projecting onto the spaces $S^j_\pm$, we measure the Dicke shift $\sigma$, and we project onto the spaces $\mathcal C^{t,\sigma}$ and $\mathcal Q^{t,\sigma}$ just like we did in Stage 1, but with timestep $t_0/N^{1+\delta}$. This allows a small signal to accumulate, and we update the change in the distortion ratio into our memory. In each time-step, the probability that at least one deletion occurs is vanishingly small, approaching zero. 

After projections onto either $S^0_\pm$ or $S^1_\pm$,
we map the state to the codespace of a shifted gnu code with $s \approx N/2$ via either a unitary that maps a normalized $|\tilde j_{\pm}\rangle $ to the state $|j_L\rangle $ for $j=0,1$,
or a unitary that maps a normalized $|\tilde Q^{j}_{\pm}\rangle $ to the state $|j_L\rangle $ for $j=0,1$.

In summary, the case of no deletion and successful projection onto $S^0_\pm$ happens with probability greater than 5/8, 
and projection onto $S^1_\pm$ or having at least one deletion happens with probability at most 3/8.
In the worst case, projection onto $S^1_\pm$ causes the amplitude imbalance to shift away from 1. 
We apply a classically conditioned unitary that maps the state back to a shifted gnu codespace after each of these projections and repeat the process until we are confident that the amplitude imbalance is close to 1, {\em i.e.}, there is a negligible amount of amplitude imbalance that remains. This two-dimensional subspace mapping can be efficiently implemented using the methods in~\cite{newpaper}.
We can describe this with a mean-reverting model, which we approximate as a biased random walk on the integer line,
where the value of the walk, starting at the value of 0, increases by 1 or decreases by 1 with probabilities 3/8 and 5/8 respectively, until it reaches a target value of $-w$.
We show in the \app that such a walk terminates on average in $4 w$ steps.
Hence, Markov's inequality shows that such a walk is not terminated after 
$(4 w)^{1+\delta}$ steps with vanishing probability at most $(4w)^{-\delta}$.
Hence, the number of timesteps needed for amplitude rebalancing is almost surely at most
$v= (4w)^{1+\delta}$.
Then we can calculate
\begin{align}
v 
% =
% (\Theta(g^{1+(1/2-\delta)(1+\delta)}
%  N^{-1/2+\delta}))^{1+\delta}
%  =
=
\Theta(
g^{3/2 + \delta - 3 \delta^2/2 - \delta^3}
 N^{-1/2+\delta/2+\delta^2}).
% &= (4w')^{1+\delta}  \\
% &= (4 
% (c(1+2\delta)N^{1/2+\delta}r^{1/2-\delta})
% )^{1+\delta}  \\
% &= (4 
% (c(1+2\delta)N^{1/2+\delta}g^{(1/2-\delta)(1+\delta)})
% )^{1+\delta}\\  
% &=
% \Theta(
% N^{(1/2+\delta)(1+\delta)}
% g^{(1/2-\delta)(1+\delta)^2}
% )\\
% &=
% \Theta(
% N^{1/2 + 3\delta/2 + \delta^2}
% g^{1/2 - 3\delta^2/2-\delta^3}
% ).
\end{align}

 \subsection{Towards the Heisenberg limit}

Iterate $k$ of \texttt{ECSense} begins with a level of precision for the parameter $\theta$. We quantify the uncertainty of $\theta$ with the notation ${B}_{k-1}$. The non-zero value of ${B}_{k-1}$ translates into uncertainty in the changes of the distortion ratio under application of amplitude rebalancing steps.
At the conclusion of iterate $k$, the uncertainty of $\theta$ is updated to ${B}_{k}$.

In the first iterate of \texttt{ECSense}, we choose ${B}_{0} = \Theta({1/\sqrt{N}})$ which corresponds to a precision achieved by a sensor that attains the SQL.
The purpose of subsequent iterates to improve the precision of $\theta$ beyond the SQL.
To explain this, we denote the exponent of the uncertainty with ${{b}_k}$, where 
${{B}_k} = \Theta(N^{-{{{b}_k}}})$.
In this notation, the first iterate takes in ${b}_0=1/2$.

In iterate $k$, each amplitude rebalancing step introduces an uncertainty of $g {B}_{k-1}/(2N^{1+\delta})$ in the amplitude imbalance, because of the $O(x)$ deviation of $|\langle 1_{\pm 1/4}|U_{\Delta'}|1_L\rangle / \langle  0_{\pm 1/4}|U_{\Delta'}|0_L\rangle |$ and 
$|\langle \bar 1_{\pm 1/4}|U_{\Delta'}|1_L\rangle / \langle \bar 0_{\pm 1/4}|U_{\Delta'}|0_L\rangle |$ from 1.
Therefore with the conclusion of the amplitude rebalancing algorithm that uses $v$ rebalancing steps, the deviation in the amplitude imbalance is $O(v g {B}_{k-1}/N^{1+\delta})$.
 We want {these} incurred shifts to vanish, and hence we impose the constraint
 \begin{align} 
v g {B}_{k-1}/(2N^{1+\delta})
\le \Theta(N^{-\delta}). \label{above-constraint}
 \end{align}
Asymptotically, the constraint \eqref{above-constraint} holds with equality when
 \begin{align}
     % \log_N g
     % =
     %\frac{1+2{{{B}}} -8\delta^2}{3+10\delta-8\delta^2}.
     % {
     % \frac{2{B}_{k-1} + 1 - 3\delta -2\delta^2}{3-\delta-2\delta^2}}.
     \log_N g = 
% \frac {1+2{B}_{k-1} -  3\delta - 2\delta^2  }
% {3-3\delta^2-2\delta^3}.
\frac
{{b}_{k-1} + 3/2 - \delta/2 - \delta^2}
{5/2+\delta-3\delta^2/2- \delta^3}
     \label{alpha-(2b+1)/3}.
 \end{align}
After implementing Stage 1 and Stage 2 in \texttt{ECSense}, 
we extract FI on $\theta$ by measuring the final state in the plus-minus basis of the gnu code. 
For constant non-zero $\theta$,
the expected FI of each iteration of \texttt{ECSense} is
$\mathbb E [ {F}]
 = \Theta( r^2g^6\tau^6 ) $
where $\tau = 1/r$ is the dimensionless timestep for signal accumulation. 
Since $r=\Theta(g^{1+\delta})$, we get
  $  \mathbb E [ {F}]
  =\Theta ( g^{2-4\delta}) 
  ,$
  and from the Cram\'er-Rao bound, the output uncertainty is  
  \begin{align}
   {B}_{k} = \Theta(g^{-1+2\delta} ).\label{b-crb}
  \end{align}
  Hence we can iteratively use \texttt{ECSense} to obtain an update on the output exponent ${{{b}}}_{k}$
  in terms of input exponent ${{{b}}}_{k-1}$.
  Using \eqref{alpha-(2b+1)/3} with \eqref{b-crb}, we obtain a recurrence relation in ${b}_k$
  that is given as
  \begin{align}
  {b}_k =
      \frac
{({b}_{k-1} + 3/2 - \delta/2 - \delta^2)(1-2 \delta)}
{5/2+\delta-3\delta^2/2- \delta^3}
     \label{final-recurrence}.
  \end{align}
  Solving this recurrence relation \eqref{final-recurrence} with the initial condition ${b}_0=1/2$, we obtain the solution
  \begin{align}
    {b}_k
    =&
    1 - \delta A_1
    -\left(\frac 1 2 - \delta A_1 \right)
    \left(
\frac{2(1-2\delta) }
{
5+ 2\delta - 3\delta^2 - 2\delta^3 
}
    \right)^k
   %  =& \frac{1 -5 \delta +4 \delta^2 + 4 \delta^3 }{(1-\delta ) 
   % (1+5\delta+2\delta^2)}
   % -
   % \frac{(2A_1/3)^k}{2} 
   % \left( 1 - 18 A_2 \delta \right) ,
\end{align} 
where 
\begin{align}
   A_1 &=  
   \frac{ 13 - 3\delta -
   6\delta^2  }
   {
   3 + 6\delta - 3\delta^2  - 2\delta^3 }, 
   % \frac{ 7+  \delta  (5-4 \delta  (\delta +1))}
   % {3+\delta  (\delta  (4 \delta  (\delta +1)-7)-8)}.
\end{align}
  To leading order in $\delta$, we have $A_1 = \frac{13}{3}
   -
   \frac{29 \delta  }{3} + O(\delta^2)$. 
In the limit of infinite $k$, we have
\begin{align}
\lim_{ k \to \infty }
    {{{b}_{k} } }
=
1
-\frac{13 \delta}{3} 
+\frac{29 \delta^2}{3} +O(\delta ^3)
\end{align}
which approaches the HL in the sense that it gives the HL in the small $\delta$ limit. 
The corresponding scaling for the shifted gnu code's gap on the $k$th iterate is therefore 
\begin{align}
    g = \Theta(
    N^{1 + O(\delta)
    -(2/5 + O(\delta))^{k}(1/2+O(\delta)) } 
    ) .
\end{align}
In the limit of small $\delta$, we roughly have 
$
g=\Theta(N^{4/5}),
g=\Theta(N^{23/25}),
g=\Theta(N^{121/125})$ for $k=1,2,3$ respectively.

When $\delta$ is too large ($\delta \approx 0.156$), we will unfortunately lose the quantum advantage, so it is important to choose $\delta$ to be sufficiently small. However if $\delta$ is too small, the probability of the protocol failing becomes large.
The failure of the protocol because of $\delta$ can happen from two places as quantified above. First the random shift $\sigma$ in the Dicke weights during each round of Stage 1 could drift too far from half of the number of deletions, thereby making the amplitude shift larger than expected. Second, probability that the number of amplitude rebalancing steps in Stage 2 takes no more than $v$ steps decreases as $\delta$ decreases.

Figure \ref{fig:iteration} illustrates how our estimation of $\theta$ has standard deviation that approaches the Heisenberg scaling exponentially fast in the number of iterations within \texttt{ECSense}, in spite of the deletion errors.

Consider a simplified variation of the \texttt{ECSense} protcol where we do not perform Stage 2, and we run only one round in Stage 1 before projection onto the logical plus-minus basis, and repeating this $r$ times so as to take roughly the same time as our \texttt{ECSense} protocol. Then each run of this protocol yields an expected FI of $\Theta(g^6\tau^6)$, 
and therefore this overall protocol’s total FI scales as $\Theta(rg^6\tau^6)$. Since $r$ increases with $N$, \texttt{ECSense}'s FI of $ F = \Theta(r^2g^6\tau^6)$ delivers a growing precision as compared with the trivial protocol as $N$ becomes large.

 \begin{figure}[t!]
\centering  
\includegraphics[width=0.48\textwidth]{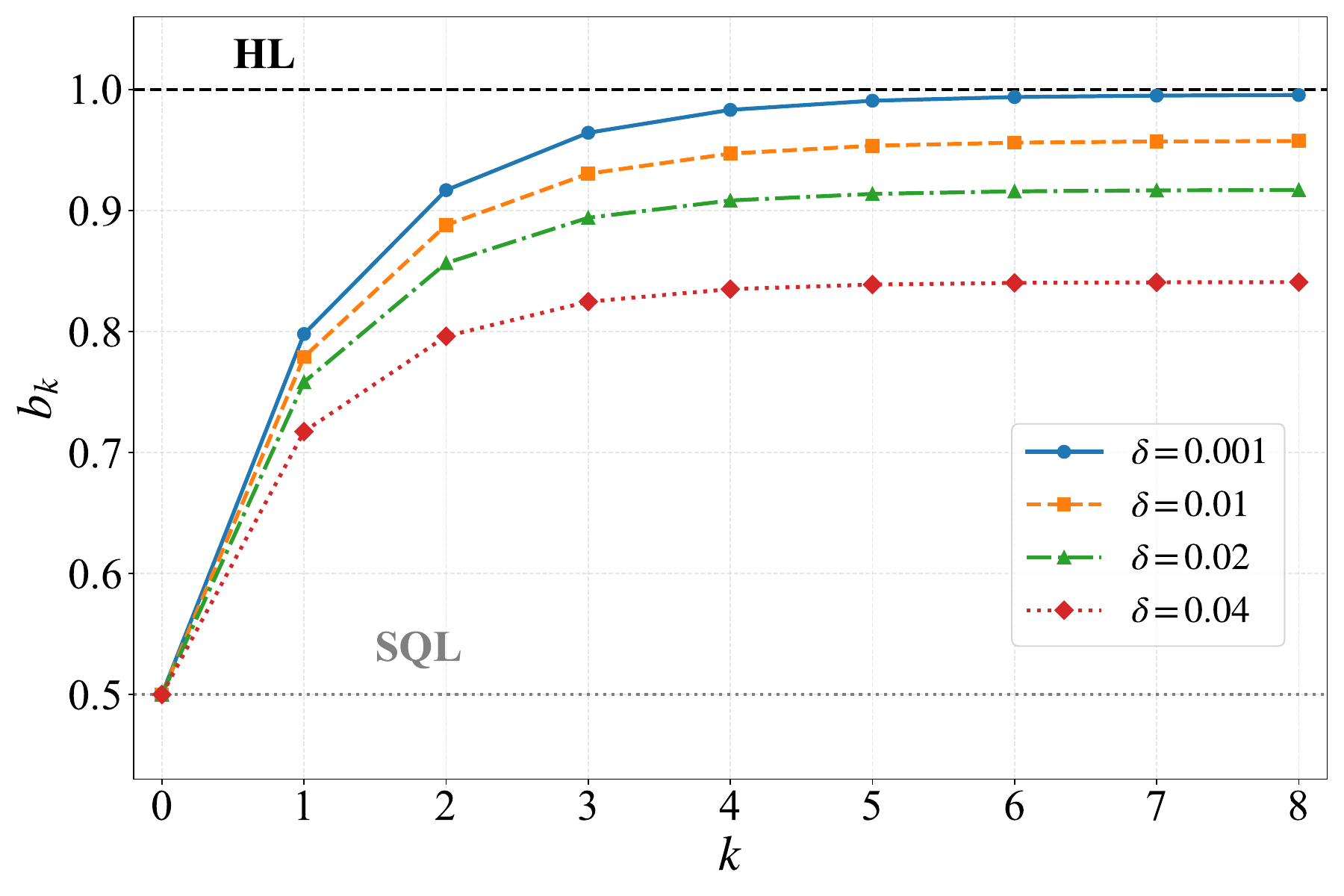} 
 \caption{
We can approach the HL exponentially fast in the number of iterations $k$ within \texttt{ECSense}.
At the $k$th iteration, we
estimate $\theta$ with standard deviation 
  ${{{B}}}_k = \Theta(N^{- {{{b}}}_k} )$.
The SQL corresponds to estimating 
$\theta$ with 
 ${{{b}}}_k =   1/2  $,
 which is achievable using classical techniques.
The HL corresponds to estimating 
$\theta$ with 
${{{b}}}_k = 1$,
which is the best possible precision using quantum techniques. 
% {Note that $\delta$ determines how close we can approach the Heisenberg scaling. The smaller the $\delta$, the larger the value of $N$ needs to be before we begin to see the illustrated advantages.}
}\label{fig:iteration}
 \end{figure}
 
\section{Complexity of \texttt{ECSense}}
\label{complexity}
{
The method of QEC used in our protocol is somewhat flexible give the resources available. Because we are using permutation invariant codes, and the target signal evolution is assumed uniform over the probe, our operations naturally inherit permutation symmetry over the qubits. 
A general purpose strategy uses either the Quantum Schur Transform (QST) or quantum teleportation for the recovery operation after extracting SYT syndromes to extract error syndrome data \cite{newpaper}. The QST on qubits has gate complexity $O(N \log(N/\epsilon)$ in the Clifford+$T$ computational model, where $\epsilon$ is the precision of performing the controlled-rotations part of the algorithm according to Solovay-Kitaev types of arguments. When the number of errors is $t$, the gate complexity of the quantum Schur transform can further reduce to $O(t \log(N/\epsilon)$, which can be advantageous when $t$ is sublinear in $N$, which we will see is the case if we are to recover a quantum advantage in sensing in the setting of errors first accumulating before noiseless signal accumulation.}
 
{
% The total number of control operations necessary 
For teleportation-based QEC \cite{newpaper},
% is enumerated in Sec.~\ref{ssec:teleport}.
the unitary operations require $2N+9$ geometric phase gates (GPG)s and $\lceil{4N/3\rceil}+3$ transversal spin rotations, and the measurements involve $N$ measurements of a bosonic mode. None of the operations need to address individual spins but it is assumed that the ensemble of control qubits can be distinguished from target ensemble.}

%%%%%%%%%%%
\section{Discussions} 
We have presented a QEC-enhanced quantum sensing protocol that recovers the quantum advantage in quantum sensing, even when there is a constant rate of deletion errors. 
Since one may carry out our protocol using near-term quantum control techniques as we describe in Section \ref{sec:Methods},
near-term QEC-enhanced quantum sensing becomes a compelling possibility.

Imperfect QEC steps unavoidably degrade the sensing performance, but quantifying the impact of imperfect QEC and mitigating these effects remains an open problem in the nascent field of fault-tolerant quantum sensing \cite{FT-quantum-sensing}. Our protocol is the first to use QEC to correct deletion errors in a metrological setting before full fault-tolerance. 
Analysis of the detailed tradeoff between resources required and performance of a suitably modified variant of our protocol would constitute developing a fault-tolerance theory for quantum sensing in the presence of deletions, 
and this would help pave the path towards achieving a practical quantum advantage for sensing in near-term devices.

The performance of GPGs, used for unitary operations in the teleportation-based QEC, in the presence of error has been extensively studied in Refs.~\cite{PhysRevA.110.062610,PhysRevResearch.7.L022072,Srivastava_2026}. In the context where the bosonic mode is a microwave or optical cavity mode, the infidelity is dominated by errors due to amplitude damping at rate $\gamma$ and cavity decay at rate $\kappa$, leading to a process infidelity that is upper bounded by $\frac{N\pi}{2\sqrt{2(1+2^{-N})C}}$, where the cooperativity is $C=g_1^2/\gamma\kappa$, with $g_1$ the single qubit-cavity coupling rate. In the presence of inhomogenous couplings of the spins, modeled as quenched random deviations of each spin's coupling from the average coupling $\bar{g_1}$ and with variance ${\rm Var}[g_1^2]$, the expected infidelity is bounded as $\mathbb{E}[1-f]\leq N^2(N+3)\pi^2{\rm Var}[g_1^2]/8\bar{g_1}^4$. 
Recent experiments have demonstrated \cite{CavityExp} 
entangling gates between 
$^{87}$Rb atoms trapped
in optical tweezers and coupled to a fiber Fabry-Perot
cavity using essentially the same mechanism as that needed for the GPGs discussed here. While the cooperativity obtained using those optical cavities is rather low at $C\sim 25$, higher values, $C\sim 1500$, are achievable (see Ref.\cite{PhysRevA.110.062610} and references therein). Much higher values are obtainable with platforms utilizing circular Rydberg transitions coupled to a microwave cavity. Toward this goal high-finesse microwave resonators have been built that predict a single-particle cooperativity of $C =6.75\times10^5$ \cite{zhang2025opticallyaccessiblehighfinessemillimeterwave}. Another promising platform is superconducting fluxonium qubits coupled to a driven microwave resonator~\cite{PhysRevA.101.022321,PhysRevA.110.062610}.

% {In view of the promising theoretical performance of permutation-invariant codes, one might wonder what other codes could have a similar performance, at least with respect to erasure errors. For erasure errors, the HNLS condition is satisfied, and hence the HL can in principle be achieved. We speculate that the family of quantum Reed-Muller codes are a promising family of quantum stabilizer codes to study for the following two reasons. First, quantum Reed-Muller codes have been studied in a fault-tolerant quantum sensing setting \cite{FT-quantum-sensing}. Second, quantum Reed-Muller codes have favorable precision in field estimation in the robust quantum metrology setting \cite{ouyang2020weight}. We furthermore speculate that quantum CSS codes with classical codes with a sparse weight distribution spectrum will be favorable to study.}

{Other} interesting avenues for future work include optimizing our protocol for a fixed initial number of qubits,
or extending our results to the simultaneous estimation of all three components of classical fields, using recent developments in multiparameter quantum metrology \cite{sidhu2021tight,hayashi2023tight,hayashi2024finding}. 
We also leave the investigation of other noise models for QEC-enhanced metrology for future work.

\section{Acknowledgements}\label{eq:acknow}

Y.O. acknowledges support from EPSRC Grant No. EP/W028115/1 and also the EPSRC funded QCI3 Hub under Grant No. EP/Z53318X/1.   Y.O. also acknowledges the Quantum Engineering Programme grant NRF2021-QEP2-01-P06, and the NUS startup grants (R-263-000-E32-133 and R-263-000-E32-731).
G.K.B. acknowledges support from the Australian Research Council Centre of Excellence for Engineered Quantum Systems (Grant No. CE 170100009).

%%%%%%%%%%%%%%% 
\section{Methods}\label{sec:Methods}
%%%%%%%%%%%%%%%  

\subsection{Estimating the field parameter $\theta$}

When we estimate $\theta$, we measure in the logical plus-minus basis of the appropriate shifted gnu code. 
That is, we will project the state onto either
$|a_+\rangle = \frac{1}{\sqrt 2} (|a_0\rangle  + |a_1\rangle ) $
or 
 $|a_-\rangle = \frac{1}{\sqrt 2} (|a_0\rangle  - |a_1\rangle ) $,
 where $|a_0\rangle $ and $|a_1\rangle $ denote the logical zero and logical one state of the appropriate shifted gnu code.
The amount of Fisher Information (FI) that we extract will depend on the total phase that the signal has accumulated within the shifted gnu code's codespace.

Interpreting the total accumulated phase $\Phi$ as a continuous function of $\theta$,
the state that we wish to measure is
$|\psi(\Phi)\rangle  = 
\cos \phi |a_0\rangle  + e^{i \Phi}\sin\phi |a_1\rangle $.
Recall that the initial probe state of ECSense is a logical plus shifted gnu state, which we can write as 
$\frac{ |a_0\rangle  + |a_1\rangle  }{\sqrt 2}$, corresponding to a total phase of 0. 
We show in the \app that the FI of $\theta$ by performing projective measurements on $|\psi(\Phi)\rangle $ with respect to the projectors $|a_+\rangle \langle a_+|$ and $|a_-\rangle \langle a_-|$ is 
\begin{align}
  F  =   \frac{\sin^2 2\phi \sin^2 \Phi}
    {    (1-\sin^2 2\phi \cos^2 \Phi)}
    \left(\frac{\partial\Phi}{\partial\theta}\right)^2.
\end{align} 
We can use this result by noting that the quantum state just before measurement has the form
$a |0_L\rangle  + b e^{i\Phi} |1_L\rangle $, where $a = \cos \phi$ and $b = \sin \phi$, where the amplitudes $a$ and $b$ are close to $1/\sqrt 2$.
Using information about only the phase $\Phi$ and not the amplitudes $a$ and $b$, we find that the Fisher information of ECSense is proportional to the square of the partial derivative of the total phase with respect to $\theta$. 
Hence, by evaluating the expected total phase, we can determine the expected Fisher information. Full details of the calculations are given in the \app.

\subsection{FI of evolved shifted gnu states}
\label{ssec:pisensors}
Here, Lemma \ref{lem:gnus-qfi} calculates how much QFI we can get if we use $|+_{g,n,u,s}\rangle $ as a probe state to estimate $\theta$ in the noiseless setting. 
 
\begin{lemma}
\label{lem:gnus-qfi}
The QFI of $|+_{g,n,u,s}\rangle $ with respect to the signal $U_\theta$ is $g^2 n$.
\end{lemma}
\begin{proof}
The QFI is four times of the variance of the state $|+_{g,n,u,s}\rangle $.
Note that 
\begin{align}
\sum_{k =0}^n \binom  n k k   &=2^{n-1} n ,\quad
\sum_{k =0}^n \binom  n k k^2 = 2^{n-2} n(n+1).
\end{align}
 
Hence the variance of $|+_{g,n,u,s}\rangle $ is
\begin{align}
&
\frac{1}{2^n}\sum_{k =0}^n \binom  n k (g^2k^2 +2 g k s + s^2) 
-
\bigl(\frac{1}{2^n}\sum_{k =0}^n \binom  n k (gk+s) \bigr)^2.\notag
\end{align}
Using binomial identities, this variance is $
\frac{g^2 n(n+1)}{4}
+ \frac{2gs n }{2}
+ s^2
-
\left( \frac{g n}{2} +  s \right)^2
=\frac{g^2 n}{4}.$ The result follows.
\end{proof}
 
Theorem \ref{thm:signal-on-code} shows that the FI can be proportional to $g^2$ by measuring an evolved gnu probe state in the code basis, spanned by the logical plus and minus operators.

\begin{theorem}
\label{thm:signal-on-code}
Consider a shifted gnu state $\sum_{w}a_w |D^N_w\rangle $
with parameters $g,n,u,s$, where $a_{gk+s} = \sqrt{\binom n k}2^{-n/2}$ for $k=0,1,\dots,n$ and $a_w = 0$ for all $w$ that cannot be written as $gk+s$.
After $U_\theta$ applies on the probe state, the FI of estimating $\theta$ by measuring in the gnu code's logical plus-minus basis is 
\begin{align}
 g^2 n^2
\left(
\sin^2(g \theta/2) 
\cos^{2n-2}(g \theta/2) 
+
\cos^2(g \theta\yo{/2}) 
\sin^{2n-2}(g \theta/2) 
\right).\notag
\end{align}
\end{theorem}
 We prove this theorem in the \app.

  \begin{figure}[H]
 \centering  
\includegraphics[width=0.45\textwidth]{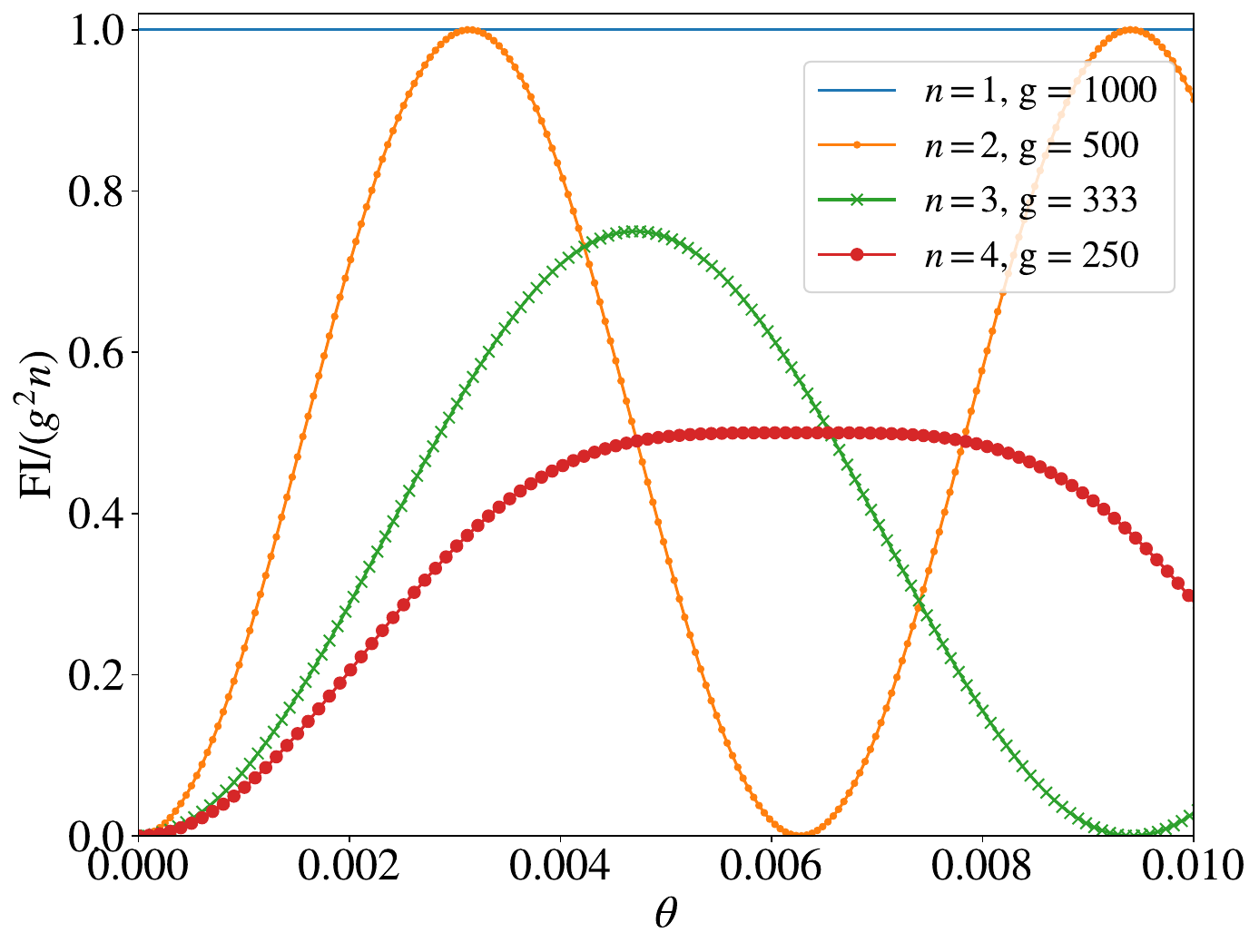} 
 \caption{{\bf FI and QFI of shifted gnu codes on 1000 qubits in the absence of noise.} Here, ($u=1$). The QFI for shifted gnu codes is $g^2n$. The FI is obtained by measuring the evolved shifted gnu states in the code's logical plus-minus basis. The plot shows how the ratio of the FI to the QFI depends on the true value of $\theta$.  For reference, the 1000-qubit GHZ state corresponds to $g=1000,n=1$.
 }\label{fig:FIcodebasis}
 \end{figure}
Using Theorem \ref{thm:signal-on-code}, we can compute the performance of the FI obtained by measuring evolved shifted gnu states on the logical plus-minus basis, and compare this FI to the QFI in Fig.~\ref{fig:FIcodebasis}. Fig.~\ref{fig:FIcodebasis} shows that the FI can be within a constant factor of the QFI, and also the tradeoff between increasing the distance of the code $d=\min\{g,n\}$ and the ratio of the FI to the QFI. Note that Theorem \ref{thm:signal-on-code} also shows that the FI is exponentially suppressed with increasing $n$, and for this reason, we focus our attention on shifted gnu codes with constant $n$.

When shifted gnu codes accumulate the signal given by the unitary $U_\theta$, they transform into what we call 
{\em $\theta$-rotated and $s$-shifted gnu codes}.
Such codes have logical codewords
\begin{align}
|0_{g,n,u,s,\theta}\rangle  & = U_\theta |0_{g,n,u,s}\rangle  ,
% \notag\\
% &= 2^{-(n-1)/2} \sum_{k {\ \rm even}} 
% e^{i(gk+s)\theta} \sqrt{\binom n k}  |D^{gnu}_{gk+s}\rangle 
\\
|1_{g,n,u,s,\theta}\rangle   &= U_\theta |1_{g,n,u,s}\rangle  .
% \notag\\
% &= 2^{-(n-1)/2} \sum_{k {\ \rm odd}} e^{i(gk+l)\theta} \sqrt{\binom n k}  |D^{gnu}_{gk+s}\rangle .
\end{align}
Using the same proof technique as in \cite{ouyang2014permutation}, we can see that these $\theta$-rotated codes also have a distance of ${\rm min}(g,n)$.
This is because for any multi-qubit Pauli operator $P$ that acts on at most $\min(g,n)-1$ qubits, we can easily check that the Knill-Laflamme quantum error criterion \cite{KnL97} for shifted gnu codes is equivalent to that for $\theta$-rotated codes:
\begin{align}
\langle 0_{g,n,u,s,\theta}|P|0_{g,n,u,s,\theta}\rangle  &= \langle 0_{g,n,u,s}|P|0_{g,n,u,s}\rangle  \\
\langle 1_{g,n,u,s,\theta}|P|1_{g,n,u,s,\theta}\rangle  &= \langle 1_{g,n,u,s}|P|1_{g,n,u,s}\rangle  \\
\langle 1_{g,n,u,s,\theta}|P|0_{g,n,u,s,\theta}\rangle  &= 0 .
\end{align}
The QEC properties of these $\theta$-rotated gnu codes is invariant of $\theta$.
It is the distinguishability of these codes with respect to $\theta$ that makes them useful as probe states for classical field sensing.

\appendix

\section{Symmetric probe states for field sensing}
\label{sec:field-sensing}
%%%%%%%%%%%%%%%%% 
When $\rho_\theta$ is pure symmetric state used for noiseless field-sensing and ${\frac{d\rho}{d\theta} =  -i [{\hat J^z},\rho]}$, we derive a corresponding rank two SLD and find its spectral decomposition in Theorem \ref{thm:SLD-pure-state}.
%%%%%%%%%%%%%%%%%%%%%%
\begin{theorem}[SLD for pure symmetric states]
\label{thm:SLD-pure-state}
Let $\rho_\theta = |\psi\>\<\psi|$ be a pure symmetric state where 
${|\psi\> = \sum_w a_w |D^N_w\>}$.
For $j=1,2$, let ${m_j} = \sum_{w=1}^N  |a_w|^2 w^j,$ and let 
${v_{\psi}} = {m_2} - m_1^2.$
Then a solution to $    \frac{d\rho_\theta }{d \theta}
    =\frac{1}{2}(L \rho_\theta +  \rho_\theta L)$ has the spectral decomposition
\begin{align}
    L =& 
      \sqrt{{v_{\psi}}} (|\psi\>+i|b\>) (\<\psi|-i\<b|) \notag\\
    &- \sqrt{{v_{\psi}}} (|\psi\>-i|b\>) (\<\psi|+i\<b|)
\end{align}
where 
% \begin{align}
$|b\> = \left(\sum_w a_w w |D^N_w\> - {m_1} |\psi\>\right) / \sqrt{{v_{\psi}}},$
% \end{align}
and $\{|\psi\>, | b\> \}$ is an orthonormal basis.
\end{theorem}

\begin{figure}[!h]
 \centering  
\includegraphics[width=\columnwidth]{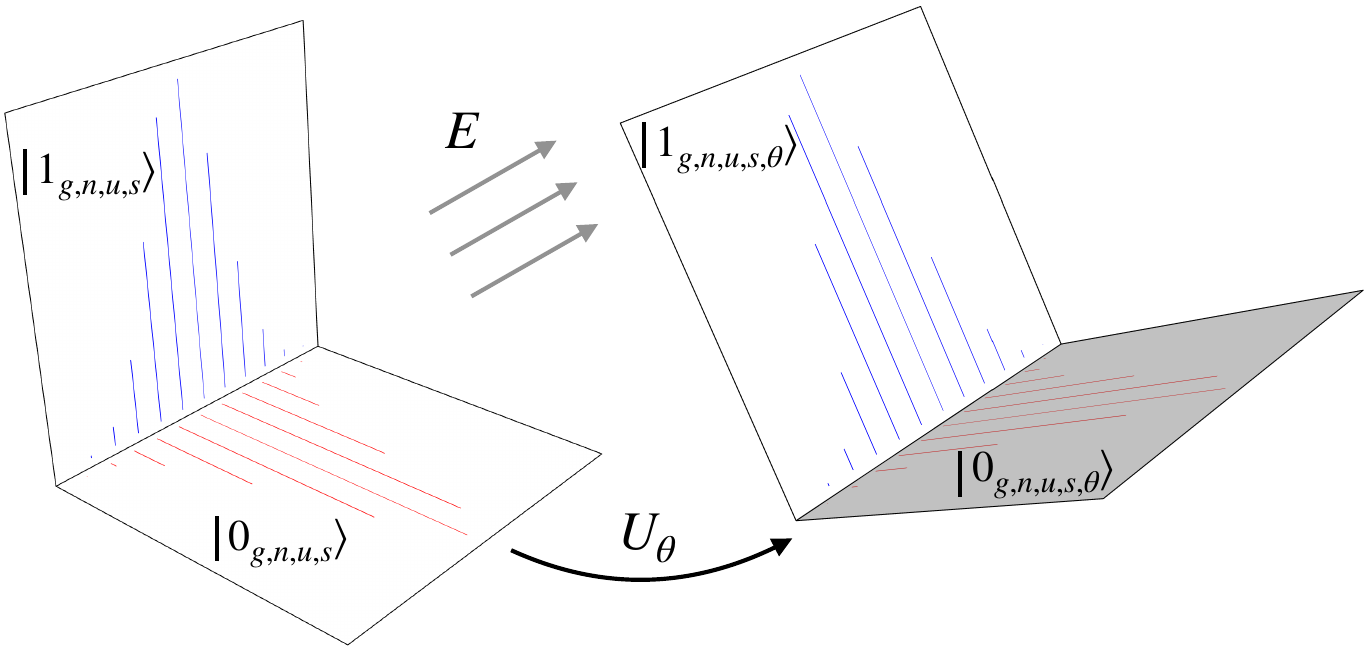} 
 \caption{Quantum sensing with the shifted gnu code to infer the strength of a field $E$ which acts uniformly on the spins. 
 We depict the amplitudes of the logical basis states in the Dicke basis.
A unitary $U_{\theta}=e^{-i\theta \hat{J}^z}$ rotates the logical basis to a new logical frame, which satisfies the same error correction criteria. Here, we estimate $\theta$.  
 }\label{fig:Sensing}
 \end{figure} 

%%%%%%%%%%%%%%%%%%%%%%%
Based on Theorem \ref{thm:SLD-pure-state}, since the QFI is $\<\psi|L^2|\psi\>$, we can ascertain that for symmetric states, the QFI is indeed ${4v_{\psi}}$ as expected.
Hence, for any pure symmetric state $|\psi\>$, and when 
$\rho_\theta = U_\theta|\psi\>\<\psi| U_\theta^\dagger$,
we can write
\begin{align}
\Var(\hat \theta) \ge Q(\rho_\theta, \frac{d\rho_\theta}{d\theta} )^{-1}
=1/(4{v_{\psi}}).
\end{align}
In order to achieve the precision given by the QFI, the optimal observable is given by the SLD $L$ which we show to have two eigenvectors
\[
\frac{|b_1\> \pm i|\bar b_2\>}{\sqrt 2}=
\frac{|\psi\> \pm i|\bar b_2\>}{\sqrt 2},
\]
with eigenvalues $\pm 2 \sqrt {v_{\psi}}$.
This can be realized by (1) performing any unitary extension of the mapping: $\frac{|b_1\> \pm i|\bar b_2\>}{\sqrt 2}\rightarrow|D^N_0\>$, (2) 
measuring $\hat{J}^z$ with outcome $m$, and (3) weighting the result with $2 {v_{\psi}}$ for the outcome $m=-N/2$ and $-2 {v_{\psi}}$ otherwise. 
The state synthesis mapping can be done using $N$ GPGs as described in the Methods section of our main manuscript.

\begin{proof}[Proof of Theorem \ref{thm:SLD-pure-state}]
Now 
\begin{align}
{\hat J^z} | D^N_{w}\> =(N/2-w) |D^N_{w}\>.
\end{align}
Hence, for $|\psi\> = \sum_w a_w |D^N_w\>$,
we have 
\begin{align}
{\hat J^z}|\psi\> &= \sum_w a_w(N/2 -w) |D^N_w\> 
\notag\\
&= (N/2)|\psi\> -  \sum_w w a_w |D^N_w\> \label{Jzpsi}.
\end{align}
For us, $\rho_\theta  = |\psi\>\<\psi|$ is a pure state. Hence a feasible solution of $L$ to 
$    \frac{d\rho_\theta }{d \theta}
    =\frac{1}{2}(L \rho_\theta +  \rho_\theta L)$. 
 is $L  =    2i (|\psi\>\<\psi| {\hat J^z}  -  {\hat J^z}|\psi\>\<\psi|)$ \cite{toth2014quantum}. Using \eqref{Jzpsi}, we can express $L$ as
\begin{align}
L  
    &= 
    2i  |\psi\>(\<\psi| (N/2) - \sum_{w'} w' a_{w'}^* \<D^N_w| ) \notag\\
    &- 2i((N/2)|\psi\> - \sum_w w a_w |D^N_w\>)\<\psi|)\notag\\
    &= 
    -2i  \sum_{w'} w' a_{w'}^* |\psi\>\<D^N_w|   +2i \sum_w w a_w |D^N_w\>\<\psi|\notag\\
    &= 
    i \sum_{w,w'} a_w a_{w'}^* |D_w\>\<D_{w'}| (2w-2w').
    \label{eq:SLD-from-commutator}
\end{align}
Note the simple identity
\begin{align}
    (u+i)(v-i)-(u-i)(v+i)
    =2i (v-u).
\end{align}
This implies that by setting $a(u) = u+i$, we have that
\begin{align}
a(u)a(v)^* - a(u)^*a(v)  =2i( v-u ).\label{eq:linear-to-auav}
\end{align}
for real $u$ and $v$.
Using \eqref{eq:linear-to-auav} in \eqref{eq:SLD-from-commutator}, we get
\begin{align}
L &= 
  \sum_{w,w'}a_{w} a_{w'}^* 
      a(w)a(w')^*
    |D^N_w\>\<D^N_{w'}| \notag\\
    &\quad
     -
       \sum_{w,w'}a_{w} a_{w'}^*
    a(w)^*a(w')
    |D^N_w\>\<D^N_{w'}|   .
\end{align}
Let 
\begin{align}
|v_1\> &= \sum_w a_w (w+i)|D^N_w\>,\\
|v_2\> &= \sum_w a_w (w-i)|D^N_w\>.
\end{align}
Substituting the definitions of $|v_1\>$ and $|v_2\>$, we find 
\begin{align}
    L = 
      |v_1\>\<v_1|
    -  |v_2\>\<v_2|.
\end{align}
While $|v_1\>$ and $|v_2\>$ are not necessarily orthorgonal, the vector 
\begin{align}
|v'_2\> 
=  |v_2\> - 
\<v_1 |v_2\>
\frac{|v_1\>}{\<v_1|v_1\>}
\end{align}
is orthogonal to $|v_1\>$.
Now $|v_1\>$ and $|v_2\>$ lie in the span of 
$|b_1\>$ and $|b_2\>$ where
\begin{align}
|b_1\> &= \sum_w a_w |D^N_w\>\\
|b_2\> &=  \sum_w w a_w |D^N_w\>.
\end{align}
By definition $\<b_1|b_1\> = 1$. Using Gram-Schmidt we see that
\begin{align}
|b_2'\> 
&= |b_2\> -  \frac{|b_1\>\<b_1| } {\< b_1| b_1\>  }  |b_2\>\notag\\
&= |b_2\> - {m_1} |b_1\>.
\end{align}
Hence $|b_2\> = |b_2'\> + {m_1}|b_1\>$.
Also note that 
\begin{align}
\<b_2'|b_2'\>
&= \<b_2|b_2\> - {m_1} \<b_2|b_1\> - {m_1} \<b_1| b_2\> + {m_1}^2 \<b_1|b_1\>\notag\\
& = {m_2} - {m_1}^2  = {v_{\psi}}.
\end{align} 
Now let $|\bar b_2\> = |b_2'\>/\sqrt{\<b_2'|b_2'\>}$ so that $\{|b_1\> , |\bar b_2\> \}$ is an orthonormal basis.
In this orthonormal basis, we can write
\begin{align}
|v_1\> 
&= i |b_1\> + |b_2\> \notag\\
&= i |b_1\> + |b_2'\> + {m_1}|b_1\> \notag\\
&= ({m_1} + i ) |b_1\> + \sqrt{{v_{\psi}}}|\bar b_2\>
\end{align}
Similarly, we get
\begin{align}
|v_2\> 
&= i |b_1\> + |b_2\> \notag\\
&= i |b_1\> + |b_2'\> + {m_1}|b_1\> \notag\\
&= ({m_1} - i ) |b_1\> + \sqrt{{v_{\psi}}}|\bar b_2\>
\end{align}
Hence on the orthonormal basis $\{|b_1\> , |\bar b_2\> \}$, the SLD is 
an effective size two matrix
\begin{align}
M =& \begin{pmatrix}
|{m_1}+i|^2 & ({m_1}+i)\sqrt{{v_{\psi}}} \\
({m_1}-i)\sqrt{{v_{\psi}}} & {v_{\psi}} \\
\end{pmatrix}\notag\\
&-
 \begin{pmatrix}
|{m_1}-i|^2 & ({m_1}-i)\sqrt{{v_{\psi}}} \\
({m_1}+i)\sqrt{{v_{\psi}}} & {v_{\psi}} \\
\end{pmatrix}\notag\\
&=
 2 \sqrt{{v_{\psi}}} \begin{pmatrix}
0 & -i \\
i & 0 \\
\end{pmatrix}
.\label{M-matrix}
\end{align}
Thus the eigenvalues of the SLD operator are given by $\pm 2 \sqrt{v_{\psi}}$ with corresponding eigenvectors 
\begin{align}
\frac{|b_1\> \pm i|\bar b_2\>}{\sqrt 2}=
\frac{|\psi\> \pm i|\bar b_2\>}{\sqrt 2}.
\end{align}
\end{proof}

%  \begin{figure}[H]
%  \centering  
% \includegraphics[width=\columnwidth]{graphics/Fig1.pdf} 
%  \caption{Illustration of a gnu state $|+_{g,n,u,s}\>$ where
%  $g=21$, $n=2\lfloor g/2\rfloor+1$, $u=1+1/n$, and $s = g$ so that the number of qubits is $N=gn+2g=483$. The horizontal axis depicts the weights $w$ of the Dicke states, and the vertical axis depicts the values of amplitude $\<D^{gnu+s}_w|+_{g,n,u,s}\>$. The colors red and blue correspond to the Dicke states of $|0_{g,n,u,s}\>$ and $|1_{g,n,u,s}\>$ respectively.
%  }\label{fig:gnucode_bar}
%  \end{figure}

\subsection{Impact of errors on uncorrected symmetric states.}
In the following we derive the QFI for shifted gnu states affected by either deletion errors or amplitude damping errors. The effect of dephasing errors on gnu codes is described in Ref. \cite{ouyang2019robust}.
\label{ssec:error-impact}
\subsubsection{Deletion errors}
The action of $t$ deletion errors on any symmetric state, is equivalent to taking the partial trace $\tr_t(\cdot)$ on the first $t$ qubits.
To discuss the action of deletion errors, it is often convenient to use the unnormalized Dicke states
\begin{align}
|H^N_w\> = \sum_{ \substack{{\bf x} \in \{0,1\}^N \\ \wt({\bf x}) = w} } |x_1\>\otimes \dots \otimes |x_N\>,\label{eq:unnormalized-dicke-defi}
\end{align}
where $\wt(\bx)$ denotes the Hamming weight of a binary string. 
Normalizing the basis states in \eqref{eq:unnormalized-dicke-defi} gives Dicke states $|D^N_w\> = |H^N_w\> / \sqrt{\binom N w}$. 
Then, Lemma \ref{lem:delete} gives the partial trace of any symmetric state.
\begin{lemma}[Impact of deletions]\label{lem:delete}
Let $N$ be a positive integer, and let $v,w$ be integers such that $0\le v,w \le N$. 
 Let $t$ be a positive integer where $t\le N$.
% Let $A_w = \{ a : \max\{ 0, t+w - N \}\le  a \le \min \{ w,t \}  \}$.
Let $|\psi\> = \sum_{w=0}^N a_w |D^N_w\>,$
and for all $a = 0, \dots, t$, define
\begin{align}
%|\psi\>_a = \sum_{w=0}^N a_w \delta[a \in A_w] |H^{N-t}_{w-a}\>.
|\psi\>_a = \sum_{w=a}^{N-t+a} a_w  
\frac{
\sqrt{\binom {N-t}{w-a}}
}
{
\sqrt{\binom {N}{w}}
}
|D^{N-t}_{w-a}\>.
\label{eq:psia-delete}
\end{align}
 Then 
 \begin{align}
 \tr_t( | \psi\> \<\psi| )
 =
 \sum_{a = 0}^t 
 \binom t a 
 |\psi\>_a \<\psi|_a. \label{partial_trace:decomposition}
 \end{align}
\end{lemma}
Lemma \ref{lem:delete} generalizes the result in \cite[Lemma 5]{ouyang2021permutation} where only shifted gnu codes were considered. 

If $|\psi\>$ is a shifted gnu state, we have 
\begin{align}
|\psi\>_a = 2^{-n/2}\sum_{j=0}^n
\sqrt{\binom{n}{j}  \frac{ \binom {N-t}{gj+s-a} }{ \binom N {gj+s}  } } 
|D^{N-t}_{gj+s-a} \>,
\end{align}
and hence its norm squared $n_a = \<\psi_a| \psi_a\>$ is
\begin{align}
n_ a  = 
2^{-n}\sum_{j=0}^n
\binom{n}{j}  \frac{ \binom {N-t}{gj+s-a} }{ \binom N {gj+s}  }  ,
\end{align}
and the variance of $|\psi_a \> / \sqrt {n_a}$ is
\begin{align}
v_a = &
2^{-n}\sum_{j=0}^n
\binom{n}{j}  \frac{ \binom {N-t}{gj+s-a} }{n_a \binom N {gj+s}  } (gj+s-a)^2 \notag\\
&-
\left(
2^{-n}\sum_{j=0}^n
\binom{n}{j}  \frac{ \binom {N-t}{gj+s-a} }{n_a \binom N {gj+s}  }  (gj+s-a)
\right)^2.
\end{align}
Thus, we can determine the QFI for shifted gnu states.
\begin{lemma}
\label{lem:gnu-delete-qfi}
The QFI of $|+_{g,n,u,s}\>$ after $t$ deletions, when $g,n \ge t+1$, is 
$4\sum_{a =0}^t  n_a \binom t a  v_a.$
\end{lemma}
\begin{proof}
After $t$ deletions, as long as $g,n \ge t+1$, the vectors $|\psi\>_a$ are pairwise orthogonal, and furthermore are supported on Dicke states of distinct weights modulo $g$. From Theorem \ref{thm:SLD-pure-state}, we can see that the SLDs of each of the states $|\psi\>_a/ \sqrt{n_a}$ is pairwise orthogonal.
This implies that the QFI of $|+_{g,n,u,s}\>$ after $t$ deletions is a convex combination of the QFIs of the $|\psi\>_a/ \sqrt{n_a}$, and hence is given by $4
\sum_{a =0}^t  n_a \binom t a  v_a.$
\end{proof}

\subsubsection{Amplitude damping errors}
Amplitude damping (AD) errors are introduced by an AD channel $\mathcal A_\gamma$ which has Kraus operators ${A_0 = |0\>\<0| + \sqrt{1-\gamma} |1\>\<1|}$ and 
${A_1 = \sqrt \gamma |0\>\<1|}$. These Kraus operators model the relaxation of an excited state to the ground state with probability $\gamma$. We denote an amplitude damping channel on $N$ qubits by $\mathcal A_{N,\gamma}  = \mathcal A_\gamma^{\otimes N}$, which has the Kraus operators $A_{\bf x} = A_{x_1}\otimes \dots \otimes A_{x_N}$ where ${\bf x} = (x_1,\dots, x_N) \in \{0,1\}^N$. 
% we let Let $\wt(\bf x)$ denote the Hamming weight of ${\bf x}$. 

Given a subset $P$ of $\{1,\dots, N\}$, we define an insertion channel ${\rm Ins}_P$ on an $N-|P|$-qubit state to insert the pure state $|0\>^{\otimes |P|}$ in the positions labeled by $P$ to result in an $N$ qubit state. Insertion channels are discussed in more detail for instance in	 \cite{shibayama2021permutation,shibayama2021equivalence}.
Lemma \ref{lem:AD-representation} then expresses any symmetric state after amplitude damping in terms of insertions channels.  
\begin{lemma}[Impact of AD errors]
\label{lem:AD-representation}
Let $|\psi\> = \sum_{w=0}^N a_w |D^N_w\>$ by any symmetric pure state.
Now for any $x=0,\dots,N$, define the subnormalized states
\begin{align}
|\phi_x\> &= \sum_{w=x}^{N-x} a_w \sqrt{p_w(x)} |D^{N-x}_{w-x}\>,
\end{align}
where 
\begin{align}
 p_w(x) = \binom w x \gamma^x (1-\gamma)^{w-x}   \label{eq:pw}.
\end{align}
Then
\begin{align}
    \mathcal A_{N,\gamma}
    (|\psi\>\<\psi|)
    = \sum_{x=0}^N  \frac{1}{\binom N x} \sum_{|P| = x} 
   {\rm Ins}_{P}(   |\phi_x\>\<\phi_x| ) .
\end{align}
\end{lemma}
We prove Lemma \ref{lem:AD-representation} in the appendix.
Let $n_x = \< \phi_x | \phi_x\>$, and let the variance of $|\phi_x\> / \sqrt{n_x} $ be
\begin{align}
q_x =&2^{-n}\sum_{
\substack{
0\le k \le n \\
x\le gk \le N-x
}}
\binom n k \frac{p_w(x)}{n_x} (gk+s-x)^2
\notag\\
&-
\Bigl(
2^{-n}\sum_{
\substack{
0\le k \le n \\
x\le gk \le N-x
}}
\binom n k \frac{p_w(x)}{n_x} (gk+s-x)
\Bigr)^2
\end{align}
Then, we obtain an upper bound on the QFI for shifted gnu states after AD errors.
\begin{lemma}
\label{lem:AD-gnu-qfi}
The QFI of $|+_{g,n,u,s}\>$ after AD errors introduced by $\mathcal A_{N,\gamma}$ is at most $4 \sum_{x=0}^N n_x q_x$.
\end{lemma}
\begin{proof}
The lemma follows directly from the convexity of the QFI with respect to the probe state, and the decomposition in Lemma  \ref{lem:AD-representation}.
\end{proof}

%%%%%%%%%%%%%%%%%%%%%%
\section{QEC before noise-free signal accumulation}
%%%%%%%%%%%%%%%%%%%%%%
\label{sec:integrating}

Now, from Lemma \ref{lem:gnus-qfi}, the QFI of using gnu codes is $g^2n$. Hence, for $t$ general errors, if $t \le (\min\{g,n\}-1)/2$, we can perform QEC to recover this QFI of $g^2n$. 
Hence, when $g = \Theta(N^\alpha)$ and $n = \Theta(N^{1-\alpha})$ where $\alpha \ge 1/2$, the maximum possible QFI for $t = \Omega(N^{1-\alpha})$ is given by the following theorem.

\begin{theorem}
\label{thm:QEC-QFI}
Let $\alpha \ge 1/2$, $ g = \Theta(N^\alpha)$ and $n = \Theta(N^{1-\alpha})$. 
Then any $o(N^{1-\alpha})$ errors can be corrected on a corresponding shifted gnu code, and the corresponding QFI thereafter on $\theta$ for noiseless signal accumulation of $U_\theta$ on $|+_{g,n,u,s}\>$ is $g^2n = \Theta(N^{1+\alpha})$.
\end{theorem}
\begin{proof}
For general errors, since $\alpha \ge 1/2$, the number of correctible errors is 
at most $(\min\{g,n\}-1)/2 = \Theta(N^{1-\alpha})$.
Hence any $o(N^{1-\alpha})$ errors are also correctible.
From Lemma \ref{lem:gnus-qfi}, we know that the QFI of $|+_{g,n,u,s}\>$ is $g^2n$. Since $\alpha \ge 1/2$, we have $g^2n =\Theta(N^{1+\alpha})$.
\end{proof}
Consider Theorem \ref{thm:QEC-QFI} when $\alpha = 1/2$. Then, any $o(\sqrt N)$ errors can be corrected and the QFI is $\Theta(N^{3/2})$.

%%%%%%%%%%%%%%%%%%%%%%
\section{No deletions during signal accumulation}
\label{ssec:QEC-during-sensing}
%%%%%%%%%%%%%%%%%%%%%%

\subsection{Overview}

\yo{Consider the case where $n$ is odd and also at least three.} We consider the set of states
\begin{align}
\left\{ ( (s-N/2) I +  \hat J^z  )^k|j_L\> : 
k=0,\dots, (n-1)/2 \right\}    ,
\end{align}
and denote the orthonormal basis obtained by the Gram Schmidt procedure to be given by
\begin{align}
\left\{ |\psi_{j,k} \> : 
k=0,\dots, (n-1)/2 \right\}   . 
\end{align}
We define the rank 2 projectors 
$    \tilde \Pi_{k} 
    =
    |\psi_{0,k}\>\<\psi_{0,k}|
    +
    |\psi_{1,k}\>\<\psi_{1,k}|,$
and the corresponding spaces as $\mathcal Q_k$.
In our protocol, we measure according to the POVM 
$\tilde M = \{\tilde \Pi_k: k=0,\dots, (n-1)/2 \}
\cup\{I - \sum_{k=0}^{(n-1)/2} \tilde \Pi_k\}$.
With probability 1, the measurement of $\tilde M$ outputs a state in one of the spaces $\mathcal Q_0,\dots, \mathcal Q_{(n-1)/2}$.

Next, for all $k=0,\dots, (n-1)/2$,
we find that 
\begin{align}
\frac
{\<\psi_{1,k} | U_{\Delta} |1_L\>}
{\<\psi_{0,k} | U_{\Delta} |0_L\>}
=e^{i \zeta_k}
\end{align}
for some real number $\zeta_k.$
After projecting onto the space $\mathcal Q_k$ and mapping $\mathcal Q_k$ back to the codespace after the signal $U_\Delta$ acts on the codespace, we obtain an effective evolution 
in the code's logical basis with phase $\zeta_k$.
We can write $\zeta_k = f_k(gb\tau/2)$ where $f_k$ is analytic function.

After obtaining the value of $k$ by measurement, 
we \yo{use a unitary to map} $|\psi_{j,k}\>$ \yo{to} the state $|j_L\>$ for $j=0,1.$

% When $g b\tau/2$ is close to zero, we have 
% $|\zeta_k| = O( (g b\tau)^{n-k} ) $. 
% When we choose $\tau = g^{-1}$ so that $gb\tau/2 = b/2$ is constant, 
% for constant $n$, we have that $f_k(gb\tau/2)$ and $f'_k(gb\tau/2)$ are constant, where $f'(z) = \frac{d}{dz}f(z)$.
% Then, $\frac{d}{db}\zeta_k = (g \tau/2) f'(g b\tau/2)$ which is constant.

% In the Supplemental Information, we also show that 
% when we project the probe state to the second space and perform a unitary operation to bring the state back to the codespace, the state has the form 
% $\xi_0|0_{g,n,u,s}\> +\xi_1 e^{i \zeta_1}|1_{g,n,u,s}\>$ for odd $n$.
% If the probe state is projected onto the third space, we abort Protocol 1 fails, and have a Type 1 error.
% In each successful round where no deletions occur, we find the expected accumulated phase to be approximately $2(1-n) i^{n-1} (g \Delta/2)^n.$ By design, this phase accumulation is small. This is because $g=\Theta(N^{\alpha})$ and $\Delta=br^{-q}$ with $b=\Theta(N^{-c})$, and $r=\Theta(N^{\gamma})$. As described above $c\geq 1/2$, and we will work in a regime where $\gamma>(\alpha-c)/q$, which implies $g\Delta/2=o(1)$. 

\subsection{Details}
   
%%%%%%%

Defining the $N$-qubit operator 
\begin{align}
 \hat K \coloneqq \hat J^z +  (s-N/2)I  + (gn/2) I   
\end{align}
where $I$ denotes the identity operator, and $g,n,s$ are parameters of an $N$-qubit shifted gnu code, we have the following lemma which we prove in Appendix \ref{app:sandwiches}.
\begin{lemma}
\label{lem:sandwich}
Let $|0_L\>$ and $|1_L\>$ be logical codewords of a shifted gnu code with parameters $g,n,u$ and $s$. Then 
$\<j_L| e^{-i \Delta \hat K} |j_L\> = \cos^n(g\Delta/2)  + (-1)^j(-i)^n  \sin^n(g\Delta/2) $ for $j=0,1$.
\end{lemma}

Given any shifted gnu logical codeword $|0_L\>$ and $|1_L\>$, we can define the vectors
\begin{align}
|{Q_0}\> &= {\hat J^z} | 0_L\> - \<0_L | {\hat J^z} | 0_L\>  |0_L\>,
\\
|{Q_1}\> &= {\hat J^z} | 1_L\> - \<1_L | {\hat J^z} | 1_L\>  |1_L\>,
\end{align}
and define $|q_j\> = |Q_j\>/\sqrt{\<Q_j|Q_j\>}$.
The vectors $|Q_j\>$ are orthogonal to $|j_L\>$ by construction.
Note that 
\begin{align}
\<Q_j|Q_j\>=  \< j_L| ({\hat J^z})^2 | j_L\> - \<j_L | {\hat J^z} | j_L\>^2.
\end{align}
When $n \ge 3$ for a shifted gnu code, we have $\<Q_0|Q_0\>=\<Q_1|Q_1\>$. Then for $j=0,1$, we have 
\begin{align}
\<Q_j|Q_j\>=  \< +_L| ({\hat J^z})^2 | +_L\> - \<+_L | {\hat J^z} | +_L\>^2.\label{qj-norm}
\end{align}
The quantity in \eqref{qj-norm} is the variance of a state with respect to the operator $
\hat J^z$. Hence, from Lemma \ref{lem:gnus-qfi},
we can see that 
\begin{align}
\<Q_j|Q_j\>=  g^2n/4.
\label{Q-normalization}
\end{align}
Furthermore, 
\begin{align}
\<Q_j| \hat J_z |j_L\>
&=
  \<j_L|\hat J_z  \hat J_z |j_L\>
  -\<j_L| \hat J_z |j_L\>
  \<j_L|\hat J_z |j_L\>,
\end{align}
which also evaluates to $\frac{g^2n}{4}$.

We introduce the function  
\begin{align}
\phi_{n,j}(\Delta) =& 
 e^{-i \Delta (N/2-s)} e^{ign\Delta/2}
\notag\\
& \times 
\left( \cos^n(g\Delta/2)  + (-1)^j(-i)^n  \sin^n(g\Delta/2) \right).
 \label{phinj-deriv-2}
\end{align}
Next Lemma \ref{lem:qU-norm-sandwich} evaluates the probability of certain projections.
\begin{lemma}
\label{lem:qU-norm-sandwich}
Let $n$ be odd and $\ge 3$. Then 
\begin{align}
%\| |q_j\>\<q_j| U_\Delta |j_L\> \|^2 
| \<q_j| U_\Delta |j_L\> |^2 
&=
n
\left|
 \phi_{n,j}(\Delta) -  e^{ig\Delta}\phi_{n-1,j \oplus 1}(\Delta) \right|^2
 \notag\\
 &=
\frac{n}{4} \sin^2 2x \left(\sin^{2n-4}x+  \cos^{2n-4}x \right) \notag
% \\
%  &=
% \frac{n(g\Delta)^2}{4} + O((g\Delta)^4) 
,
\end{align}
where $x= g \Delta/2$.
\end{lemma}
Now Lemma \ref{lem:pflag} calculates the probability of projections onto various spaces. 
\begin{lemma}
\label{lem:pflag}
Let $n$ be odd, let $|\psi\> = a|0_L\> + b|1_L\>$ be a shifted gnu codestate, and let $|\phi\> = U_\Delta |\psi\>$.   
Let $\Pi$ be the code projector and $\Pi_1$ be a projector onto the space spanned by $|Q_0\>$ and $|Q_1\>$.
Let $x = g \Delta/ 2$.
Then
\begin{align}
\| \Pi  |\phi\> \|^2 =& \cos^{2n} x + \sin^{2n} x,
\label{Pi proj prob}
\\
\| \Pi_1  |\phi\> \|^2 =&
\frac{n}{4} \sin^2 2x \left( \sin^{2n-4}x + \cos^{2n-4}x \right)
\label{Pi1 proj prob}.
\end{align}
Also we have 
\begin{align}
1-\| \Pi |\phi\> \|^2-\| \Pi_0 |\phi\> \|^2
= 
\sum_{k=2}^{n-2} \binom n k \cos^{2k} x \sin^{2n-2k}x ,
\end{align} 
and this expression is $O(x^{2n-2})$.
\end{lemma}
\begin{proof}[Proof of Lemma \ref{lem:pflag}] 
From Lemma \ref{lem:sandwich}, we have
\begin{align}
&\| \Pi |\phi\> \|^2  \notag\\
=&
\| \xi_0|0_L\>\<0_L|U_\Delta |0_L\> + \xi_1|1_L\>\<1_L|U_\Delta |1_L\>\|^2
\notag\\
=&
|\xi_0|^2 |\phi_{n,0}(\Delta)|^2
+
|\xi_1|^2 |\phi_{n,1}(\Delta)|^2
\notag\\
=&
(|\xi_0|^2+|\xi_1|^2) (\cos^{2n} x + \sin^{2n} x)
\notag\\
=&
\cos^{2n} x + \sin^{2n} x.
\end{align}
Similarly, using Lemma \ref{lem:qU-norm-sandwich}, we can see that
\begin{align}
&\| \Pi_1 |\phi\> \|^2  \notag\\
=&
\| \xi_0|q_0\>\<q_0|U_\Delta |0_L\> +  \xi_1|q_1\>\<q_1|U_\Delta |1_L\>\|^2
\notag\\
=&
 |\xi_0|^2 |\<q_0|U_\Delta |0_L\>|^2 +  |\xi_1|^2 |\<q_1|U_\Delta |1_L\>|^2
\notag\\
=&
( |\xi_0|^2 + |\xi_1|^2 ) 
\frac{n}{4} \sin^2 2x \left( \sin^{2n-4}x + \cos^{2n-4}x \right)
\notag\\
=&
\frac{n}{4} \sin^2 2x \left( \sin^{2n-4}x + \cos^{2n-4}x \right).
\end{align}
Now note that 
\begin{align}
&\cos^{2n} x + \sin^{2n} x
+ \frac{n}{4} \sin^2 2x \left( \sin^{2n-4}x + \cos^{2n-4}x \right)
\notag\\
=&
\cos^{2n} x + \sin^{2n} x
+ n  \left( \sin^{2n-2}x\cos^2 x + \cos^{2n-2}x \sin^2 x \right)
\notag\\
=& \sum_{k=0,1,n-1,n} \binom n k \cos^{2k} x \sin^{2n-2k}x \notag\\
=&
(\cos^2 x + \sin^2 x)^n -
 \sum_{k=2}^{n-2} \binom n k \cos^{2k} x \sin^{2n-2k}x \notag\\
 =&
 1- \sum_{k=2}^{n-2} \binom n k \cos^{2k} x \sin^{2n-2k}x .
\end{align}
Hence the result follows.
\end{proof}

From Lemma \ref{lem:sandwich-ratios}
we can determine the effective evolution on the codespace performing projections on $U_\Delta|\psi\>$
where $|\psi\>$ is a shifted gnu state with no deletions. 
%%%%%%%%
\begin{lemma}
\label{lem:sandwich-ratios}
Let $|0_L\>$ and $|1_L\>$ be the logical codewords of a shifted gnu code, 
and let $n$ be odd and $\ge 3$. Let $\Delta$ be a real number.
Then 
\begin{align}
\frac{\<1_L|  U_\Delta  |1_L\>}{ \<0_L|  U_\Delta  |0_L\>} = 
e^{i\zeta_0},
\quad
\frac
{\<{q_1}|U_\Delta|1_L\>}
{\<{q_0}|U_\Delta|0_L\>}
=  
e^{i\zeta_1}  ,
\end{align}
where
\begin{align} 
\zeta_j &= 2\arctan\left( (-1)^j i^{n-1} \tan^{n-2j}(g\Delta/2)  \right).
\label{def:zetas}
\end{align}
\end{lemma}
%%%%%%%%
%%%%%%%%
When $n=3$, we have
\begin{align}
\zeta_0 &= -2\arctan(\tan^3(g\Delta/2) )\\
\zeta_1 &= g\Delta = gb\tau.
\end{align}
From Lemma \ref{lem:sandwich-ratios}, if we project $U_\Delta |\psi\>$ on the codespace where  $|\psi\> = \xi_0|0_L\> + \xi_1|1_L\>$,
then the normalized projected state is either  
$\xi_0|0_L\> + \xi_1 e^{i \zeta_0}|1_L\>$.
If $U_\Delta |\psi\>$ is projected onto the support of $\Pi_1$, then the normalized projected state is
$\xi_0|0_L\> + \xi_1 e^{i \zeta_1}|1_L\>$.

When $g\Delta/2$ is close to zero, we have the approximation 
\begin{align}
    \zeta_j   &\approx 2(-1)^j i^{n-1} (g \Delta/2)^{n-2j}.
\end{align}
Since the probability of projecting onto the supports of $\Pi$ and $\Pi_1$ are approximately 1 and $n (g\Delta/2)^2$ respectively, the expected accumulated phase is approximately
\begin{align}
2(1-n) i^{n-1} (g \Delta/2)^n.
\end{align} 

The proof of Lemma \ref{lem:sandwich-ratios} uses trigonometric identities.
\begin{proof}[Proof of Lemma \ref{lem:sandwich-ratios}]
From Lemma \ref{lem:sandwich} and \eqref{phinj-deriv-2}, for odd $n$ we have 
\begin{align}
\frac{\<1_L|  U_\Delta  |1_L\>}{ \<0_L|  U_\Delta  |0_L\>} &= 
\frac{ \cos^n( g\Delta/2 ) + i^n\sin^n (g\Delta/2) } 
{ \cos^n( g\Delta/2 )  -   i^n \sin^n (g\Delta/2) }.
\label{eq:basic-fraction}
\end{align}
Since $n$ is odd, the absolute values of both the numerator and denominator in  \eqref{eq:basic-fraction} are identical.
Hence the quantity in \eqref{eq:basic-fraction} can be written as $e^{i a}$ for some real number $a$.
In fact, we can write the numerator and denominator in the form $R e^{i b}$ and $R e^{-ib}$ respectively for some positive number $R$, so that $a = 2b$.
We proceed to determine the value of $b$ using trigonometrical identities.
If $\mod(n,4)=1$, then 
\begin{align}
\frac{\<1_L|  U_\Delta  |1_L\>}{ \<0_L|  U_\Delta  |0_L\>} &= 
\frac{ \cos^n( g\Delta/2 ) + i\sin^n (g\Delta/2) } 
{ \cos^n( g\Delta/2 )  -   i \sin^n (g\Delta/2) }.
\label{eq:basic-fraction-mod4=1}
\end{align}
and $\tan b = \tan^n (g \Delta/2 )$.
If $\mod(n,4)=3$, then 
\begin{align}
\frac{\<1_L|  U_\Delta  |1_L\>}{ \<0_L|  U_\Delta  |0_L\>} &= 
\frac{ \cos^n( g\Delta/2 ) - i\sin^n (g\Delta/2) } 
{ \cos^n( g\Delta/2 )  +  i \sin^n (g\Delta/2) }.
\label{eq:basic-fraction-mod4=2}
\end{align}
and $- \tan b = \tan^{n-1} (g \Delta/2 )$.
Hence we conclude that 
$\tan b = i^n \tan^n (g \Delta/2 )$, and hence
$b = \arctan (i^{n-1} \tan^n (g \Delta/2 ))$.
By setting $\zeta_0 = 2b$ we get the first part of the lemma.

Next, note that 
\begin{align}
\frac
{\<{q_1}|U_\Delta|1_L\>}
{\<{q_0}|U_\Delta|0_L\>}
=
\frac
{\<{Q_1}|U_\Delta|1_L\>}
{\<{Q_0}|U_\Delta|0_L\>}
.
\end{align}
For simplicity, let $x = g\Delta/2$. 
Using \eqref{eq:83} in \eqref{lem:Q-U-sandwich} we get
\begin{align}
\frac
{\<{q_1}|U_\Delta|1_L\>}
{\<{q_0}|U_\Delta|0_L\>}
=&
\frac
{
-i  \cos^{n-1}x  \sin x 
- i^{n-1}   \sin^{n-1}x  \cos x 
}
{
-i  \cos^{n-1}x  \sin x 
 + i^{n-1}   \sin^{n-1}x  \cos x 
 }
 \notag\\
=&
\frac
{
 \cos^{n-2}x    
- i^{n }   \sin^{n-2}x  
}
{
   \cos^{n-2}x   
 + i^{n }   \sin^{n-2}x  
 }.
 \end{align}
Hence 
$\frac {\<{q_1}|U_\Delta|1_L\>} {\<{q_0}|U_\Delta|0_L\>}$ is of the form $\frac{e^{-i c}}{e^{+i c}}$ 
where $i^{n-1} \tan c = \tan^{n-2} x$, and setting $\zeta_1=-c$, we get the second result of the lemma.
\end{proof}
%%%%%%%%%%%%%%%%%%

Now we consider the case of larger $n$. 
Let $\hat K_0 = \hat J^z + (s-N/2)I$, and note that 
$\hat K_0|D^N_{s+gk}\>=-gk$.
In particular, we consider the set of states
\begin{align}
\mathcal K_j \coloneqq
\left\{   (\hat K_0/g) ^k|j_L\> : 
k=0,\dots, (n-1)/2 \right\}    ,
\end{align}
and denote the orthonormal vectors obtained from $\mathcal K_j$ by the Gram-Schmidt procedure to be given by
\begin{align}
\left\{ |\psi_{j,k} \> : 
k=0,\dots, (n-1)/2 \right\}   . 
\end{align}
We define the rank 2 projectors 
\begin{align}
    \tilde \Pi_{k} 
    =
    |\psi_{0,k}\>\<\psi_{0,k}|
    +
    |\psi_{1,k}\>\<\psi_{1,k}|,
\end{align}
and the corresponding spaces as $\mathcal Q_k$.
We show that for all $k=0,\dots, (n-1)/2$, we have
\begin{align}
\frac
{\<\psi_{1,k} | e^{-i \Delta \hat K }|1_L\>}
{\<\psi_{0,k} | e^{-i \Delta \hat K }|0_L\>}
\end{align}
is equal to $e^{i \zeta_k}$ for some real number $\zeta_k.$
Furthermore, when $g \Delta/2$ is close to zero, we have 
$|\zeta_k| = O( (g\Delta)^{n-k} ) $. 
This shows that after projection onto the space $\mathcal Q^k$ and mapping $\mathcal Q^k$ back to the codespace after the signal $U_\Delta$ accumulates, we pick up an effective evolution 
with phase $\zeta_k$ in the code's logical basis.

To see this, we first note that $|j^t_\sigma\>$ is supported on $(n+1)/2$ Dicke states, which we label as $|D^N_{f(j,k)}\>$ where $f(j,k)$ is some non-negative integer for $k=0,\dots, (n-1)/2$. We denote the space spanned by these Dicke states as $\mathcal Z_j$.

Second, since Dicke states are eigenstates of the unitary operator $U_\Delta$,
the unitary operator $U_\Delta$ only maps $|j^t_\sigma\>$ 
to quantum states that are superposition of the $(n+1)/2$ Dicke states that $|j^t_\sigma\>$ is supported on.

We numerically verify that the states 
$|j_L\>, (\hat J^z)|j_L\>, \dots,  (\hat J^z)^{(n-1)/2}|j_L\>$ are basis vectors of $\mathcal Z_j$ for \yo{values of} odd $n$ from 3 to 55.
We verify this by evaluating \yo{symbolically} the rank of the Gramian matrix $G_j
=\sum_{k,\ell=0}^{(n-1)/2} G_{j,k,\ell}|k\>\<\ell|$ for an orthonormal basis $\{|0\>,\dots,|(n-1)/2\>\}$ where
\begin{align}
    G_{j,k,\ell} =
    \<j_L | ((s-N/2)I + \hat J^z)^k ((s-N/2)I + \hat J^z)^\ell | j_L \>,
\end{align}
because the linear independence of the vectors
$|j_L\>, (\hat J^z)|j_L\>, \dots,  (\hat J^z)^{(n-1)/2}|j_L\>$ 
is equivalent to the linear independence of the vectors 
$|j_L\>, ((s-N/2)I + \hat J^z)|j_L\>, \dots,  ((s-N/2)I + \hat J^z)^{(n-1)/2}|j_L\>$. 

Note that $G_0 = G_1$ because 
 Lemma \ref{lem:Jzj sandwich} allows us to show that $\<0_L | (\hat J^z)^k (\hat J^z)^\ell | 0_L \>
=
\<1_L | (\hat J^z)^k (\hat J^z)^\ell | 1_L \>$
for all $k,\ell =0,\dots, (n-1)/2$.
When the Gramian matrix $G_j$ is full rank, we can write down basis vectors of $\mathcal Z_j$ 
as 
\begin{align}
    |\phi_{j,k}\>
    = \sum_{\ell} u_{k,\ell} (\hat J^z)^\ell |j_L\>
\end{align}
for $k=1,\dots, (n+1)/2$ and $j=0,1$. The coefficients $u_{k,\ell}$ are independent of $j$.

\section{Multiple deletions during signal accumulation}

\yo{
\subsection{Concentration of $\sigma$ about $t/2$ in every one of $r$ rounds}}

\yo{
With probability approaching 1, the random shift $\sigma$ is close to $t/2$. This is because the probability of obtaining a shift $\sigma$ is 
between 
$\binom t \sigma 
\<0^t_\sigma
|0^t_\sigma\>
$
and 
$\binom t \sigma 
\<1^t_\sigma
|1^t_\sigma\>
$,
which approaches the binomial distribution $\binom t \sigma 2^{-t}$ for large $N$ when $gt=o(N)$. 
This condition is almost surely satisfied when $\tau = 1/r$ for $r = \Theta(g^{1+\delta})$ for some positive $\delta$.}
 
From the Hoeffding inequality, 
given \yo{any} $e_2$ in the interval $[0,1/2]$,
the probability that the shift $\sigma$ does not satisfy the 
inequality  
\begin{align}
|\sigma - t/2| \le t^{1/2+e_2}
\label{eq:sigma-concentration}
\end{align}
is at most 
$2 \exp[-c t^{2e_2}]$ for some positive constant $c$.
% By the union bound, the probability that, in any of $r$ rounds, the random shift does not satisfy \eqref{eq:sigma-concentration} is at most $2 r \exp[-c t^{2e_2}]$, which is exponentially suppressed whenever $t$ and $r$ relate to each other by a polynomial. 

 \subsection{Amplitude distortions} 

% The situation of multiple deletion errors is more complicated. 
When $t < g$, $t$ deletions shifts the Dicke weights randomly either from 0 to $t$.
From \eqref{eq:psia-delete}, deleting $t$ qubits from an initial $N$-qubit logical state $|\psi\> = \xi_0 |0_L\>+\xi_1 e^{i\varphi}|1_L\>$ in 
a shifted gnu code gives a probabilistic mixture of the (unnormalized) $(N-1)$-qubit states
\begin{align}
|\psi^t_\sigma\>  = 
\xi_0|0^t_\sigma\>
+ \xi_1 e^{i\varphi}|1^t_\sigma\> 
\end{align}
where $\sigma = 0,\dots, t$. 
Note that 
\begin{align}
\frac{\binom{N-1}{gk+s-\sigma}} {\binom{N}{gk+s} }
&=
\case{
1-(gk+s)/N &, \quad \sigma= 0 \\
(gk+s)/N&, \quad \sigma= 1
}\notag\\
&=\case{
1/2 + (\lfloor{gn/2}\rfloor-gk)/N &, \quad \sigma= 0 \\
1/2 - (\lfloor{gn/2}\rfloor-gk)/N&, \quad \sigma= 1
}.\notag
\end{align}
In the case of general $t$, we have 
\begin{align}
\frac{\binom{N-t}{gk+s-\sigma}} {\binom{N}{gk+s} }
&=
\frac{(gk+s)_{(\sigma)} (N-gk-s)_{(t-\sigma)}}{N_{(t)}} .\label{binom fraction}
\end{align}
When $t=1$, we have 
\begin{align}
\frac{\binom{N-1}{gk+s-\sigma}} {\binom{N}{gk+s} }
=&\case{
(1-s/N) -gk/N &, \quad \sigma= 0 \\
s/N  + gk/N&, \quad \sigma= 1
} 
.\end{align}

When $s = N/2 - gn/2 + O(1)$, we have 
\begin{align}
   \frac{\binom{N-t}{gk+s-\sigma}} {\binom{N}{gk+s} }= &
\left(N/2+g y_k\right)_{(\sigma)}
\left(N/2 - g y_k \right)_{(t-\sigma)}\notag\\
&\times N_{(t)}^{-1}(1+O(t/N)),
\label{eq:binom-fraction-approx}
\end{align}
where $y_k = k-n/2$.

Now define 
\begin{align}
    \hat K^{t,\sigma}_{0}
    \coloneqq 
    \hat J^z + (s-\sigma - (N-t)/2)I,
\end{align}
and 
\begin{align}
    \hat K^{t,\sigma}
    \coloneqq 
    \hat K^{t,\sigma}_0 + (gn/2)I.
\end{align}
The purpose of this definition is exploit the identity 
\begin{align}
\hat K^{t,\sigma}_0|D^{N-t}_{s-\sigma + gk}\> = -gk |D^{N-t}_{s-\sigma + gk}\>,
\end{align}
from which we have
\begin{align}
    \hat K^{t,\sigma}|D^{N-t}_{s-\sigma + gk}\> = -g y_k|D^{N-t}_{s-\sigma + gk}\>,
\end{align}
which allows us to write
\begin{align}
(I/2 + \hat K^{t,\sigma}/N)|D^{N-t}_{s-\sigma + gk}\> 
% = 
% (1/2 + (-gk+gn/2)/N) |D^{N-t}_{s-\sigma + gk}\>
=(1/2 -g y_k /N) |D^{N-t}_{s-\sigma + gk}\>  \label{K-id}.
\end{align}
Moreover, we have that 
\begin{align}
&\<j^{t,\sigma}_L|
e^{ -i \Delta    \hat K^{t,\sigma} }
|j^{t,\sigma}_L\>\notag\\
=& \cos^{n} (g\Delta/ 2) 
+
(-1)^j(-i)^n
\sin^{n} (g\Delta/ 2).
\end{align}
Then for odd $n$, 
for any even function $F(\cdot)$, we have that
for any choice of $\Delta$ that
\begin{align}
\frac{
|\<0^{t,\sigma}_L|
F(\hat K^{t,\sigma})
e^{ -i \Delta    \hat K^{t,\sigma} }
|0^{t,\sigma}_L\>|}
{
|\<1^{t,\sigma}_L|
F(\hat K^{t,\sigma})
e^{ -i \Delta    \hat K^{t,\sigma} }
|1^{t,\sigma}_L\>|} = 1.
\label{even-function-same-amplitude}
\end{align}
To see why this relation holds, observe that when $n$ is an odd integer, the index reflection $k \mapsto n-k$ maps the even parity support of the logical state $|0_{t,\sigma}^L\rangle$ has a one-to-one correspondence onto the odd parity support of $|1_{t,\sigma}^L\rangle$. 
Under this transformation, the binomial coefficients $\binom n k$ are symmetric, in the sense that $\binom n k = \binom n {n-k}$.
Furthermore, the eigenvalues of $\hat K^{t,\sigma}$ transform from $y_k$ to $y_{n-k} = -y_k$. Because $F(\cdot)$ is an even function, for any even $\bar z$, the monomial
$y_k^{\bar z}$ transform to $y_{n-k}^{\bar z} = y_k^{\bar z}$.

Furthermore, we have that
\begin{align}
c x^n \le 
\left|\arg \left(\frac{\<0^{t,\sigma}_L|
F(\hat K^{t,\sigma})
e^{ -i \Delta    \hat K^{t,\sigma} }
|0^{t,\sigma}_L\>}
{
\<1^{t,\sigma}_L|
F(\hat K^{t,\sigma})
e^{ -i \Delta    \hat K^{t,\sigma} }
|1^{t,\sigma}_L\>}\right)\right| \le c x.
\label{even-function-same-amplitude-arg}
\end{align}
for some real constants $c,c'$ for some positive constants $c$ and $c'$, where $x=g \Delta/2$.

The signal accumulates on the state
$\xi_0|0^t_\sigma\>
+\xi_1|1^t_\sigma\>$
now to yield the state
\begin{align}
e^{-i \Delta  \hat K^{t,\sigma}}(\xi_0|0^t_\sigma\>
+\xi_1|1^t_\sigma\>)    
\end{align}
Next we consider the logical codewords of a shifted gnu code with total number of qubits $N-t$, shift $s-\sigma$, and parameters $g$ and $n$, with logical codewords given by 
$|0^{t,\sigma}_L\>$ and $|1^{t,\sigma}_L\>$.
Then we consider the set of states
\begin{align}
    \mathcal K_j^{t,\sigma} 
    \coloneqq 
    \left\{
(\hat K_{t,\sigma,0}/g)^k |j^{t,\sigma}_L\>
    : k = 0,\dots, (n-1)/2
    \right\}
\end{align}
and denote the orthonormal vectors obtained from $\mathcal K_j^{t,\sigma}$ by the Gram-Schmidt procedure to be given by 
\begin{align}
\left\{ |\psi_{j,k}^{t,\sigma} \> : 
k=0,\dots, (n-1)/2 \right\}   . 
\end{align}
% We define the rank 2 projectors 
% \begin{align}
%     \tilde \Pi_{k} ^{t,\sigma}
%     =
%     |\psi_{0,k} ^{t,\sigma}\>\<\psi_{0,k} ^{t,\sigma}|
%     +
%     |\psi_{1,k} ^{t,\sigma}\>\<\psi_{1,k} ^{t,\sigma}|,
% \end{align}
% and the corresponding spaces as $\mathcal Q^{t,\sigma}_k$.
We write the Gram matrix associated with $\mathcal K^{t,\sigma}_j$ as 
\begin{align}
    G^{t,\sigma}_j 
    \coloneqq
    \sum_{a,b=0}^{(n-1)/2}
    \<j^{t,\sigma}_L|
    (\hat K_{t,\sigma,0}/g)^{a+b}
    |j^{t,\sigma}_L\>
    |a\>\<b|.
\end{align}
The error correction property of the shifted gnu code allows us to find that the Gram matrices $G^{t,\sigma}_0$ and $G^{t,\sigma}_1$ are identical, and hence we denote 
$G^{t,\sigma} \coloneqq G^{t,\sigma}_j$.
When the Gram matrix $G^{t,\sigma}$ is full rank, we have that 
\begin{align}
    |\psi^{t,\sigma}_{j,k} \>
    &= \sum_{a=0}^{(n-1)/2} 
    c_{k,a}
    (\hat K_{t,\sigma,0}/g)^a |j^{t,\sigma}_L\>, \label{Gram-linear-combi}\\
     &= \sum_{w=0}^{(n-1)/2} 
    b_{k,w}
    (I/2 + \hat K_{t,\sigma}/N)^w |j^{t,\sigma}_L\>
    \label{nice-linear-combi}
\end{align}
for some real constants $c_{k,a},b_{k,w}$.
Since all of the matrix elements of $G^{t,\sigma}$ are $O(1)$ when $n=O(1)$,
the matrix elements $ \<a|(G^{t,\sigma})^{-1/2})|k\>$ are also $O(1)$.

%%%%%%%%%%%% LEMMA %%%%%%%%
\begin{lemma}
Let $\epsilon > 0$, and let $t$ and $\sigma$ be non-negative integers such that 
$|t-2\sigma| \le c t^{1/2+\epsilon}$
for some positive constant $c$.
Then for all $k=0,\dots, \lfloor (n-1)/2 \rfloor$, and for any value of $\Delta$, we have 
\begin{align}
\frac{
|\<\psi_{1,k}^{t,\sigma}|
  e^{-i \Delta  \hat K^{t,\sigma} }
|1^t_\sigma\>|}
{
|\<\psi_{0,k}^{t,\sigma}|
  e^{-i \Delta  \hat K^{t,\sigma} }
|0^t_\sigma\>|
}
=
1+O(gt^{1/2+\epsilon}/N)
 +O( t/N).
\end{align}
\end{lemma}
%%%%%%%%%%%% LEMMA %%%%%%%%
%%%%%%%%%%%% LEMMA PROOF %%%%%%%%
\begin{proof}
Then, from \eqref{nice-linear-combi}
and \eqref{eq:binom-fraction-approx}, %and \eqref{K-id},
we can write $\<\psi_{j,k}^{t,\sigma}|
  e^{-i \Delta  \hat K^{t,\sigma} }
|j^t_\sigma\>$ as
\begin{align} 
&
\sum_{w=0}^{(n-1)/2}
b_{k,w}
\<j^{t,\sigma}_L|
(I/2 +\hat K_{t,\sigma}/N)^w 
   N_{(t)}^{-1/2} 
     e^{-i \Delta  \hat K^{t,\sigma} }
\notag\\
&\times\!\!\!
\sum_{{\rm mod}(a,2)=j}\!\!\!
(1+O(t/N)) \sqrt{
  (N/2+g y_a)_{(\sigma)} 
  (N/2-g y_a )_{(t-\sigma)}
  }
 \notag\\
 &\times
 2^{-(n-1)/2}
 \sqrt { \binom n a }
 |D^{N-t}_{s-\sigma+ga}\> .
 \label{complicated-expr}
\end{align}
For any non-negative integer $\mu$ such that $\mu\le  \sigma$ and $\mu \le t-\sigma,$ we can rewrite \eqref{complicated-expr} as
\begin{align} 
&
\sum_{w=0}^{(n-1)/2}
b_{k,w}
\<j^{t,\sigma}_L|
(I/2 +\hat K_{t,\sigma}/N)^w 
   N_{(t)}^{-1/2} 
(1+O(t/N))\notag\\
&\times\!\!\!\!\!\!
\!\!\!\! \sum_{{\rm mod}(a,2)=j}\!\!\!\!\!\!\!\!\!\!
 \sqrt{
  ((N/2-\mu)+g y_a)_{(\sigma-\mu)} 
  ((N/2-\mu)-g y_a )_{(t-\sigma-\mu)}
  }
 \notag\\
 &\times
F_\mu(y_a)
     e^{-i \Delta  \hat K^{t,\sigma} }
 \sqrt { \binom n a }
 |D^{N-t}_{s-\sigma+ga}\> ,
 \label{complicated-expr2}
\end{align}
where
$F_\mu(y_a)\coloneqq \prod_{v=0}^{  \mu-1  }
 \sqrt{ (N/2-v)^2 - g^2 y_a^2  }
$ is an even function of $y_a$.
Now we can always pick $w = \min\{\sigma,t-\sigma\}$ so that 
both $\sigma-\mu$ and $t-\sigma-\mu$
are $O(t^{1/2+\epsilon})$.
Then, expanding the logical codeword as a linear combination of Dicke states, \eqref{complicated-expr2} can be written as 
\begin{align} 
&
\!\!\!\! \sum_{{\rm mod}(a,2)=j}\!\!
\sum_{w=0}^{(n-1)/2}\!\!\!\!
b_{k,w}
\<D^{N-t}_{s-\sigma+ga}|
2^{-w}(1+ O(wg/N))
\notag\\
  &\times
(1+O(t/N)+O(gt^{1/2+\epsilon}/N)  )
\notag\\
&\times
 \frac{\binom n a  r'}
 {2^{n-1}}
 F_\mu(y_a)
  e^{-i \Delta  \hat K^{t,\sigma} }
 |D^{N-t}_{s-\sigma+ga}\> ,
 \label{complicated-expr3}
\end{align}
where 
\begin{align}
r' = 
   N_{(t)}^{-1/2} 
 \sqrt{(N/2-\mu)_{\sigma-\mu}
(N/2-\mu)_{t-\sigma-\mu}}.
\end{align}
We further simplify \eqref{complicated-expr3} as
\begin{align}
&
r'' \<j^{t,\sigma}_L |  
F_\mu(\hat K^{t,\sigma})
  e^{-i \Delta  \hat K^{t,\sigma} }
|j^{t,\sigma}_L\> 
\notag\\
  &\times
(1+O(t/N)+O(gt^{1/2+\epsilon}/N) ).
 \label{complicated-expr4}
\end{align}
for some real number $r''$ that does not depend on $j$. 

Since
 $F_\mu(y_a)$ is an even power series of $y_a$, we use \eqref{even-function-same-amplitude}, which is an equality that is independent of $\Delta$, to show the lemma.
\end{proof}
If we have $ g/\sqrt N = \Omega(1)$, the above lemma simplifies to give
\begin{align}
\frac{
|\<\psi_{1,k}^{t,\sigma}|
  e^{-i \Delta  \hat K^{t,\sigma} }
|1^t_\sigma\>|}
{
|\<\psi_{0,k}^{t,\sigma}|
  e^{-i \Delta  \hat K^{t,\sigma} }
|0^t_\sigma\>|
}
=
1+O(gt^{1/2+\epsilon}/N).
\end{align}
which means that there is a constant $c>0$ such that
\begin{align}
e^{-c gt^{1/2+\epsilon}/N }
\le
\frac{|\<\psi_{1,k}^{t,\sigma}|
  e^{-i \Delta  \hat K^{t,\sigma} }
|1^t_\sigma\>|}
{|\<\psi_{0,k}^{t,\sigma}|
  e^{-i \Delta  \hat K^{t,\sigma} }
|0^t_\sigma\>|}
\le 
e^{c gt^{1/2+\epsilon}/N },
\label{amp-shift-result}
\end{align}
After multiple rounds, the effective amplitude shifts on the logical codespace can become severe.
Later, we discuss how we can reverse this amplitude-shifting process.

 \section{Amplitude rebalancing}

Let $x= g \Delta'/2$. 
Here, we calculate 
$\<j_{h}|U_{\Delta'}|j_L\>$ and focus on the shifted gnu code with $n=3$.
Note that we have
\begin{align}
&\<j_{h}|U_{\Delta'}|j_L\>\notag\\
=&\frac{1}{2}
\left(
\sqrt {3+(-1)^j h}
\<j_L|U_{\Delta'}|j_L\>
+
\sqrt {1-(-1)^j h}
\<q_j|U_{\Delta'}|j_L\>
\right).
\end{align}
We calculate
\begin{align}
\<j_L|U_{\Delta'}|j_L\>
=& 
\left(
\cos^3 x + (-1)^j i \sin^3 x
\right) e^{-i \Delta' (N/2-s)e^{3ix}}
\end{align}
and 
\begin{align}
\<q_j|U_{\Delta'}|j_L\>
=(\sqrt 3/2)
(-1)^j 
 e^{-i \Delta' (N/2-s)e^{3ix}}
 e^{(-1)^j ix}
 \sin 2x.
\end{align}
Therefore, we find that
\begin{align}
&\frac
{ \<1_{h}|U_{\Delta'}|1_L\> }
{ \<0_{h}|U_{\Delta'}|0_L\> }    \notag\\
=&
\frac
{\sqrt {3-h}
\<1_L|U_{\Delta'}|1_L\>
+
\sqrt {1+h}
\<q_1|U_{\Delta'}|1_L\>
}
{
\sqrt {3+h}
\<0_L|U_{\Delta'}|0_L\>
+
\sqrt {1-h}
\<q_0|U_{\Delta'}|0_L\>
}\\
=&
\frac
{\sqrt {3-h}
\left(
\cos^3 x - i \sin^3 x
\right)
-
\sqrt {1+h}(\sqrt 3/2)
 e^{- ix} \sin 2x
}
{
\sqrt {3+h}
\left(
\cos^3 x +  i \sin^3 x
\right)
+
\sqrt {1-h}(\sqrt 3/2)  e^{ix} \sin 2x
}.
\end{align}
Expanding $e^{\pm i x} = \cos x \pm i \sin x$, we get
\begin{widetext}
    \begin{align}
\frac
{ \<1_{h}|U_{\Delta'}|1_L\> }
{ \<0_{h}|U_{\Delta'}|0_L\> }  
=&        
 \frac
{\sqrt {3-h}
\left(
\cos^3 x - i \sin^3 x
\right)
-
\sqrt {1+h}(\sqrt 3/2)
\cos x \sin 2x
+
i \sqrt {1+h}(\sqrt 3/2)
\sin x \sin 2x
}
{
\sqrt {3+h}
\left(
\cos^3 x +  i \sin^3 x
\right)
+
\sqrt {1-h}(\sqrt 3/2) 
\cos x \sin 2x
+
i \sqrt {1-h}(\sqrt 3/2) 
\sin x \sin 2x
} 
\\
=&        
 \frac
{
\left(
\sqrt {3-h} \cos^3 x -
\sqrt {1+h}(\sqrt 3/2)
\cos x \sin 2x
\right)
- i \sqrt {3-h} \sin^3 x
+
i \sqrt {1+h}(\sqrt 3/2)
\sin x \sin 2x
}
{
\left(
\sqrt {3+h}\cos^3 x +
\sqrt {1-h}(\sqrt 3/2) 
\cos x \sin 2x
\right)
+  i \sqrt {3+h} \sin^3 x
+
i \sqrt {1-h}(\sqrt 3/2) 
\sin x \sin 2x
}.
\end{align}  
Next, we calculate $\frac
{ \<1_{h}|U_{\Delta'}|1_L\> }
{ \<0_{h}|U_{\Delta'}|0_L\> }  $.
Now we write
\begin{align}
\frac{ \<1_{h}|U_{\Delta'}|1_L\> }
{ \<0_{h}|U_{\Delta'}|0_L\> }
&=
\frac{A_1 - i A_2}{B_1 + i B_2}
\\
&=
\frac{(A_1 - i A_2)
(B_1 - i B_2)}{B_1^2 + B_2^2}\\
&=
\frac{
(A_1 B_1- A_2 B_2)
- i (A_2 B_1 + A_1 B_2 )
}{B_1^2 + B_2^2}
,
\end{align}
where $A_1, A_2,B_1,B_2$ are real numbers which are given as
\begin{align}
    A_1 &= \sqrt {3-h} \cos^3 x -
\sqrt {1+h}(\sqrt 3/2)
\cos x \sin 2x\\
A_2 &= \sqrt {3-h} \sin^3 x
- \sqrt {1+h}(\sqrt 3/2)
\sin x \sin 2x\\
B_1&= \sqrt {3+h}\cos^3 x +
\sqrt {1-h}(\sqrt 3/2) 
\cos x \sin 2x\\
B_2 &= \sqrt {3+h} \sin^3 x
+
 \sqrt {1-h}(\sqrt 3/2) 
\sin x \sin 2x.
\end{align}

Next we calculate the argument of $\frac
{ \<1_{h}|U_{\Delta'}|1_L\> }
{ \<0_{h}|U_{\Delta'}|0_L\> }  $.
Namely, 
\begin{align}
{\rm arg} \left(    
\frac
{ \<1_{h}|U_{\Delta'}|1_L\> }
{ \<0_{h}|U_{\Delta'}|0_L\> }  
\right)
=&
\arctan
\left( \frac{-A_2}{A_1} \right)
-
\arctan
\left( \frac{B_2}{B_1} \right) 
\notag\\
=&
\left(\frac{\sqrt{1+h}}{\sqrt{3-h}}
-
\frac{ \sqrt{1-h}}{\sqrt{3+h}}\right) \sqrt 3 x^2
+\frac{8 h^2 x^3}
{9-h^2}
   +
   \left(
   \frac{\sqrt{1+h} (3+5 h)}
   {(3-h)^{3/2}}
   -
   \frac{\sqrt{1-h} (3-5 h)}{ (3+h)^{3/2}}
   \right)
   \frac{2 x^4}{\sqrt 3}
   +O\left(x^5\right),
\end{align}
which implies that the phase accumulation here scales as 
$\left(\frac{\sqrt{1+h}}{\sqrt{3-h}}
-
\frac{ \sqrt{1-h}}{\sqrt{3+h}}\right) \sqrt 3 x^2 + 
O(x^3)$.
We also have that 
\begin{align}
\frac
{ \<1_{h}|U_{\Delta'}|1_L\> }
{ \<0_{h}|U_{\Delta'}|0_L\> }  
&=
\sqrt{\frac{3-h}{3+h}}
-
\frac{\sqrt 3 x
(
\sqrt{(1-h)(3-h)}
+
\sqrt{(1+h)(3+h)}
)}
{3+h} + O(x^2).
\end{align}

Next we calculate
\begin{align}
&  \frac
{ \<\bar 1_{h}|U_{\Delta'}|1_L\> }
{ \<\bar 0_{h}|U_{\Delta'}|0_L\> } \notag\\
=&
\frac
{\sqrt {1+h}
\<1_L|U_{\Delta'}|1_L\>
-
\sqrt {3-h}
\<q_1|U_{\Delta'}|1_L\>
}
{
\sqrt {1-h}
\<0_L|U_{\Delta'}|0_L\>
-
\sqrt {3+h}
\<q_0|U_{\Delta'}|0_L\>
}\\
=&
\frac
{\sqrt {1+h}
\left(
\cos^3 x - i \sin^3 x
\right)
+
\sqrt {3-h}(\sqrt 3/2)
 e^{- ix} \sin 2x
}
{
\sqrt {1-h}
\left(
\cos^3 x +  i \sin^3 x
\right)
-
\sqrt {3+h}(\sqrt 3/2)  e^{ix} \sin 2x
}\notag\\
=&
\frac
{\sqrt {1+h}
\cos^3 x
 - i \sqrt {1+h} \sin^3 x
+
\sqrt {3-h}(\sqrt 3/2)
 \cos{x} \sin 2x
- i
\sqrt {3-h}(\sqrt 3/2)
 \sin{x} \sin 2x
}
{
\sqrt {1-h}
\cos^3 x 
+ i  \sqrt {1-h}    \sin^3 x
-
\sqrt {3+h}(\sqrt 3/2)  \cos x \sin 2x
-
i \sqrt {3+h}(\sqrt 3/2)  \sin x \sin 2x
}\notag\\
=&
\frac
{
\left(
\sqrt {1+h} \cos^3 x
+
\sqrt {3-h}(\sqrt 3/2)
 \cos{x} \sin 2x
\right)
 - i \left(\sqrt {1+h} \sin^3 x
+
\sqrt {3-h}(\sqrt 3/2)
 \sin{x} \sin 2x \right)
}
{
\left(
\sqrt {1-h} \cos^3 x 
-
\sqrt {3+h}(\sqrt 3/2)  \cos x \sin 2x
\right)
+ i \left( \sqrt {1-h}    \sin^3 x
-
 \sqrt {3+h}(\sqrt 3/2)  \sin x \sin 2x\right)
}.
\end{align}
Similarly, we write 
\begin{align}
 &  \frac
{ \<\bar 1_{h}|U_{\Delta'}|1_L\> }
{ \<\bar 0_{h}|U_{\Delta'}|0_L\> } 
=
\frac{\bar A_1  - i \bar A_2}
{\bar B_1 + i B_2},
\end{align}
where 
\begin{align}
 \bar   A_1 &=\sqrt {1+h} \cos^3 x
+
\sqrt {3-h}(\sqrt 3/2)
 \cos{x} \sin 2x\\
\bar A_2 &= \sqrt {1+h} \sin^3 x
+
\sqrt {3-h}(\sqrt 3/2)
 \sin{x} \sin 2x\\
\bar  B_1 &=\sqrt {1-h} \cos^3 x 
-
\sqrt {3+h}(\sqrt 3/2)  \cos x \sin 2x\\
\bar B_2 &=
 \sqrt {1-h}    \sin^3 x
-
 \sqrt {3+h}(\sqrt 3/2)  \sin x \sin 2x.
\end{align}
We have 
\begin{align}
&{\rm arg}\left(
{ \<\bar 1_{h}|U_{\Delta'}|1_L\> }
{ \<\bar 0_{h}|U_{\Delta'}|0_L\> } \right)\notag\\
=&
-\arctan\left(
\frac
{-\bar A_2}
{
\bar A_1
}
\right)
-
\arctan\left(
\frac
{\bar B_2}
{
\bar B_1}
\right)\\
=&
\left(\frac{ \sqrt{3+h}}{\sqrt{1-h}}
-
\frac{ \sqrt{3-h}}{\sqrt{1+h}}\right) \sqrt 3 x^2
+
\frac{8 \left(2+h^2\right) x^3}
{1-h^2}
   +
   \left(-
   \frac{ \sqrt{3-h} (13-5 h)}
   { (1+h)^{3/2}}
   +
   \frac{ \sqrt{3+h} (13+5h)}
   { (1-h)^{3/2}}\right)
   \frac{2x^4}{\sqrt 3}
   +O\left(x^5\right),
\end{align}
which shows that the phase accumulation is $\left(\frac{ \sqrt{3+h}}{\sqrt{1-h}}
-
\frac{ \sqrt{3-h}}{\sqrt{1+h}}\right) \sqrt 3 x^2 + O(x^3)$.
Furthermore,
we have 
\begin{align}
&
{ \<\bar 1_{h}|U_{\Delta'}|1_L\> }
{ \<\bar 0_{h}|U_{\Delta'}|0_L\> } \notag\\
=&
\frac
{(\bar A_1 \bar B_1 - \bar A_2 \bar B_2)
-i(\bar A_2 \bar B_1 + \bar A_1 \bar B_2)}
{\bar  B_1^2 + \bar B_2^2}\notag\\
=&
\sqrt{\frac{1+h}{1-h}}
+ 
\frac{
\sqrt 3
\left(
\sqrt{(1-h)(3-h)}
+
\sqrt{(1+h)(3+h)}
\right) x
}{1-h} + O(x^2).
\end{align}

\end{widetext} 
Now consider an asymmetric random walk on the integer line,
where the value of the walk, starting at the value of 0, increases by 1 or decreases by 1 with probabilities 3/8 and 5/8 respectively, until it reaches a target value of $-w$. 
We like to calculate the number of steps of the walk for us to almost surely terminate the walk.

Now we define $E_k$ as the expected number of steps to reach position $-w$ starting from position $k$. Then, for all $ k > -w $, we use the recurrence
$E_k = 1+  \frac{3}{8} E_{k+1} + \frac{5}{8} E_{k-1}.$
The corresponding characteristic equation to the homogeneous part is $3r^2- 8r +5 = 0$,
which factorises to $(3r-5)(r-1)=0$. Thus,
the solution to the homogeneous part of this linear recurrence is $(5/3)^k A + B$ for some constants $A$ and $B$. A particular solution to the inhomogeneous part of the linear recurrence is $4k$, implying that the general solution to this linear recurrence is $E_k=  (5/3)^k A + B + 4k$.
Imposing the boundary condition $E_{-w} = 0,$ we find that 
$B = 4w - (5/3)^{-w} A$.
Now we calculate $E_{-w+1}$ using the idea of Dyck paths and Catalan numbers. 
Namely,
\begin{align}
    E_{-w+1}
    = 
    \sum_{n=0}^\infty
    (2n+1)C_n (3/8)^n  (5/8)^{n+1}
    = 4,
\end{align}
where $(2n+1)$ represents the length of the walk, and the Catalan number $C_n = \frac{1}{n+1} \binom {2n}{n}$ counts the number of Dyck paths of length $n$ and $C_n (3/8)^n (5/8)^{n+1}$ counts the probability that a length $2n+1$ walk that starts from $-w+1$ terminates at $-w$.
The linear recurrence solution implies that 
$E_{-w+1} -E_{-w}
= 4 + A(5/3).$
But since $E_{-w+1}=4$ and $E_{-w} = 0$, this means that $A=0$.
Therefore,
\begin{align}
    E_k = 4k+4w,
\end{align}
which implies that
\begin{align}
E_0 &=  4 w.
\end{align} 
This means that on average, 
a linear number of steps are required to terminate the walk.
Hence, the Markov's inequality shows that the probability that a walk is not terminated after 
$(4 w)^{1+\delta}$ 
steps is at most 
$(4 w)^{-\delta}$.

Our above random walk analysis shows that 
with at least 
$\nu = \Theta(\mathbb E[S_w]^{1+4\delta})$ rebalancing steps,
allowing projections onto both 
$S^0_\pm$
and 
$S^1_\pm$,
we can with probability at least 
$1- (4 w)^{-\delta}$ (and hence almost surely for large $N$)
bring the amplitude shift within a constant factor of 1.
Now, 
\begin{align}
\mathbb E[S_w] 
&= 
O(w (g/N)(N/g)^{1/2+\epsilon})   \\
&=
O(w (g/N)^{1/2-\epsilon}).  
\end{align}

\section{Evaluating the FI of ECSense}

% Let $r_{t,\sigma}$ count the number of timesteps where the number of deletions is $t$ and where the random shift is $\sigma$.
% Then, the total phase that accumulates during a successful run of Protocol 1 depends is 
% \begin{align}
%     \Phi = \sum_{t\ge 0} \sum_{\sigma = 0}^t r_{t,\sigma} \phi_{t,\sigma},
%     \label{eq:P1-phase-accumulated}
% \end{align}
% Note that the total phase $\Phi$ accumulated is a random variable, because $r_{t,\sigma}$ are random variables.

As argued earlier, 
when a timestep has $t$ deletions, the random shifts are almost surely concentrated around $t/2 \pm \Theta(t^{1/2+\epsilon})$,
in which case the phases $\phi_{t,\sigma}$ vary from $\Theta(x)$ to $\Theta(x^n)$,
where $x = g b \tau/ 2$.
Hence the total phase is almost surely at least 
$\Theta(r (g b \tau)^3 )$ which is a function of $b$. Then the following lemma gives the FI that we can extract from our protocol.
 
\begin{lemma}
\label{lem:signal-within-code}
Let $\Phi$ be a continuous function of $b$, and let 
$|\psi(\Phi)\> = 
\cos \phi |a_0\> + e^{i \Phi}\sin\phi |a_1\> $ where
$|a_0\>$ and $|a_1\>$ be orthonormal vectors.
Let 
$|a_+\>= \frac{1}{\sqrt 2} (|a_0\> + |a_1\>) $
 and 
 $|a_-\>= \frac{1}{\sqrt 2} (|a_0\> - |a_1\>) $.
Then the FI of $b$ by performing projective measurements on $|\psi(\Phi)\>$ with respect to the projectors $|a_+\>\<a_+|$ and $|a_-\>\<a_-|$ when $\phi = \pi/4$ is
$(\frac{d\Phi}{db})^2$.
In general the FI is 
\begin{align}
  F  =   \frac{\sin^2 2\phi \sin^2 \Phi}
    {    (1-\sin^2 2\phi \cos^2 \Phi)}
    \left(\frac{\partial\Phi}{\partial b}\right)^2.
\end{align}
\end{lemma}  

\begin{proof}[Proof of Lemma \ref{lem:signal-within-code}]
Let us define $\bar U_\Phi$, so that 
\begin{align}
\bar U_{\Phi} |a_+\> 
&=
\frac{|a_0\> + e^{i\Phi}|a_1\>}{\sqrt 2}\\
\bar U_{\Phi} |a_-\> 
&=
\frac{|a_0\> - e^{i\Phi}|a_1\>}{\sqrt 2}.
\end{align}
Note that 
$|\psi(\Phi)\>\<\psi(\Phi)| = 
\bar U_\Phi 
|a_+\>\<a_+| 
\bar U_\Phi ^\dagger$.
Hence we find that 
\begin{align}
\tr |a_+\>\<a_+| \psi(\Phi)\>\<\psi(\Phi)|
    &=\<a_+|  \bar U_\Phi |a_+\>\<a_+| \bar U_\Phi ^\dagger|a_+\> \notag\\
    &=|\<a_+|  \bar U_\Phi |a_+\>|^2\\
\tr |a_-\>\<a_-| \psi(\Phi)\>\<\psi(\Phi)|
    &=\<a_-|  \bar U_\Phi |a_+\>\<a_+| \bar U_\Phi ^\dagger|a_-\>\notag\\
    &=|\<a_-|  \bar U_\Phi |a_+\>|^2.
\end{align}
The FI from measurement corresponding to the projectors 
$|a_+\>\<a_+|$ and $|a_-\>\<a_-|$ is
\begin{align}
I(\theta)
= \frac{1}{p_+}\left(\frac{\partial p_+}{\partial\theta}\right)^2
+ \frac{1}{p_-}\left(\frac{\partial p_-}{\partial\theta}\right)^2,
\end{align}
where 
\begin{align}
p_+
&=
\< a_+ |  \bar U_{\Phi} |a_+\>\<a_+| \bar U_{\Phi}^\dagger  | a_+ \>
\notag\\
p_-
&=
\< a_- |  \bar U_{\Phi} |a_-\>\<a_-| \bar U_{\Phi}^\dagger  | a_- \>.
\end{align}
When $\phi = \pi/4$,
$p_+ = \cos^2(\Phi/2)$ and 
$p_- = \sin^2(\Phi/2).$
Hence
$\frac{\partial p_\pm } {\partial \theta} 
=\mp
  \sin(\Phi/2) \cos(\Phi/2) \frac{d\Phi}{d\theta} $
and it follows that 
\begin{align}
\frac{1}{p_+}\left(\frac{\partial p_+ } {\partial \theta} \right)^2
&=
 \left(\frac{\partial \Phi}{\partial\theta} \right)^2  \sin^2( \Phi / 2) 
\\
\frac{1}{p_-}\left(\frac{\partial p_- } {\partial \theta} \right)^2
&=
 \left(\frac{\partial \Phi}{\partial\theta} \right)^2
\cos^2( \Phi / 2) .
\end{align}
Adding these two terms above gives the FI.

Now let us consider the case for general $\phi$. Then
\begin{align}
    p_+ &= \frac{1}{2}
    (\cos \phi + \sin \phi e^{i\Phi})
    (\cos \phi + \sin \phi e^{-i\Phi})
    \notag\\
    &=
    \frac{1}{2}
    (
     1 + \sin \phi \cos \phi (e^{i\Phi}+e^{-i\Phi})
    )
    \notag\\
    &=
    \frac{1}{2}
    (
     1 +  \sin 2 \phi  \cos \Phi
    ).
\end{align}
We also have 
\begin{align}
    p_- &= \frac{1}{2}
    (\cos \phi - \sin \phi e^{i\Phi})
    (\cos \phi - \sin \phi e^{-i\Phi})
    \notag\\
    &=
    \frac{1}{2}
    (
     1 - \sin \phi \cos \phi (e^{i\Phi}+e^{-i\Phi})
    )
    \notag\\
    &=
    \frac{1}{2}
    (
     1 -  \sin 2 \phi  \cos \Phi
    ).
\end{align}
Hence it follows that
\begin{align}
    \frac{\partial}{\partial \theta} p_\pm
    &=\mp \frac{1}{2}\sin 2\phi \sin \Phi 
    \frac{\partial \Phi}{\partial \theta}.
\end{align}
Therefore the FI is
\begin{align}
    &
    \frac{\frac{1}{4} \sin^2 2\phi \sin^2 \Phi 
    \Bigl(\frac{\partial \Phi}{\partial \theta}\Bigr)^2}
    {\frac{1}{2}(1+\sin 2\phi \cos \Phi) }
    +
    \frac{\frac{1}{4} \sin^2 2\phi \sin^2 \Phi 
    \bigl(\frac{\partial \Phi}{\partial \theta}\Bigr)^2}
    {\frac{1}{2}(1-\sin 2\phi \cos \Phi) }
    \notag\\
    =&
    \frac{ \sin^2 2\phi \sin^2 \Phi 
    \Bigl(\frac{\partial \Phi}{\partial \theta}\Bigr)^2}
    {2(1+\sin 2\phi \cos \Phi) }
    +
    \frac{\sin^2 2\phi \sin^2 \Phi 
    \Bigl(\frac{\partial \Phi}{\partial \theta}\Bigr)^2}
    {2(1-\sin 2\phi \cos \Phi) }    
    \notag\\
    =&
    \frac{1}{2}\sin^2 2\phi \sin^2 \Phi 
    \Bigl(\frac{\partial \Phi}{\partial \theta}\Bigr)^2
    \frac
    {1-\sin 2\phi \cos \Phi + 1 + \sin 2\phi \cos \Phi}
    {1- \sin^2 2\phi \cos^2 \Phi}
    \notag\\
    =&
    \frac{\sin^2 2\phi \sin^2 \Phi }{1- \sin^2 2\phi \cos^2 \Phi}
    \Bigl(\frac{\partial \Phi}{\partial \theta}\Bigr)^2.
\end{align} 
This gives the result.
\end{proof}

The state at the end of a successful run of each iteration of ECSense has the form
$a' |0_L\> + b' e^{i\Phi} |1_L\>$, 
where $a' = \cos \phi$ and $b' = \sin \phi$.
In the phase accumulation stage of the protocol, the amplitude shift becomes very far from 1.
Fortunately, in the amplitude rebalancing stage, we can get the amplitude shift to move back close to a constant factor of 1, without compromising the information of the phase,
provided that $g =\Theta(N^\alpha)$ is not too large. 

Denote the FI at the end of each iteration in ECSense as ${F'}$. Then, almost surely, this FI can be calculated as follows. 
When $\phi = \pi/4 + \epsilon$, we have
$\sin 2 \phi 
=  \sin( \pi / 2+ 2 \epsilon) 
=  \sin( \pi / 2 ) \cos (2 \epsilon) + \sin(2\epsilon) \cos(\pi/2)
= \cos (2 \epsilon).$
When $\epsilon = 0$, the prefactor $f =  \frac{\sin^2 2\phi \sin^2 \Phi}{(1-\sin^2 2\phi \cos^2 \Phi)}$ is 
$f=\sin^2\Phi / (1-\cos^2\Phi) =  1$.
Let $\delta_\phi > 0$. 
When $|\epsilon| \le \delta_\phi$,
it follows that
$|\cos(2 \epsilon) - (1-4\epsilon^2)| \le \frac{16}{3}\delta_\phi^4$.
Hence, for small $\delta_\phi$ and small $\Phi$, we have
\begin{align} 
\left|{F'}  - ( 1 -4 \delta_\phi^2 \sin^{-2} \Phi)  
\left(\frac{\partial\Phi}{\partial b}\right)^2
\right|\le 
\frac{16 \delta_\phi}{3}
\left(\frac{\partial\Phi}{\partial b}\right)^2.
\label{prefactor-perturbation}
\end{align} 

Hence the FI of at the end of each iteration of ECSense is almost surely
\begin{align}
    {F'}  = 
    \Theta(   (r (gb\tau)^2 g \tau)^2 ) .
\end{align}

We like to solve a linear recurrence of the form 
\begin{align}
    b_k = a b_{k-1} + c,
\end{align}
with the initial condition $b_0 = 1/2$.
The characteristic equation for the homogeneous part of the recurrence is 
$r - a = 0$, and the solution is trivially $r=a$. For the particular solution, we try
$d_k = D$ for some constant $D$. 
Then we get $D = a D + c$ which implies that 
$D(1-a)=c$ and so $D = c/(1-a)$.
Therefore the general solution is of the form 
$b_k = A a^k + D$ for some constants $A,D$. 
Solving for the initial condition, we get 
$1/2 = A + D = A + c/(1-a)$.
Therefore, $A = 1/2 - c/(1-a)$.

\section{Technical calculations: Deletions}
%%%%%%%%%%%%%%%%%%%%%%
Here, we evaluate what happens to an arbitrary symmetric state after $t$ deletions occur.

\begin{proof}[Proof of Lemma \ref{lem:delete}]
Recall that the input state before deletions is $|\psi\> = \sum_w a_w |D^N_w\>$ in the Dicke basis. 
We can alternatively write 
$|\psi\> = \sum_w \bar a_w |H^N_w\>$
where $\bar a_w = a_w / \sqrt{\binom N w}$.

We first determine the domain of the summation index in the Vandermonde decomposition
\begin{align}
|H^N_w\> = \sum_a |H^t_a\> \otimes |H^{N-t}_{w-a}\>.
\end{align}
We must have $0 \le a\le t$ and $0 \le w-a\le N-t$.
Together, these inequalities are equivalent to 
$\max\{ 0, t+w - N \}\le  a \le \min \{ w,t \} $. Hence it follows that 
\begin{align}
|H^N_w\> = \sum_{a \in A_w} |H^t_a\> \otimes |H^{N-t}_{w-a}\>,
\label{eq:vandermonde}
\end{align}
where 
\begin{align}
A_w = \{ a: 0,t+w-N\le a \le w,t\} .
\end{align}
Applying the definition of the partial trace and its linearity, we see that $\tr_t(|\psi\>\<\psi|)$ is equal to 
\begin{align}
&
\sum_{{\bf x} \in \{ 0,1\}^t} 
\sum_{w,v =0 }^N  \bar a_w \bar a_v^*
\left(\< {\bf x} | \otimes I^{\otimes N-t}\right)
|H^N_w\>\<H^N_v|
\notag\\
&\quad\quad\quad\quad\quad\quad\quad\quad\quad\quad\times
\left(| {\bf x} \> \otimes I^{\otimes N-t}\right).
\end{align}
Using \eqref{eq:vandermonde}, we find that $|H^N_w\>\<H^N_v|
$ is equal to 
\begin{align}
 \sum_{a \in A_w}  \sum_{b \in A_v}
  |H^t_a\> \<H^t_b| 
 \otimes |H^{N-t}_{w-a}\>\<H^{N-t}_{v-b}|.
 \end{align}
 Hence we find that $\tr_t(|\psi\>\<\psi|)$ is equal to 
 \begin{align}
\sum_{{\bf x} \in \{ 0,1\}^t} 
\sum_{w,v =0 }^N 
 \sum_{\substack{ a \in A_w\\b \in A_v}}
\bar a_w \bar a_v^*
\< {\bf x}  |H^t_a\> \<H^t_b|  {\bf x} \> 
|H^{N-t}_{w-a}\>\<H^{N-t}_{v-b}| .
\end{align}
Now note that $\< {\bf x}  |H^t_a\> = 1$ when the Hamming weight of ${\bf x}$ is equal to $a$, and $\< {\bf x}  |H^t_a\> = 0$ otherwise. 
Hence for non-trivial contributions to the above sum, we must have $a=b.$
Using these facts and moving the summation over ${\bf x}$ inside, we find that 
$ \tr_t(  |\psi\> \<\psi| )$ equals to 
 \begin{align}
 \sum_{a = 0}^t 
 \binom t a 
 \sum_{w,v=0}^N \bar a_w \bar a_v^*
 \delta[a \in A_{w}]  \delta[a \in A_{v}]
 |H^{N-t}_{w-a}\>\<H^{N-t}_{v-a} |. 
 \end{align}
 Hence
 \begin{align}
 \tr_t (|\psi\>\<\psi| )
 = 
 \sum_{a=0}^t \binom t a |\phi_a\>\<\phi_a|,
 \end{align}
 where 
 \begin{align}
 |\phi_a\> 
 &= \sum_{w} \delta[a \in A_w] \bar a_w |H^{N-t}_{w-a}\>\notag\\
 &= \sum_{w=a}^{N-t+a}  \bar a_w |H^{N-t}_{w-a}\>\notag\\
 &= \sum_{w=a}^{N-t+a} \bar a_w \sqrt{\binom{N-t}{w-a}}|D^{N-t}_{w-a}\>\notag\\
  &= \sum_{w=a}^{N-t+a}   a_w 
  \frac
  {\sqrt{\binom{N-t}{w-a}}}
  {\sqrt{\binom{N}{w}}}
  |D^{N-t}_{w-a}\>.
 \end{align}
 Since $|\phi_a\> = |\psi\>_a$, the result follows.
\end{proof}

\section{Technical calculations: Amplitude damping errors}
Here, we study what happens after amplitude damping errors afflict a pure symmetric state.
\begin{proof}[Proof of Lemma \ref{lem:AD-representation}]
For non-negative integer $x$ such that $x \le N$, 
let $({\bf 1}^{x}, {\bf 0}^{N-x})$ denote a length $N$ binary vector that has its first $x$ bits equal to 1 and the remaining bits equal to 0. Let ${\bf 1}^{x}$ denote a ones vector of length $x$ and ${\bf 0}^{N-x}$ be a zeros vector of length $N-x$.
Then
\begin{align}
&A_{({\bf 1}^{x}, {\bf 0}^{N-x})} |H^N_w\> \notag\\
=&
\sum_{j=0}^x (A_{{\bf 1}^{x}} |H^N_j\>) \otimes ( A_{{\bf 0}^{N-x}}  |H^{N-x}_{w-j}\>)
\notag\\
=& 
 (A_{{\bf 1}^{x}} |H^N_x\>) \otimes ( A_{{\bf 0}^{N-x}}  |H^{N-x}_{w-x}\>)
\notag\\
=& 
 ( \gamma^{x/2} |H^N_0\>) \otimes (  (1-\gamma)^{(w-x)/2}  |H^{N-x}_{w-x}\>)
\notag\\
=& \sqrt{q_w(x)} |0\>^{\otimes x} \otimes |H^{N-x}_{w-x}\>,
\end{align}
where $q_w(x) =   \gamma^{x} (1-\gamma)^{w-x}$.
From this it follows that
\begin{align}
A_{({\bf 1}^{x}, {\bf 0}^{N-x})} |D^N_w\>
&=
\frac{1}{\sqrt{\binom N w}}A_{({\bf 1}^{x}, {\bf 0}^{N-x})} |H^N_w\> \notag\\
&=
\frac{\sqrt{q_w(x)}}{\sqrt{\binom N w}} |0\>^{\otimes x} \otimes |H^{N-x}_{w-x}\> \notag\\
&=
\sqrt{\frac{\binom {N-x}{w-x}}{\binom N w}}\sqrt{q_w(x)} |0\>^{\otimes x} \otimes |D^{N-x}_{w-x}\>\notag\\
&=
\sqrt{\frac{\binom {w}{x}}{\binom N x}}\sqrt{q_w(x)} |0\>^{\otimes x} \otimes |D^{N-x}_{w-x}\>\notag\\
&=
\sqrt{\frac{p_{w}(x)}{\binom N x}} |0\>^{\otimes x} \otimes |D^{N-x}_{w-x}\>.
\end{align}
When $x > w$ we trivially have $A_{({\bf 1}^{x}, {\bf 0}^{N-t})} |D^N_w\> = 0$. 
Since $\binom w x = 0$ and $p_w(x) = 0$ whenever $x > w$, for all $x = 0,1,\dots, N$, we have
\begin{align}
A_{({\bf 1}^{x}, {\bf 0}^{N-t})} |D^N_w\>
=
\sqrt{\frac{p_{w}(x)}{\binom N x}}|0\>^{\otimes x} \otimes |D^{N-x}_{w-x}\>.
\end{align}
Hence,
\begin{align}
&\mathcal A_{N,\gamma}( |D^N_w\>\<D^N_{w'}| )\notag\\
=&
\sum_{x=0}^{N} \sum_{|P| = x}
\frac{\sqrt{p_w(x) p_{w'}(x)} }{\binom N x }
{\rm Ins}_P ( |D^N_w\>\<D^N_{w'}| ) .
\end{align}
Hence,
\begin{align}
&\mathcal A_{N,\gamma}( |\psi\>\<\psi| )\notag\\
=&
\sum_{w,w'=0}^N a_w a_{w'}^* 
\mathcal A_{N,\gamma}( |D^N_w\>\<D^N_{w'}|  ) \notag\\
=&
\sum_{w,w'=0}^N a_w a_{w'}^* 
\sum_{x=0}^{N} \sum_{|P| = x}
\frac{\sqrt{p_w(x) p_{w'}(x)} }{\binom N x }
{\rm Ins}_P ( |D^N_w\>\<D^N_{w'}| ) .
\end{align}
Exchanging the orders of the summations, the lemma follows after substituting the definition of $|\phi_x\>$.
\end{proof}

\section{Technical calculations: Binomial identities}
In this section, let $S(n,s)$ denote the Stirling number of the second kind,
and denote $k_{(j)}$ as the falling factorial $k\dots (k-j+1)$.

The Stirling numbers of the second kind can be calculated according the formula
\begin{align}
    S(n,s) = \frac{1}{s!} \sum_{j=0}^s (-1)^{s-j} \binom s j j^n.
\end{align}
When $n=0$. $S(n,0)=1$ and $S(n,s)=0$ for all positive integers $s$.
When $n\ge 1$, $S(n,s)=0$ for all integers $n > s$, and when $n\ge 3$, we have
\begin{align}
    S(n,0) &= 0,\label{Sn0}\\
    S(n,1) &= 1,\label{Sn1}\\
    S(n,2) &= 2^{n-1}-1,\label{Sn2}\\
    S(n,3) &= (3^{n}-2^n+1)/2\label{Sn3}.
\end{align}
 The Stirling numbers of the second kind allow us to express monomials as a linear combination of falling factorials, in the sense that
 \begin{align}
     k^s = \sum_{j=0}^s S(s,j) k_{(j)}.
     \label{monomial as falling}
 \end{align}
Next, we have the following lemma.
\begin{lemma}[{\cite[Lemma 1]{ouyang2014permutation}}]
\label{lem:binoms1}
Let $n$ and $s$ be non-negative integers. If $s < n$, then
\begin{align}
\sum_{k\ {\rm even}} \binom n k  k^s 
&= \sum_{k\ {\rm odd}} \binom n k  k^s.
\end{align}    
\end{lemma}
Here we extend the above lemma for larger $s$.
\begin{lemma}
\label{lem:binom even odd monomial}
Let $n$ and $s$ be non-negative integers. Let $s \ge n$.
When $n$ is even, then 
\begin{align}
\sum_{k\ {\rm even}} \binom n k  k^s   + S(s,n) n!
&= 
\sum_{k\ {\rm odd}} \binom n k  k^s .
\end{align}
When $n$ is odd, then
\begin{align}
\sum_{k\ {\rm even}} \binom n k  k^s   
&= 
\sum_{k\ {\rm odd}} \binom n k  k^s + S(s,n) n!.
\end{align}
\end{lemma} 
\begin{proof}
   We use \eqref{monomial as falling}.
    Next we note that whenever $k$ is a non-negative integer and if $k \le n$, then 
    $k_{(j)} = 0$ for all $j >n$. 
    Hence in this situation,
    $k^s = \sum_{j=0}^n S(s,j) k_{(j)}$.
Note that 
\begin{align}
\sum_{k\ {\rm even}} \binom n k  k^s  
&=
\sum_{k\ {\rm even}} \binom n k
 \sum_{j=0}^n S(s,j) k_{(j)}\notag\\
&= 
 \sum_{j=0}^n S(s,j)
 \sum_{k\ {\rm even}} \binom n k
 k_{(j)}
 \end{align}
 When $n$ is even, we can use Lemma \ref{lem:binoms1} to get
 \begin{align}
 \sum_{k\ {\rm even}} \binom n k  k^s  
&= 
S(s,n)n! + 
 \sum_{j=0}^{n-1} S(s,j)
 \sum_{k\ {\rm even}} \binom n k
 k_{(j)}\notag\\
 &= 
S(s,n)n! + 
 \sum_{j=0}^{n-1} S(s,j)
 \sum_{k\ {\rm odd}} \binom n k
 k_{(j)}\notag\\
 &= 
S(s,n)n! + 
 \sum_{k\ {\rm odd}} \binom n k
 k^s.
\end{align}
Similarly, when $n$ is odd, we get
\begin{align}
 \sum_{k\ {\rm odd}} \binom n k  k^s  
&= 
S(s,n)n! + 
 \sum_{k\ {\rm even}} \binom n k
 k^s.
\end{align}    
\end{proof}
Next we evaluate binomial sums weighted by falling factorials and exponentials.
\begin{lemma}
\label{lem:binom-exp-monomial-falling}
    Let $n,s$ be non-negative integers, and $y' = y/2$ where $y \in \mathbb R$. Then
\begin{align}
&\sum_{k\ {\rm even}}^n \binom n k k_{(s)} e^{iky}  \notag\\
    =&
    n_{(s)} 
e^{iy'(n+s) } 2^{n-s-1} 
( \cos^{n-s}y' + 
(-1)^n i^{n-s} \sin^{n-s} y' )
,    
    % \frac{ (1+e^{iy})^{n-s}  
    % +(-1)^{s}
    % (1-e^{iy})^{n-s}      }
    % {2},
\end{align}
and 
\begin{align}
&\sum_{k\ {\rm odd}}^n \binom n k k_{(s)} e^{iky}  \notag\\
    =&
    n_{(s)} 
e^{iy'(n+s)} 2^{n-s-1} 
( \cos^{n-s}y' -
(-1)^n i^{n-s} \sin^{n-s} y' )
.
    % n_{(s)}e^{i s y} 
    % \frac{ (1+e^{iy})^{n-s}  
    % -(-1)^{s}
    % (1-e^{iy})^{n-s}      }
    % {2}.
\end{align}
When $y=u+iv$ where $u,v \in \mathbb R$,
we have 
\begin{align}
&\sum_{k\ {\rm even}}^n \binom n k k_{(s)} e^{iky}  \notag\\
    =&
    n_{(s)} 
e^{isy}
( (1+e^{i u-v})^{n-s} + 
(-1)^s (1-e^{iu-v})^{n-s}
,     
\end{align}
and 
\begin{align}
&\sum_{k\ {\rm odd}}^n \binom n k k_{(s)} e^{iky}  \notag\\
    =&
    n_{(s)} 
e^{i s y }
( (1+e^{i u-v})^{n-s} - 
(-1)^s (1-e^{iu-v})^{n-s}. 
\end{align}
\end{lemma}
Note that when $s > n$, the above summations are equal to zero because of the falling factorials.
\begin{proof}
    To prove this lemma, we use the method of generating functions. Consider the generating function $f(x) = (1+x)^n.$
Expanding $f(x)$ as a power series in $x$ using the binomial theorem, we see that 
\begin{align}
    f(x) = \sum_{k=0}^n \binom n k x^k.
\end{align}
Then by taking $s$ formal derivatives of $f(x)$ with respect to $x$, we get that for all non-negative integers $n'$ that
\begin{align}
    \frac{ d^{s} f(x)} {dx^{s} } 
    &= n_{(s)} (1+x)^{n-s} \notag\\
    &= \sum_{k=0}^n \binom n k k_{(s)} x^{k-s}\notag\\
    &= \sum_{k=0}^n \binom n k k_{(s)} x^{k-s}.
\end{align}
Next, we make the substitution 
$x= e^{iy}$ to get
\begin{align}
    n_{s} (1+e^{iy})^{n-s} 
    &= \sum_{k=0}^n \binom n k k_{(s)} (e^{i y})^{k-s} \notag\\
    &= \sum_{k=0}^n \binom n k k_{(s)} e^{i( k-s)y}.
\end{align}
Rearranging the terms in the above equation, we get
\begin{align}
\sum_{k=0}^n \binom n k k_{(s)} e^{iky}
    &= 
    n_{(s)} (1+e^{iy})^{n-s}  e^{i s y}.
\end{align}
We can also make the substitution $x=-e^{iy}$ to get
\begin{align}
\sum_{k=0}^n \binom n k k_{(s)} e^{iky} (-1)^k
    &= 
    n_{(s)} (1-e^{iy})^{n-s}  e^{i s y} (-1)^{s}.
\end{align}
By the method of generating functions,
it follows that 
\begin{align}
&\sum_{k \ {\rm even}}
\binom n k k_{(s)} e^{iky} (-1)^k\notag\\
    = &
\frac 1 2
\left(\sum_{k=0}^n \binom n k k_{(s)} e^{iky}
+    
\sum_{k=0}^n \binom n k k_{(s)} e^{iky} (-1)^k\right),
\end{align}
and 
\begin{align}
&\sum_{k \ {\rm odd}}
\binom n k k_{(s)} e^{iky} (-1)^k\notag\\
    = &
\frac 1 2
\left(\sum_{k=0}^n \binom n k k_{(s)} e^{iky}
-    
\sum_{k=0}^n \binom n k k_{(s)} e^{iky} (-1)^k\right).
\end{align}
Hence
\begin{align}
&\sum_{k \ {\rm even}}
\binom n k k_{(s)} e^{iky} (-1)^k\notag\\
    = &
\frac{n_{(s)}
(  (1+e^{iy})^{n-s}  e^{i s y}
+ 
 (1-e^{iy})^{n-s}  e^{i s y} (-1)^{s}
}{2}
\end{align}
and
\begin{align}
&\sum_{k \ {\rm odd}}
\binom n k k_{(s)} e^{iky} (-1)^k\notag\\
    = &
\frac{n_{(s)}
(  (1+e^{iy})^{n-s}  e^{i s y}
- 
 (1-e^{iy})^{n-s}  e^{i s y} (-1)^{s}
}{2}.
\end{align}

Next we use 
\begin{align}
&\frac{
(1+e^{iy})^{n-s}
\pm
(-1)^s
(1-e^{iy})^{n-s}
}{2}  \notag\\
=&
e^{iy(n-s)/2} 2^{n-s-1} 
\notag\\
&\times
( \cos^{n-s}(y/2) \pm (-1)^s
(-i)^{n-s} \sin^{n-s} (y/2) )
\end{align}
to get the result.
\end{proof}

Next we evaluate binomial sums weighted by the product of monomials and exponentials.
\begin{lemma}
\label{lem:binom sum exp iky}
Let $n$ and $s$ be non-negative integers and $y'=y/2$ where $y \in \mathbb R$. Define $0^0 \coloneqq 1$. Then 
\begin{align}
    &2^{-n+1 }\sum_{k\ {\rm even}} \binom n k k^{s} e^{ik y}\notag\\
    =&
     \sum_{j=0}^{{\rm min}(s,n)} 
2^{-j} S(s,j) 
n_{(j)}e^{i (n+j)y'}\notag\\
&\quad \times
(\cos^{n-j} y' + (-1)^n i^{n-j} \sin^{n-j} y')
\end{align}
and
\begin{align}
    &2^{-n+1 }\sum_{k\ {\rm odd}} \binom n k k^{s} e^{ik y}\notag\\
    =&
\sum_{j=0}^{{\rm min}(s,n)} 
2^{-j} S(s,j) 
n_{(j)}e^{i (n+j)y'}\notag\\
&\quad \times
(\cos^{n-j} y'  - (-1)^n i^{n-j} \sin^{n-j} y').
\end{align}  
When $n,s$ are positive integers and $s<n$, this simplifies to 
\begin{align}
    &2^{-n+1 }\sum_{k\ {\rm even}} \binom n k k^{s} e^{ik y}\notag\\
    =&
     \sum_{j=1}^{s} 
2^{-j} S(s,j) 
n_{(j)}e^{i (n+j)y'}\notag\\
&\quad \times
(\cos^{n-j} y' + (-1)^n i^{n-j} \sin^{n-j} y')
\end{align}
and
\begin{align}
    &2^{-n+1 }\sum_{k\ {\rm odd}} \binom n k k^{s} e^{ik y}\notag\\
    =&
\sum_{j=1}^{s} 
2^{-j} S(s,j) 
n_{(j)}e^{i (n+j)y'}\notag\\
&\quad \times
(\cos^{n-j} y'  - (-1)^n i^{n-j} \sin^{n-j} y').
\end{align}  
\end{lemma}
For example, we have for $s = 0$,
\begin{align}
   & 2^{-n+1 }\sum_{k\ {\rm even}} \binom n k  e^{ik y}\notag\\
    =&
e^{i n y'}
(\cos^{n} y' +(-1)^n i^{n} \sin^{n} y'), 
\\
   & 2^{-n+1 }\sum_{k\ {\rm odd}} \binom n k  e^{ik y}\notag\\
    =&
e^{i n y'}
(\cos^{n} y' - (-1)^n i^{n} \sin^{n} y').
\end{align}
When $s=1$ and $n\ge 1$ we have
\begin{align}
%% np = 1
   & 2^{-n+1 }\sum_{k\ {\rm even}} \binom n k k e^{ik y}\notag\\
     =&2^{-1} n e^{i (n+1) y'}
(\cos^{n-1} y' +(-1)^n i^{n-1} \sin^{n-1} y') \notag\\
   & 2^{-n+1 }\sum_{k\ {\rm odd}} \binom n k k e^{ik y}\notag\\
     =&2^{-1} n e^{i (n+1) y'}
(\cos^{n-1} y' -(-1)^n i^{n-1} \sin^{n-1} y') .
\end{align} 
When $s=2$ and $n\ge 2$,
we have 
\begin{align}
%% np = 1
   & 2^{-n+1 }\sum_{k\ {\rm even}} \binom n k k^2 e^{ik y}\notag\\
     =&\frac n 2  e^{i (n+1) y'}
(\cos^{n-1} y' +(-1)^n i^{n-1} \sin^{n-1} y') \notag\\
&+
\frac{  n(n-1)}{4} e^{i (n+2) y'}
% \notag\\
% &\quad \times
(\cos^{n-2} y' +(-1)^n i^{n-2} \sin^{n-2} y') \notag\\
   & 2^{-n+1 }\sum_{k\ {\rm odd}} \binom n k k e^{ik y}\notag\\
     =&\frac{n}{2} e^{i (n+1) y'}
(\cos^{n-1} y' -(-1)^n i^{n-1} \sin^{n-1} y') 
\notag\\
&+
\frac{  n(n-1)}{4} e^{i (n+2) y'}
% \notag\\
% &\quad \times
(\cos^{n-2} y' -(-1)^n i^{n-2} \sin^{n-2} y')  .
\end{align}

\begin{proof}[Proof of Lemma \ref{lem:binom sum exp iky}]
We can express monomials $k^s$ in terms of falling factorials:
\begin{align}
    k^s  = \sum_{j=0}^s S(n,j) k_{(j)}.
\end{align}
Hence we see that 
\begin{align}
&2^{-n+1} 
\sum_{k\ {\rm even}} 
\binom n k k^{s} e^{iky}\notag\\
= &
2^{-n+1} 
\sum_{k\ {\rm even}} 
\binom n k 
\left(\sum_{j=0}^{s} 
S(s,j) k_{(j)}\right)
e^{iky}\notag\\
= &
2^{-n+1} 
\sum_{j=0}^{s} 
S(s,j) 
\sum_{k\ {\rm even}} 
\binom n k 
k_{(j)}e^{iky} 
.
\end{align}
Using Lemma \ref{lem:binom-exp-monomial-falling} gives the first result.
Similarly, we use Lemma \ref{lem:binom-exp-monomial-falling} to get the second result. 
\end{proof}
By taking the limit of $y$ to approach 0 in Lemma \ref{lem:binom sum exp iky}, we have the following corllary.
\begin{corollary}\label{coro:binom sum}
    Let $n$ and $s$ be positive integers. Then 
    $    2^{-n}
        \sum_{k=0}^n
        \binom n k k^s
        = 
        \sum_{j=0}^s
        2^{-j}S(s,j) n_{(j)}.$
\end{corollary}

\section{Technical calculations: Sandwiches of shifted gnu logical states}
\label{app:sandwiches}
In this section we evaluate sandwiches of shifted gnu logical states, which are expressions of the form $\<j_L | A | j_L\>$ where $|j_L\>$ are logical codewords of shifted gnu logical states and $A$ is an appropriately sized complex matrix.
We begin by considering
$\<j_L | {\hat J^z} | j_L\>$
and
$\<j_L | ({\hat J^z})^2 | j_L\>$.
We make a mild assumption that $n \ge 3$. We note that
\begin{align}
\<0_L | {\hat J^z} | 0_L\>
&= (N/2-s) - 2^{-n+1} \sum_{j {\rm \ even}} \binom{n}{j} gj ,\\
\<1_L | {\hat J^z} | 1_L\>
&= (N/2-s) - 2^{-n+1} \sum_{j {\rm \ odd}} \binom{n}{j} gj .
\end{align}
Now 
$\sum_{j = 0}^n j\binom{n} {j} = 2^{n-1} n$.
Whenever $n \ge 2$, we have
$2^{-n+1} \sum_{j {\rm \ odd}} \binom{n}{j} j =
2^{-n+1} \sum_{j {\rm \ even}} \binom{n}{j} j
= \frac{n}{2}$.
Hence
\begin{align}
\<0_L | {\hat J^z} | 0_L\>
&= \<1_L | {\hat J^z} | 1_L\>
= N/2-s- gn/2 \label{simple-J-sandwich}.
\end{align}

Now 
$\sum_{j = 0}^n j^2 \binom{n} {j}  = 2^{n-2} n ( n+1).$ 
When $n\ge 3$,
$2^{-n+1}\sum_{j {\rm \ odd}}^n j^2 \binom{n} {j}  =
2^{-n+1}\sum_{j {\rm \ even}}^n j^2 \binom{n} {j}  = n ( n+1)/4$.
Next, we evaluate the following.
\begin{align}
&\<0_L | ({\hat J^z})^2 | 0_L\>\notag\\
=&  2^{-n+1} \sum_{j {\rm \ even}} 
\binom{n}{j} 
(N/2 - (gj+s) )^2,\notag\\ 
=&  2^{-n+1} \sum_{j {\rm \ even}} 
\binom{n}{j} 
((N/2-s)^2 +(2s - N )gj + g^2j^2 ) )\notag\\
=&  
(N/2-s)^2 +\frac{(2s - N )gn}{2} + \frac{g^2n(n+1)}{4} \label{J2-0-sandwich}.
\end{align}
Similarly, 
\begin{align}
&\<1_L | ({\hat J^z})^2 | 1_L\>
%\notag\\ 
= %&  
(N/2-s)^2 +\frac{(2s - N )gn}{2} + \frac{g^2n(n+1)}{4} 
\label{J2-1-sandwich} .
\end{align}
Generalizing the above, we have the following lemma.
\begin{lemma}
\label{lem:Jzj sandwich}
 Let $|0_L\>$ and $|1_L\>$ be logical codewords of a shifted gnu code, and let $j$ be a non-negative integer such that $j<n$.  
 Then we have 
 $\<0_L|((s-N/2)I + \hat J^z)^j|0_L\>
 =
 \<1_L|((s-N/2)I +\hat J^z)^j|1_L\> = 
 g^j  \sum_{\ell =0 
 }^j 2^{-\ell} S(j,\ell)n_{(\ell)} $.
\end{lemma}
Furthermore, when $n$ is constant, we have 
\begin{align}
\<0_L|(\hat K)^j|0_L\>
=
\<1_L|(\hat K)^j|1_L\>
=g^j O(1).
\end{align}
\begin{proof}
Note that
\begin{align}
    \<0_L|((s-N/2)I +\hat J^z)^{j}|0_L\> 
    &= 2^{-n+1} \sum_{k \ {\rm even}}
    \binom n k g^j k^j \notag   
\end{align}
which is equal to $g^j 2^{-n+1} \sum_{k \ {\rm even}}
    \binom n k  k^j,$
and we get a similar sum (but over the odd indices of $k$) when we evaluate $ \<1_L|((s-N/2)I +\hat J^z)^{j}|1_L\>$. Next, when $j<n$, these two binomial sums are the same because of the quantum error correction criterion that the logical codewords satisfy. Then we use Corollary \ref{coro:binom sum} to obtain the first result. 
The second result follows trivially from expanding 
$\hat K^j =  \sum_{m=0}^j \binom j m (\hat J + (s-N/2)I)^m (gn/2)^{j-m}$, using the first result, and rearranging terms in the summation.
\end{proof}

Now we prove Lemma \ref{lem:sandwich}, which involves evaluating the quantities 
$\<0_L|U_\Delta|0_L\>$ and $\<1_L|U_\Delta|1_L\>$.
\begin{proof}[Proof of Lemma \ref{lem:sandwich}]
Now ${\hat J^z} |D^N_w\> = (N/2 - w)  |D^N_w\>$.
Hence 
\begin{align}
U_\Delta |D^N_w\> 
&= \sum_{w}a_w \exp(-i \Delta(N/2 - w)) |D^N_w\>
\notag\\
&= e^{-i \Delta N /2}\sum_{w}a_w e^{i \Delta w} |D^N_w\>.
\label{U-Dicke}
\end{align}
Using \eqref{U-Dicke} and the definition of logical codewords for shifted gnu codes with shift $s$, we see that 
\begin{align}
\<0_L|  U_\Delta  |0_L\>
&=
 e^{-i \Delta N /2} 2^{-n+1} \sum_{\substack{0 \le k \le n \\ {k\ {\rm even}}}} 
\binom n k e^{i(gk+s)\Delta},\notag\\
\<1_L|  U_\Delta  |1_L\>
&=
 e^{-i \Delta N /2} 2^{-n+1} \sum_{\substack{0 \le k \le n \\ {k\ {\rm odd}}}} 
\binom n k e^{i(gk+s)\Delta}\notag.
\end{align}
Note that
\begin{align}
&\sum_{k\ {\rm even}}  {\binom n k} e^{igk\Delta}\notag\\
=& \frac{ (1 + e^{ig\Delta} )^n +  (1 - e^{ig\Delta} )^n } {2}\notag\\
=& \frac{ (2e^{ ig\Delta/2}\cos(g\Delta/2) )^n +  (-2ie^{ ig\Delta/2}\sin(g\Delta/2) )^n } {2}
\notag\\ 
=& 2^{n-1} e^{ ign\Delta/2} 
(\cos^n(g\Delta/2)  +  (-i)^n \sin^n(g\Delta/2) ).\notag
\end{align}
Similarly,
\begin{align}
&\sum_{k\ {\rm odd}}  {\binom n k} e^{igk\Delta} \notag\\
=& \frac{ (1 + e^{ig\Delta} )^n -   (1 - e^{ig\Delta} )^n } {2}\notag\\
=& 
2^{n-1} e^{ ign\Delta/2} 
(\cos^n(g\Delta/2)  -  (-i)^n \sin^n(g\Delta/2) ).\notag
\end{align}
 Hence 
 \begin{align}
     \frac{\<0_L|  U_\Delta  |0_L\>}{e^{-i \Delta (N /2-s)}
 e^{ ign\Delta/2} }
&=
 \cos^n(g\Delta/2)  +  (-i)^n \sin^n(g\Delta/2) .\notag
 \end{align}
 Similarly,
 \begin{align}
     \frac{\<1_L|  U_\Delta  |1_L\>}{e^{-i \Delta (N /2-s)}
 e^{ ign\Delta/2} }
&=
 \cos^n(g\Delta/2)  -  (-i)^n \sin^n(g\Delta/2) .\notag
 \end{align}
The result follows.
\end{proof}
Hence we have shown that 
\begin{align}
    &\<j_L|U_\Delta|j_L\>
    = \phi_{n,j}(\Delta)\notag\\
    =& e^{-i \Delta (N /2-s)
    } e^{ i n x}
    \left(
    \cos^n x  + (-1)^j  (-i)^n \sin^n x \right),
\end{align}
where $x=  g \Delta/2$.
Note that when $n=2$,
 the above has a simple expression. 
 Namely, for $n=2$, we have
 \begin{align}
   \phi_{2,0}(\Delta) 
    =& e^{-i \Delta (N /2-s)
    } e^{ i 2 x}
    \left(
    \cos^2 x - \sin^2 x \right) \notag\\
     =&
    e^{-i \Delta (N /2-s)
    } e^{ i 2 x}
    \cos 2 x\\    
    \phi_{2,1}(\Delta) 
    =& e^{-i \Delta (N /2-s)
    } e^{ i 2 x}
    \left(
    \cos^2 x + \sin^2 x \right) \notag\\
     =& e^{-i \Delta (N /2-s)
    } e^{ i 2 x}.    
 \end{align}
 When $n=1$, we have
 \begin{align}
   \phi_{1,0}(\Delta) 
    =& e^{-i \Delta (N /2-s)
    } e^{ i   x}
    \left(
    \cos  x - i \sin  x \right) \notag\\
     =&
    e^{-i \Delta (N /2-s)
    }   \\    
    \phi_{1,1}(\Delta) 
    =& e^{-i \Delta (N /2-s)
    } e^{ i   x}
    \left(
    \cos  x + i \sin  x \right) \notag\\
     =& e^{-i \Delta (N /2-s)
    } e^{ i 2 x}.    
 \end{align}

Next we prove the following recursion relationship 
\begin{lemma}
\label{lem:phi-recursion}
For positive integer $n,N$ and integers $j=0,1$ and $s$, and for real $\Delta, b$, the expression $\frac{\partial }{\partial \Delta} (e^{i b \Delta} \phi_{n,j}(\Delta) )$ is equal to 
\begin{align}
-i\bigl(\frac{N}{2}-s -b \bigr) e^{i b \Delta } \phi_{n,j}(\Delta)
+ \frac{ign}{2}e^{ib\Delta}e^{ig\Delta}\phi_{n-1,j\oplus 1}(\Delta),\notag
\end{align}    
where the notation $\oplus$ denotes additional modulo 2.
\end{lemma}
\begin{proof}
Let $x=g\Delta/2$.
For $n$, we have
\begin{align}
    &\frac{\partial }{\partial \Delta} (e^{i b \Delta}  \phi_{n,j}(\Delta) )
    =
-i(N/2-s-b)e^{i b \Delta } \phi_{n,j}(\Delta)\notag\\
&
+(i g n /2)e^{i b \Delta } \phi_{n,j}(\Delta)
%\notag\\&
+
e^{-i \Delta (N /2-s)
    } e^{i b \Delta } e^{ i n x} (gn/2)
    \notag\\
    &\times
    \left(
    -\cos^{n-1} x \sin x  + (-1)^j  (-i)^n \sin^{n-1} x \cos x \right)
    \notag\\
&=
-i(N/2-s-b) e^{i b \Delta}\phi_{n,j}(\Delta)
+e^{-i \Delta (N /2-s)
    } e^{i b \Delta } \times
  \notag\\
&e^{ i n x} (gn/2) \left(
\cos^{n-1} x(i \cos x-\sin x) 
+\right.\notag\\
&\left.
(-1)^j(-i)^n
\sin^{n-1} x (i \sin x +\cos x)
\right).
\end{align}
Next it suffices to show that 
\begin{align}
    &-\cos^{n-1}x \sin x + (-1)^j(-i)^n \sin^{n-1} x \cos x\notag\\
    =& e^{i x}
    (\cos^{n-1}x + (-1)^j(-i)^{n-1} \sin^{n-1} x ).
\end{align}

Now $i \cos x - \sin x=  i ( \cos x + i\sin x) = i e^{i x}$ 
and 
    $i \sin x + \cos x=  e^{ix}= i(-i)e^{ix}=-i(-i)^{-1}e^{ix}$ .
Hence we have 
\begin{align}
    & 
\cos^{n-1} x(i \cos x-\sin x) 
+ \notag\\
& 
(-1)^j(-i)^n
\sin^{n-1} x (i \sin x +\cos x)
 \notag\\
&
=
i e^{ix}\left(
\cos^{n-1} x  
+ 
(-1)^{j+1}(-i)^{n-1}
\sin^{n-1} x  
\right).
\end{align}
from which the result follows.
\end{proof}
Hence, when $b=N/2-s$, we have the simple recursion relation
\begin{align}
 &\frac{\partial }{\partial \Delta} (e^{i b \Delta} \phi_{n,j}(\Delta) )
 =
 \frac{ign}{2}e^{ig\Delta} e^{i b \Delta} \phi_{n-1,j\oplus 1}(\Delta), 
\end{align}
Given shifted gnu logical codewords $|0_L\>$ and $|1_L\>$ on $N$ qubits, we like to emphasize the identity
\begin{align}
   & \<0_L|e^{-i\hat K
    b \tau}|0_L\> =
    \cos^n(g b \tau /2 )
    +(-i)^n \sin^n ( g b\tau/2 ),\\
 & \<1_L|e^{-i\hat K
    b \tau}|1_L\> =
    \cos^n(g b \tau /2 )
    -(-i)^n \sin^n ( g b\tau/2 ),
\end{align}
where $\hat K \coloneqq \hat J^z + (N/2-s-gn/2) I$ and $I$ denotes the identity operator.

Here, Lemma \ref{lem:Q-U-sandwich} evaluates $\<Q_j|U_\Delta|j_L\>$ for shifted gnu codes with $n \ge 3$. 
\begin{lemma}
\label{lem:Q-U-sandwich}
Let $|0_L\>$ and $|1_L\>$ be logical codewords of a shifted gnu code where $n \ge 3.$ Then,
\begin{align}
\<Q_j | U_\Delta |j_L\> = 
\frac {gn} 2 
\left(\phi_{n,j}(\Delta) - e^{ig\Delta}\phi_{n-1,j \oplus 1}(\Delta) \right) .
\label{lem11-first}
\end{align}
Furthermore, if $n$ is odd, then we have 
\begin{align}
\frac{2}{gn}\<Q_j | U_\Delta |j_L\> =& 
 e^{-i \Delta (N/2-s)} e^{inx} \left(
-i  \cos^{n-1}x  \sin x \right.\notag\\
 &
 \left. + i^{n-1} (-1)^j  \sin^{n-1}x  \cos x \right),
 \label{eq:83}
\end{align}
where $x= g \Delta/ 2$.
\end{lemma}
\begin{proof}[Proof of Lemma 
\ref{lem:Q-U-sandwich}
]
Consider $\<{Q_j}| U_\Delta | j_L\>$ for $j=0,1$.
We apply 
$\frac{\partial}{\partial \Delta} U_\Delta = (-i {\hat J^z}) U_\Delta $ together with the definitions of $|{Q_j}\>$ and find that
\begin{align}
\<{Q_j}| U_\Delta | j_L\>
&=
\<j_L| {\hat J^z} U_\Delta | j_L\>
-
\<j_L|  U_\Delta | j_L\>
\<j_L| {\hat J^z} | j_L\>\notag
\\
&=
i \frac{\partial}{\partial \Delta}  \<j_L|  U_\Delta | j_L\>
-
\<j_L|  U_\Delta | j_L\>
\<j_L| {\hat J^z} | j_L\>.
\end{align}
Using Lemma \ref{lem:sandwich} and \eqref{simple-J-sandwich}, we can write 
\begin{align}
\<{Q_j}| U_\Delta | j_L\>
&=
i \frac{\partial}{\partial \Delta} \phi_{n,j}(\Delta)
-
 \phi_{n,j}(\Delta)
 (N/2-s-gn/2).
\end{align}
Next, note the recursion relation
\begin{align}
\frac{\partial}{\partial \Delta}  \phi_{n,j}(\Delta)
= -i\bigl(\frac{N}{2}-s\bigr)\phi_{n,j}(\Delta) + \frac{ign}{2} e^{ig\Delta}\phi_{n-1,j \oplus 1}(\Delta),
\end{align}
where $0 \oplus 1 = 1$ and $1 \oplus 1 = 0$.
Hence,
\begin{align}
\<{Q_j}| U_\Delta| j_L\>
=& \left( \bigl(\frac{N}{2}-s\bigl) \phi_{n,j}(\Delta) -\frac{gn}{2} e^{ig\Delta}\phi_{n-1,j \oplus 1}(\Delta) \right)
\notag\\
 &- 
 \phi_{n,j}(\Delta)  (N/2 - s - gn/2),
\end{align}
and simplifying this, the first result follows.

Now let $x = g \Delta/2$, and note that 
\begin{align}
\cos x -e^{ix} &= 
%\cos x - (\cos x -  i \sin x) = 
-i \sin x, \notag\\
-i\sin x +e^{ix} &
%=-i \sin x + (\cos x -  i \sin x) 
= \cos x .
\end{align}
Now we evaluate $ \phi_{n,j}(\Delta) - e^{ig\Delta}\phi_{n-1,j \oplus 1}(\Delta) .$
Note that 
\begin{align}
&
 e^{i \Delta (N/2-s)} e^{-inx}\left(
 \phi_{n,j}(\Delta) -  e^{ig\Delta}\phi_{n-1,j \oplus 1}(\Delta)\right) \notag\\
=&
(\cos^{n}x+(-1)^j(-i)^n \sin^n x)
\notag\\
&\quad - e^{-ix}e^{2ix}(\cos^{n-1} x -(-1)^{1-j} (-i)^{n-1}\sin^{n-1}x)\notag\\
=&
\cos^{n-1}x(\cos x - e^{ix})
\notag\\
&\quad +(-1)^j(-i)^{n-1} \sin^{n-1} x ((-i)\sin x + e^{ix}).
\end{align}
Hence when $n$ is odd, $(-i)^{n-1} = i^{n-1}$ and the result follows.
\end{proof}

\begin{proof}[Proof of Lemma 
\ref{lem:qU-norm-sandwich}]
Using \eqref{Q-normalization}, we get
\begin{align}
  \<q_j| U_\Delta | j_L\> 
=
\frac{2}{g\sqrt n}  \<Q_j| U_\Delta | j_L\> .
\end{align} 
To simplify notation, let $x= g \Delta /2$.
Since $n$ is odd, we use 
\eqref{lem11-first} in Lemma \ref{lem:Q-U-sandwich} to get
\begin{align}
& |  \<q_j| U_\Delta | j_L\>  |\notag\\
=&
\sqrt n
\left|
 \phi_{n,j}(\Delta) -  e^{ig\Delta}\phi_{n-1,j \oplus 1}(\Delta) \right|
 \notag\\
 =&
\sqrt{n}\left|
-i  \cos^{n-1}x  \sin x  + i^{n-1} (-1)^j  \sin^{n-1}x  \cos x \right|
\notag\\
 =&
\sqrt{n}\sqrt{\sin^{2n-2}x  \cos^2 x +  \cos^{2n-2}x  \sin^2 x  }
\notag\\
 =&
\frac{\sqrt{n}}{2} |\sin 2x| \sqrt{\sin^{2n-4}x+  \cos^{2n-4}x }.
\end{align}
Hence the result follows.
\end{proof}
Furthermore, when $n=3$, for $j=0,1$, we have 
\begin{align}
   |\<q_j| U_\Delta | j_L\>  |^2
   = 3 \sin^2 x / 4.
\end{align}
In this case, we can see that 
Furthermore, when $n=3$, for $j=0,1$, we have 
\begin{align}
   |\<j_L| U_\Delta | j_L\>  |^2 +
   |\<q_j| U_\Delta | j_L\>  |^2 = 1.
\end{align}
This means that when $n=3$ and when there are no deletions, when we perform a projective measurement according to the POVM $\{\Pi, \Pi_1, I-\Pi-\Pi_1\}$, it is only possible to project onto the codespace and the space supported by $\Pi_1$.

\bibliography{ref}{}

@article{PhysRevA.110.062610,
  title = {Nonlocal multiqubit quantum gates via a driven cavity},
  author = {Jandura, Sven and Srivastava, Vineesha and Pecorari, Laura and Brennen, Gavin K. and Pupillo, Guido},
  journal = {Phys. Rev. A},
  volume = {110},
  issue = {6},
  pages = {062610},
  numpages = {17},
  year = {2024},
  month = {Dec},
  publisher = {American Physical Society},
  doi = {10.1103/PhysRevA.110.062610},
  url = {https://link.aps.org/doi/10.1103/PhysRevA.110.062610}
}

@Article{Wu2022,
author={Wu, Yue
and Kolkowitz, Shimon
and Puri, Shruti
and Thompson, Jeff D.},
title={Erasure conversion for fault-tolerant quantum computing in alkaline earth Rydberg atom arrays},
journal={Nature Communications},
year={2022},
month={Aug},
day={09},
volume={13},
number={1},
pages={4657},
issn={2041-1723},
doi={10.1038/s41467-022-32094-6},
url={https://doi.org/10.1038/s41467-022-32094-6}
}

@article{ouyang2019robust,
  author={Ouyang, Yingkai and Shettell, Nathan and Markham, Damian},
  journal={IEEE Transactions on Information Theory}, 
  title={Robust Quantum Metrology With Explicit Symmetric States}, 
  year={2022},
  volume={68},
  number={3},
  pages={1809-1821},
  doi={10.1109/TIT.2021.3132634}}

@article{ouyang2019permutation,
  title={Permutation-invariant constant-excitation quantum codes for amplitude damping},
  author={Ouyang, Yingkai and Chao, Rui},
  doi={10.1109/TIT.2019.2956142},
  journal={IEEE Transactions on Information Theory},
  volume={66},
  number={5},
  pages={2921--2933},
  year={2019},
  publisher={IEEE}
}

@article{OUYANG201743,
title = "Permutation-invariant qudit codes from polynomials",
journal = "Linear Algebra and its Applications",
volume = "532",
number = "",
pages = "43 - 59",
year = "2017",
note = "",
issn = "0024-3795",
doi = "http://dx.doi.org/10.1016/j.laa.2017.06.031",
url = "http://www.sciencedirect.com/science/article/pii/S0024379517303956",
author = "Yingkai Ouyang"
}

@article{ouyang2015permutation,
  title = {Permutation-invariant codes encoding more than one qubit},
  author = {Ouyang, Yingkai and Fitzsimons, Joseph},
  journal = {Physical Review A},
  volume = {93},
  issue = {4},
  pages = {042340},
  numpages = {4},
  year = {2016},
  month = {Apr},
  publisher = {American Physical Society},
  doi = {10.1103/PhysRevA.93.042340},
  url = {http://link.aps.org/doi/10.1103/PhysRevA.93.042340}
}

@article{ouyang2014permutation,
  author = {Yingkai Ouyang},
  title = {{P}ermutation-invariant quantum codes},
  journal = {Physical Review A},
  volume={90},
  number={6},
  pages={062317},
  year={2014},
  eprint = {1302.3247},
  doi = {10.1103/PhysRevA.90.062317},
  publisher={American Physical Society}
}

@article{PoR04,
author = {Pollatsek, Harriet and Ruskai, Mary Beth},
doi = {10.1016/j.laa.2004.06.014},
issn = {0024-3795},
journal = {Linear Algebra and its Applications},
number = {0},
pages = {255--288},
title = {{Permutationally invariant codes for quantum error correction}},
url = {http://www.sciencedirect.com/science/article/pii/S0024379504002903},
volume = {392},
year = {2004}
}

@article{Rus00,
author = {Ruskai, Mary Beth},
doi = {10.1103/PhysRevLett.85.194},
journal = {Physical Review Letters},
month = jul,
number = {1},
pages = {194--197},
publisher = {American Physical Society},
title = {{Pauli Exchange Errors in Quantum Computation}},
url = {http://link.aps.org/doi/10.1103/PhysRevLett.85.194},
volume = {85},
year = {2000}
}

@article{sidhu2019geometric,
  title={Geometric perspective on quantum parameter estimation},
  author={Sidhu, Jasminder S and Kok, Pieter},
  journal={AVS Quantum Science},
  volume={2},
  number={1},
  pages={014701},
  year={2020},
  publisher={American Vacuum Society},
doi={10.1116/1.5119961}
}

@article{kessler2014quantum,
  title={Quantum error correction for metrology},
  author={Kessler, Eric M and Lovchinsky, Igor and Sushkov, Alexander O and Lukin, Mikhail D},
  journal={Physical review letters},
  volume={112},
  number={15},
  pages={150802},
  year={2014},
  publisher={APS},
doi={10.1103/PhysRevLett.112.150802}
}

@article{arrad2014increasing,
  title={Increasing sensing resolution with error correction},
  author={Arrad, Gilad and Vinkler, Yuval and Aharonov, Dorit and Retzker, Alex},
  journal={Physical review letters},
  volume={112},
  number={15},
  pages={150801},
  year={2014},
  publisher={APS},
doi={10.1103/PhysRevLett.112.150801}
}

@article{dur2014improved,
  title={Improved quantum metrology using quantum error correction},
  author={D{\"u}r, W and Skotiniotis, M and Froewis, Florian and Kraus, B},
  journal={Physical Review Letters},
  volume={112},
  number={8},
  pages={080801},
  year={2014},
  publisher={APS},
doi={10.1103/PhysRevLett.112.080801}
}

@article{matsuzaki2017magnetic,
  title={Magnetic-field sensing with quantum error detection under the effect of energy relaxation},
  author={Matsuzaki, Yuichiro and Benjamin, Simon},
  journal={Physical Review A},
  volume={95},
  number={3},
  pages={032303},
  year={2017},
  publisher={APS},
doi={10.1103/PhysRevA.95.032303}
}

@article{unden2016quantum,
  title={Quantum metrology enhanced by repetitive quantum error correction},
  author={Unden, Thomas and Balasubramanian, Priya and Louzon, Daniel and Vinkler, Yuval and Plenio, Martin B and Markham, Matthew and Twitchen, Daniel and Stacey, Alastair and Lovchinsky, Igor and Sushkov, Alexander O and others},
  journal={Physical review letters},
  volume={116},
  number={23},
  pages={230502},
  year={2016},
  publisher={APS},
doi={10.1103/PhysRevLett.116.230502}
}

@article{layden2019ancilla,
  title={Ancilla-free quantum error correction codes for quantum metrology},
  author={Layden, David and Zhou, Sisi and Cappellaro, Paola and Jiang, Liang},
  journal={Physical review letters},
  volume={122},
  number={4},
  pages={040502},
  year={2019},
  publisher={APS},
doi={10.1103/PhysRevLett.122.040502}
}

@Article{Zhou2018,
author={Zhou, Sisi
and Zhang, Mengzhen
and Preskill, John
and Jiang, Liang},
title={Achieving the Heisenberg limit in quantum metrology using quantum error correction},
journal={Nature Communications},
year={2018},
volume={9},
number={1},
pages={78},
issn={2041-1723},
doi={10.1038/s41467-017-02510-3},
url={https://doi.org/10.1038/s41467-017-02510-3}
}

@inbook{LDU1,
author = {John Preskill},
title = {FAULT-TOLERANT QUANTUM COMPUTATION},
booktitle = {Introduction to Quantum Computation and Information},
chapter = {},
pages = {213-269},
doi = {10.1142/9789812385253_0008},
URL = {https://www.worldscientific.com/doi/abs/10.1142/9789812385253_0008}
}

@article{LDU2,
author = {Aliferis, Panos and Terhal, Barbara M.},
title = {Fault-tolerant quantum computation for local leakage faults},
year = {2007},
issue_date = {January 2007},
publisher = {Rinton Press, Incorporated},
address = {Paramus, NJ},
volume = {7},
number = {1},
issn = {1533-7146},
journal = {Quantum Info. Comput.},
month = jan,
pages = {139–156},
numpages = {18}
}

@article{LDU3,
  title = {Experimental deterministic correction of qubit loss},
  author = {Roman Stricker and Davide Vodola and Alexander Erhard and Lukas Postler and Michael Meth and Martin Ringbauer and Philipp Schindler and Thomas Monz and Markus Müller and Rainer Blatt},
  journal = {Nature},
  year = {2020},
  volume = {585},
  number = {7824},
  pages = {207--210},
  isbn = {1476-4687},
  doi = {10.1038/s41586-020-2667-0},
  url = {https://doi.org/10.1038/s41586-020-2667-0}
}

@article{LDU4,
  doi = {10.22331/q-2025-10-13-1884},
  url = {https://doi.org/10.22331/q-2025-10-13-1884},
  title = {Quantum {E}rror {C}orrection resilient against {A}tom {L}oss},
  author = {Perrin, Hugo and Jandura, Sven and Pupillo, Guido},
  journal = {{Quantum}},
  issn = {2521-327X},
  publisher = {{Verein zur F{\"{o}}rderung des Open Access Publizierens in den Quantenwissenschaften}},
  volume = {9},
  pages = {1884},
  month = oct,
  year = {2025}
}

@article{RydbergMicrowave,
  title = {Quantum Rabi Oscillations in Coherent and in Mesoscopic Cat Field States},
  author = {Assemat, F. and Grosso, D. and Signoles, A. and Facon, A. and Dotsenko, I. and Haroche, S. and Raimond, J. M. and Brune, M. and Gleyzes, S.},
  journal = {Phys. Rev. Lett.},
  volume = {123},
  issue = {14},
  pages = {143605},
  numpages = {5},
  year = {2019},
  month = {Oct},
  publisher = {American Physical Society},
  doi = {10.1103/PhysRevLett.123.143605},
  url = {https://link.aps.org/doi/10.1103/PhysRevLett.123.143605}
}

@misc{lancaster2025quantumsensingpresencepulse,
      title={Quantum sensing in the presence of pulse errors and qubit leakage}, 
      author={David M. Lancaster and Muhammad Ali Shahbaz and Hamed Goli Yousefabad and Sanway Chatterjee and Eegan Ram and Jonathan D. Weinstein},
      year={2025},
      eprint={2509.09874},
      archivePrefix={arXiv},
      primaryClass={quant-ph},
      url={https://arxiv.org/abs/2509.09874}, 
}

@article{
Neutralatomoptical,
author = {Brandon Grinkemeyer  and Elmer Guardado-Sanchez  and Ivana Dimitrova  and Danilo Shchepanovich  and G. Eirini Mandopoulou  and Johannes Borregaard  and Vladan Vuletić  and Mikhail D. Lukin },
title = {Error-detected quantum operations with neutral atoms mediated by an optical cavity},
journal = {Science},
volume = {387},
number = {6740},
pages = {1301-1305},
year = {2025},
doi = {10.1126/science.adr7075},
URL = {https://www.science.org/doi/abs/10.1126/science.adr7075},
eprint = {https://www.science.org/doi/pdf/10.1126/science.adr7075}}

@article{trappedionrev2025,
   author = "Foss-Feig, Michael and Pagano, Guido and Potter, Andrew C. and Yao, Norman Y.",
   title = "Progress in Trapped-Ion Quantum Simulation", 
   journal= "Annual Review of Condensed Matter Physics",
   year = "2025",
   volume = "16",
   number = "Volume 16, 2025",
   pages = "145-172",
   doi = "https://doi.org/10.1146/annurev-conmatphys-032822-045619",
   url = "https://www.annualreviews.org/content/journals/10.1146/annurev-conmatphys-032822-045619",
   publisher = "Annual Reviews",
   issn = "1947-5462",
   type = "Journal Article",
  }

@article{Srivastava_2026,
  title = {Entanglement-Enhanced Quantum Sensing via Optimal Global Control with Neutral Atoms in a Cavity},
  author = {Srivastava, Vineesha and Jandura, Sven and Brennen, Gavin K. and Pupillo, Guido},
  journal = {Phys. Rev. Lett.},
  volume = {136},
  issue = {6},
  pages = {060806},
  numpages = {7},
  year = {2026},
  month = {Feb},
  publisher = {American Physical Society},
  doi = {10.1103/k3bb-yfdv},
  url = {https://link.aps.org/doi/10.1103/k3bb-yfdv}
}

@article{JZL2024,
    author = {Huang, Jiahao and Zhuang, Min and Lee, Chaohong},
    title = {Entanglement-enhanced quantum metrology: From standard quantum limit to Heisenberg limit},
    journal = {Applied Physics Reviews},
    volume = {11},
    number = {3},
    pages = {031302},
    year = {2024},
    month = {07},
    issn = {1931-9401},
    doi = {10.1063/5.0204102},
    url = {https://doi.org/10.1063/5.0204102}
    }

@article{RevModPhys.89.035002,
  title = {Quantum sensing},
  author = {Degen, C. L. and Reinhard, F. and Cappellaro, P.},
  journal = {Rev. Mod. Phys.},
  volume = {89},
  issue = {3},
  pages = {035002},
  numpages = {39},
  year = {2017},
  month = {Jul},
  publisher = {American Physical Society},
  doi = {10.1103/RevModPhys.89.035002},
  url = {https://link.aps.org/doi/10.1103/RevModPhys.89.035002}
}

@article{PhysRevA.101.022321,
  title = {Universal gates for protected superconducting qubits using optimal control},
  author = {Abdelhafez, Mohamed and Baker, Brian and Gyenis, Andr\'as and Mundada, Pranav and Houck, Andrew A. and Schuster, David and Koch, Jens},
  journal = {Phys. Rev. A},
  volume = {101},
  issue = {2},
  pages = {022321},
  numpages = {13},
  year = {2020},
  month = {Feb},
  publisher = {American Physical Society},
  doi = {10.1103/PhysRevA.101.022321},
  url = {https://link.aps.org/doi/10.1103/PhysRevA.101.022321}
}

@article{zhang2025opticallyaccessiblehighfinessemillimeterwave,
  title = {Optically accessible high-finesse millimeter-wave resonator for cavity quantum electrodynamics with atom arrays},
  author = {Zhang, Tony and Wu, Michelle and Cohen, Sam R. and Xin, Lin and Das, Debadri and Multani, Kevin K.S. and Peard, Nolan and Valente-Feliciano, Anne-Marie and Welander, Paul B. and Safavi-Naeini, Amir H. and Nanni, Emilio A. and Schleier-Smith, Monika},
  journal = {Phys. Rev. Appl.},
  volume = {24},
  issue = {4},
  pages = {L041001},
  numpages = {6},
  year = {2025},
  month = {Oct},
  publisher = {American Physical Society},
  doi = {10.1103/4b8v-qdcj},
  url = {https://link.aps.org/doi/10.1103/4b8v-qdcj}
}

@article{
CavityExp,
author = {Brandon Grinkemeyer  and Elmer Guardado-Sanchez  and Ivana Dimitrova  and Danilo Shchepanovich  and G. Eirini Mandopoulou  and Johannes Borregaard  and Vladan Vuletić  and Mikhail D. Lukin },
title = {Error-detected quantum operations with neutral atoms mediated by an optical cavity},
journal = {Science},
volume = {387},
number = {6740},
pages = {1301-1305},
year = {2025},
doi = {10.1126/science.adr7075},
URL = {https://www.science.org/doi/abs/10.1126/science.adr7075},
eprint = {https://www.science.org/doi/pdf/10.1126/science.adr7075}}

@article{PhysRevResearch.7.L022072,
  title = {Global variational quantum circuits for arbitrary symmetric state preparation},
  author = {Bond, Liam J. and Davis, Matthew J. and Min\'a\ifmmode \check{r}\else \v{r}\fi{}, Ji\ifmmode \check{r}\else \v{r}\fi{}\'{\i} and Gerritsma, Rene and Brennen, Gavin K. and Safavi-Naini, Arghavan},
  journal = {Phys. Rev. Res.},
  volume = {7},
  issue = {2},
  pages = {L022072},
  numpages = {7},
  year = {2025},
  month = {Jun},
  publisher = {American Physical Society},
  doi = {10.1103/PhysRevResearch.7.L022072},
  url = {https://link.aps.org/doi/10.1103/PhysRevResearch.7.L022072}
}

@inproceedings{HagiwaraISIT2020,
  author    = {Manabu Hagiwara and
               Ayumu Nakayama},
  title     = {A Four-Qubits Code that is a Quantum Deletion Error-Correcting Code
               with the Optimal Length},
  booktitle = {{IEEE} International Symposium on Information Theory, {ISIT} 2020,
               Los Angeles, CA, USA, June 21-26, 2020},
  pages     = {1870--1874},
  publisher = {{IEEE}},
  year      = {2020},
  url       = {https://doi.org/10.1109/ISIT44484.2020.9174339},
  doi       = {10.1109/ISIT44484.2020.9174339}
}

@article{leahy2019quantum,
  title={Quantum insertion-deletion channels},
  author={Leahy, Janet and Touchette, Dave and Yao, Penghui},
  journal={arXiv preprint arXiv:1901.00984},
  year={2019}
}

@article{johnsson2020geometric,
  title={Geometric Pathway to Scalable Quantum Sensing},
  author={Johnsson, Mattias T and Mukty, Nabomita Roy and Burgarth, Daniel and Volz, Thomas and Brennen, Gavin K},
  journal={Physical Review Letters},
  volume={125},
  number={19},
  pages={190403},
  year={2020},
  publisher={APS},
doi={10.1103/PhysRevLett.125.190403}
}

@article{KnL97,
author = {Knill, Emanuel and Laflamme, Raymond},
doi = {10.1103/PhysRevA.55.900},
journal = {Physical Review A},
month = feb,
number = {2},
pages = {900--911},
publisher = {American Physical Society},
title = {{Theory of quantum error-correcting codes}},
url = {http://link.aps.org/doi/10.1103/PhysRevA.55.900},
volume = {55},
year = {1997}
}

@INPROCEEDINGS{shibayama2021permutation,
  author={Shibayama, Taro and Hagiwara, Manabu},
  booktitle={2021 IEEE International Symposium on Information Theory (ISIT)}, 
  title={Permutation-Invariant Quantum Codes for Deletion Errors}, 
  year={2021},
  volume={},
  number={},
  pages={1493-1498},
  doi={10.1109/ISIT45174.2021.9517870}}

@inproceedings{shibayama2021equivalence,
  title={The equivalence between correctability of deletions and insertions of separable states in quantum codes},
  author={Shibayama, Taro and Ouyang, Yingkai},
  booktitle={2021 IEEE Information Theory Workshop (ITW)},
  pages={1--6},
  year={2021},
  organization={IEEE},
doi={10.1109/ITW48936.2021.9611450}
}

@inproceedings{ouyang2021permutation,
  author={Ouyang, Yingkai},
  booktitle={2021 IEEE International Symposium on Information Theory (ISIT)}, 
  title={Permutation-invariant quantum coding for quantum deletion channels}, 
  year={2021},
  volume={},
  number={},
  pages={1499-1503},
  doi={10.1109/ISIT45174.2021.9518078}}

@article{toth2014quantum,
  title={Quantum metrology from a quantum information science perspective},
  author={T{\'o}th, G{\'e}za and Apellaniz, Iagoba},
  journal={Journal of Physics A: Mathematical and Theoretical},
  volume={47},
  number={42},
  pages={424006},
  year={2014},
  publisher={IOP Publishing},
doi = {10.1088/1751-8113/47/42/424006}
}

@article{gorecki2020optimal,
  title={Optimal probes and error-correction schemes in multi-parameter quantum metrology},
  author={G{\'o}recki, Wojciech and Zhou, Sisi and Jiang, Liang and Demkowicz-Dobrza{\'n}ski, Rafa{\l}},
  journal={Quantum},
  volume={4},
  pages={288},
  year={2020},
  publisher={Verein zur F{\"o}rderung des Open Access Publizierens in den Quantenwissenschaften},
doi={10.22331/q-2020-07-02-288}
}

@article{jordan2009permutational,
author = {Jordan, Stephen P.},
title = {Permutational Quantum Computing},
year = {2010},
issue_date = {May 2010},
publisher = {Rinton Press, Incorporated},
address = {Paramus, NJ},
volume = {10},
number = {5},
issn = {1533-7146},
journal = {Quantum Info. Comput.},
month = {may},
pages = {470–497},
numpages = {28},
doi={10.26421/QIC10.5-6-7}
}

@article{havlivcek2018quantum,
  title = {Quantum Schur Sampling Circuits can be Strongly Simulated},
  author = {{Havl\'{\i}\ifmmode \check{c}\else \v{c}\fi{}ek}, {Vojt\ifmmode \check{e}\else \v{e}\fi{}ch} and Strelchuk, Sergii},
  journal = {Phys. Rev. Lett.},
  volume = {121},
  issue = {6},
  pages = {060505},
  numpages = {5},
  year = {2018},
  month = {Aug},
  publisher = {American Physical Society},
  doi = {10.1103/PhysRevLett.121.060505},
  url = {https://link.aps.org/doi/10.1103/PhysRevLett.121.060505}
}

@article{sidhu2021tight,
  title = {Tight Bounds on the Simultaneous Estimation of Incompatible Parameters},
  author = {Sidhu, Jasminder S. and Ouyang, Yingkai and Campbell, Earl T. and Kok, Pieter},
  journal = {Phys. Rev. X},
  volume = {11},
  issue = {1},
  pages = {011028},
  numpages = {27},
  year = {2021},
  month = {Feb},
  publisher = {American Physical Society},
  doi = {10.1103/PhysRevX.11.011028},
  url = {https://link.aps.org/doi/10.1103/PhysRevX.11.011028}
}

@article{hayashi2023tight,
  title={Tight Cram{\'e}r-Rao type bounds for multiparameter quantum metrology through conic programming},
  author={Hayashi, Masahito and Ouyang, Yingkai},
  journal={Quantum},
  volume={7},
  pages={1094},
  year={2023},
  publisher={Verein zur F{\"o}rderung des Open Access Publizierens in den Quantenwissenschaften},
doi={10.22331/q-2023-08-29-1094}
}

@article{shettell2021practical,
  title={Practical limits of error correction for quantum metrology},
  author={Shettell, Nathan and Munro, William J and Markham, Damian and Nemoto, Kae},
  journal={New Journal of Physics},
  volume={23},
  number={4},
  pages={043038},
  year={2021},
  publisher={IOP Publishing},
doi={10.1088/1367-2630/abf533}
}

@article{zhou2020optimal,
  title={Optimal approximate quantum error correction for quantum metrology},
  author={Zhou, Sisi and Jiang, Liang},
  journal={Physical Review Research},
  volume={2},
  number={1},
  pages={013235},
  year={2020},
  publisher={APS},
  doi={10.1103/PhysRevResearch.2.013235}
}

@article{movassagh2020constructing,
   doi = {10.22331/q-2024-11-27-1541},
  url = {https://doi.org/10.22331/q-2024-11-27-1541},
  title = {Constructing quantum codes from any classical code and their embedding in ground space of local {H}amiltonians},
  author = {Movassagh, Ramis and Ouyang, Yingkai},
  journal = {{Quantum}},
  issn = {2521-327X},
  publisher = {{Verein zur F{\"{o}}rderung des Open Access Publizierens in den Quantenwissenschaften}},
  volume = {8},
  pages = {1541},
  month = nov,
  year = {2024}
}

@article{FT-quantum-sensing,
  title = {Fault-tolerant quantum metrology},
  author = {Kapourniotis, Theodoros and Datta, Animesh},
  journal = {Phys. Rev. A},
  volume = {100},
  issue = {2},
  pages = {022335},
  numpages = {15},
  year = {2019},
  month = {Aug},
  publisher = {American Physical Society},
  doi = {10.1103/PhysRevA.100.022335},
  url = {https://link.aps.org/doi/10.1103/PhysRevA.100.022335}
}

@article{hayashi2024finding,
  title={Finding the optimal probe state for multiparameter quantum metrology using conic programming},
  author={Hayashi, Masahito and Ouyang, Yingkai},
  journal={npj Quantum Information},
  volume={10},
  number={1},
  pages={111},
  year={2024},
doi={10.1038/s41534-024-00905-x},
  publisher={Nature Publishing Group UK London}
}

@article{wang2024dispersive,
  title={Dispersive nonreciprocity between a qubit and a cavity},
  author={Wang, Ying-Ying and Wang, Yu-Xin and van Geldern, Sean and Connolly, Thomas and Clerk, Aashish A and Wang, Chen},
  journal={Science Advances},
  volume={10},
  number={16},
  pages={eadj8796},
  year={2024},
doi={10.1126/sciadv.adj8796},
  publisher={American Association for the Advancement of Science}
}

@article{eickbusch2022fast,
  title={Fast universal control of an oscillator with weak dispersive coupling to a qubit},
  author={Eickbusch, Alec and Sivak, Volodymyr and Ding, Andy Z and Elder, Salvatore S and Jha, Shantanu R and Venkatraman, Jayameenakshi and Royer, Baptiste and Girvin, Steven M and Schoelkopf, Robert J and Devoret, Michel H},
  journal={Nature Physics},
  volume={18},
  number={12},
  pages={1464--1469},
  year={2022},
doi={10.1038/s41567-022-01776-9},
  publisher={Nature Publishing Group UK London}
}

@article{PhysRevLett.94.113601,
  title = {Nondestructive Rydberg Atom Counting with Mesoscopic Fields in a Cavity},
  author = {Maioli, P. and Meunier, T. and Gleyzes, S. and Auffeves, A. and Nogues, G. and Brune, M. and Raimond, J. M. and Haroche, S.},
  journal = {Phys. Rev. Lett.},
  volume = {94},
  issue = {11},
  pages = {113601},
  numpages = {4},
  year = {2005},
  month = {Mar},
  publisher = {American Physical Society},
  doi = {10.1103/PhysRevLett.94.113601},
  url = {https://link.aps.org/doi/10.1103/PhysRevLett.94.113601}
}

@article{PhysRevLett.130.143004,
  title = {Fast Single-Shot Imaging of Individual Ions via Homodyne Detection of Rydberg-Blockade-Induced Absorption},
  author = {Du, Jinjin and Vogt, Thibault and Li, Wenhui},
  journal = {Phys. Rev. Lett.},
  volume = {130},
  issue = {14},
  pages = {143004},
  numpages = {6},
  year = {2023},
  month = {Apr},
  publisher = {American Physical Society},
  doi = {10.1103/PhysRevLett.130.143004},
  url = {https://link.aps.org/doi/10.1103/PhysRevLett.130.143004}
}

@article{newpaper,
  title={A theory of quantum error correction for permutation-invariant codes},
  author={Yingkai Ouyang and Gavin K. Brennen},
  journal={arXiv preprint arXiv:2602.13638},
  doi={10.48550/arXiv.2602.13638},
  year={2025}
}

\end{document}